\documentclass[12pt]{article}

\usepackage[margin = 1in]{geometry}
\usepackage{graphicx}              
\usepackage{amsmath}               
\usepackage{amsfonts}              
\usepackage{amsthm}                
\usepackage{amssymb}
\usepackage{mathrsfs}
\usepackage{url}
\usepackage{hyperref}
\usepackage{tikz}
\usetikzlibrary{positioning,calc,arrows}
\usepackage{algorithmic}
\usepackage{subcaption}
\usepackage{enumitem}

\hypersetup{
	unicode=true,
	colorlinks=true,
	citecolor=blue!70!black,
	filecolor=black,
	linkcolor=red!70!black,
	urlcolor=blue,
	pdfstartview={FitH},
}

\theoremstyle{plain}
\newtheorem{theorem}{Theorem}
\newtheorem{lemma}[theorem]{Lemma}
\newtheorem{corollary}[theorem]{Corollary}
\newtheorem{proposition}[theorem]{Proposition}
\theoremstyle{definition}
\newtheorem{definition}[theorem]{Definition}
\newtheorem{conjecture}[theorem]{Conjecture}
\newtheorem*{problem}{Problem}
\newtheorem{example}[theorem]{Example}
\newtheorem*{remark}{Remark}
\newtheorem{note}[theorem]{Note}

\algsetup{linenodelimiter=.}

\newcommand{\ang}[1]{\langle#1\rangle}

\newcommand{\tildO}{\tilde{O}}

\newcommand{\Romnum}[1]{\uppercase\expandafter{\romannumeral #1}}

\DeclareMathOperator{\trace}{Tr} 
\DeclareMathOperator{\gal}{Gal} 
\DeclareMathOperator{\order}{ord} 
\DeclareMathOperator{\Res}{Res}
\DeclareMathOperator{\Aut}{Aut}

\DeclareMathOperator{\loglog}{loglog}

\def\Q{\ensuremath{\mathbb{Q}}}

\def\Z{\ensuremath{\mathbb{Z}}}
\def\F{\ensuremath{\mathbb{F}}}
\def\MM{\ensuremath{\mathsf{M}}}
\def\euler{\ensuremath{\varphi}}

\makeatletter
\newcounter{algorithm}
\renewcommand{\thealgorithm}{\arabic{algorithm}}
\def\algorithm{\@ifnextchar[{\@algorithma}{\@algorithmb}}
\def\@algorithma[#1]{%
	\refstepcounter{algorithm}
	\trivlist
	\leftmargin\z@
	\itemindent\z@
	\labelsep\z@
	\item[\parbox{\columnwidth}{%
		\hrule
		\hrule
		\noindent\strut\textbf{Algorithm \thealgorithm} #1
		\hrule
	}]\hfil\vskip0em%
}
\def\@algorithmb{\@algorithma[]}

\makeatother

\title{Computing isomorphisms and embeddings of finite fields}
\author{Ludovic Brieulle, Luca De Feo, Javad Doliskani,\\ Jean-Pierre
  Flori and \'Eric Schost}

\begin{document}

\maketitle
\begin{abstract}
  Let $\F_q$ be a finite field. %
  Given two irreducible polynomials $f,g$ over $\F_q$, with $\deg f$
  dividing $\deg g$, the finite field embedding problem asks to
  compute an explicit description of a field embedding of
  $\F_q[X]/f(X)$ into $\F_q[Y]/g(Y)$. %
  When $\deg f = \deg g$, this is also known as the isomorphism
  problem.

  This problem, a special instance of polynomial factorization, plays
  a central role in computer algebra software. %
  We review previous algorithms, due to Lenstra, Allombert, Rains, and
  Narayanan, and propose improvements and generalizations. %
  Our detailed complexity analysis shows that our newly proposed
  variants are at least as efficient as previously known algorithms,
  and in many cases significantly better.

  We also implement most of the presented algorithms, compare them
  with the state of the art computer algebra software, and make the
  code available as open source. %
  Our experiments show that our new variants consistently outperform
  available software.
\end{abstract}

{
  \hypersetup{linkcolor=black}
  \setcounter{tocdepth}{2}
  \tableofcontents
}

\section{Introduction}
\label{sec:introduction}

Let $q$ be a prime power and let $\F_q$ be a field with $q$
elements. Let $f$ and $g$ be irreducible polynomials over $\F_q$, with
$\deg f$ dividing $\deg g$. Define $k=\F_q[X]/f(X)$ and
$K=\F_q[Y]/g(Y)$; then, there is an embedding $\phi:k\hookrightarrow
K$, unique up to $\F_q$-automorphisms of $k$. The goal of this paper
is to describe algorithms to efficiently represent and evaluate one
such embedding.

All the algorithms we are aware of split the embedding problem in two
sub-problems:
\begin{enumerate}
\item Determine elements $\alpha\in k$ and $\beta\in K$ such that
  $k=\F_q(\alpha)$, and such that there exists an
  embedding $\phi$ mapping $\alpha\mapsto\beta$. We refer to this
  problem as the \emph{embedding description problem}.
  It is easily seen that $\alpha$ and $\beta$ describe an embedding
  if and only if they share the same minimal polynomial.
\item Given elements $\alpha$ and $\beta$ as above, given $\gamma\in
  k$ and $\delta\in K$, solve the following problems:
  \begin{itemize}
  \item Compute $\phi(\gamma)\in K$.
  \item Test if $\delta\in\phi(k)$.
  \item If $\delta\in\phi(k)$, compute $\phi^{-1}(\delta)\in k$.
  \end{itemize}
  We refer collectively to these problems as the \emph{embedding
    evaluation problem}.
\end{enumerate}

\paragraph{Motivation, previous work}
The first to get interested in this problem was H.~Lenstra: in his
seminal paper~\cite{LenstraJr91} he shows that it can be solved in
deterministic polynomial time, by using a representation for finite
fields that he calls \emph{explicit data}.\footnote{Technically,
  Lenstra only proved his theorem in the case where $k$ and $K$ are
  isomorphic; however, the generalization to the embedding problem
  poses no difficulties.} %
In practice, the embedding problem arises naturally when designing a
computer algebra system: as soon as a system is capable of
representing arbitrary finite fields, it is natural to ask it to
compute the morphisms between them. %
Ultimately, by representing effectively the lattice of finite fields
with inclusions, the user is given access to the algebraic closure of
$\F_q$. %
The first system to implement a general embedding algorithm was
Magma~\cite{MAGMA}. %
As detailed by its developers~\cite{bosma+cannon+steel97}, it used a
much simpler approach than Lenstra's algorithm, entirely based on
polynomial factorization and linear algebra. %
Lenstra's algorithm was later revived by
Allombert~\cite{Allombert02,Allombert02-rev} who modified some key
steps in order to make it practical; his implementation has since been
part of the PARI/GP system~\cite{Pari}.

Meanwhile, a distinct family of algorithms for the embedding problem
was started by Pinch~\cite{Pinch}, and later improved by
Rains~\cite{rains2008}. %
These algorithms, based on principles radically different from
Lenstra's, are intrinsically probabilistic. %
Although their worst-case complexity is no better than that of
Allombert's algorithm, they are potentially much more efficient on a
large set of parameters. %
This potential was understood by Magma's developers, who implemented
Rains' algorithm in Magma v2.14.%
\footnote{As a matter of fact, Rains' algorithm was never published;
  the only publicly available source for it is in Magma's source code
  (file \texttt{package/Ring/FldFin/embed.m}, since v2.14).}

With the exception of Lenstra's work, the aforementioned papers were
mostly concerned with the practical aspects of the embedding
problem. %
While it was generally understood that computing embeddings is an
easier problem than general polynomial factoring, no results on its
complexity more precise than Lenstra's had appeared until recently. %
A few months before the present paper was finalized, Narayanan
published a novel generalization of Allombert's
algorithm~\cite{narayanan2016fast}, based on elliptic curve
computations, and showed that its (randomized) complexity is at most
quadratic.
Narayanan's
  generalization relies on the fact that Artin--Schreier and Kummer
  theories are special cases of a more general situation:
  as already emphasized by Couveignes and Lercier~\cite{CL08}, whereas
  the former theory acts on the additive group of a finite field, and
  the latter on its multiplicative group, they can be extended to more
  general commutative algebraic groups, in particular to elliptic
  curves.

\paragraph{Our contribution}
This work aims to be, in large part, a complete review of all known
algorithms for the embedding problem; we analyze in detail the cost of 
existing algorithms and introduce several new variants. %
After laying out the foundations in the next section, we start with
algorithms for the embedding description problem. %

Section~\ref{sec:kummer} describes the family of algorithms based
(more or less loosely) on Lenstra's work; we call these
\emph{Kummer-type} algorithms. %
In doing so, we pay a particular attention to Allombert's algorithm:
to our knowledge, this is the first detailed and complete complexity
analysis of this algorithm and its variants. %
Thanks to our work on asymptotic complexity, we were able to devise
improvements to the original variants of Allombert that largely
outperform them both in theory and practice. %
One notable omission in this section is Narayanan's algorithm, which
is, in our opinion, mostly of theoretical rather than practical
interest. %
We present instead in Subsection~\ref{sec:fast-algor-large} a simpler
algorithm with essentially the same complexity.

In Section~\ref{sec:rains-algorithm} we describe Rains' algorithm. %
Rains' original preprint~\cite{rains2008} went unpublished, thus we
give here a complete description and analysis of his algorithm, for
reference. %
We also give new variants of Rains' algorithm of lesser
interest in Appendix~\ref{app:rains-vars}.

Then, in Section~\ref{sec:rains-elliptic} we present a generalization
of Rains' algorithm using elliptic curves. %
The possibility of this algorithm was hinted at by Rains, but never
fully developed; we show that it is indeed possible to use
\emph{elliptic periods} to solve the embedding description problem,
and that the resulting algorithm behaves well both in theory and in
practice. %
While working out the correctness proof of the elliptic variant of
Rains' algorithm, we encounter an unexpected difficulty: whereas roots
of unity enjoy Galois properties that guarantee the success of Rains'
original algorithm, points of elliptic curves fail to provide the
same. %
Heuristically, the failure probability of the elliptic variant is
extremely small, however we are not able to prove it formally. %
Our experimental searches even seem to suggest that the failure
probability might be, surprisingly, zero. %
We state this as a conjecture on elliptic periods (see
Conjecture~\ref{conj:ellperiods}); our findings and supporting
evidence are summarized in Appendix~\ref{app:ellprdsdata}.

Section~\ref{sec:selection} does a global comparison of all the
algorithms presented previously. %
In particular, Rains' algorithm and variants require a non-trivial
search for parameters, which we discuss thoroughly. %
Then we present an algorithm to select the best performing embedding
description algorithm from a practical point of view. %
This theoretical study is complemented by the experimental
Section~\ref{sec:experimental-results}, where we compare our
implementations of all the algorithms; our source code is made
available through the Git repository
\url{https://github.com/defeo/ffisom} for replication and further
scrutiny.

Our review could not be complete without a presentation of all the
embedding evaluation algorithms, which we undertake in
Section~\ref{sec:eval}. %
Given that the algorithms of this section are much more classical and
well understood, we only give a theoretical presentation, with no
experimental support. %

In conclusion, we hope that our review will constitute a reference
guide for researchers and engineers interested in implementing
embeddings of finite fields in a computer algebra system.

\paragraph{Acknowledgments}
We would like to thank Eric M. Rains for sharing his preprint with
us. %
We also thank Bill Allombert, Christian Berghoff, Jean-Marc
Couveignes, Reynald Lercier, and Benjamin Smith for fruitful
discussions.

\section{Preliminaries}
\label{sec:preliminaries}

\subsection{Fundamental algorithms and complexity}
\label{sec:fundamentalgo}
We review the fundamental building blocks that constitute the
algorithms presented next.  We are going to measure all complexities
in number of operations $+$, $\times$, $\div$ in $\F_q$, unless
explicitly stated otherwise. Most of the algorithms we present are
randomized; we use the \emph{big-Oh} notation $O(\,)$ to express {\em average}
asymptotic complexity, and we will make it clear when this complexity
depends on heuristics. We also occasionally use the notation
$\tildO(\,)$ to neglect logarithmic factors in the parameters.

We let $\MM(m)$ be a function such that polynomials in $\F_q[X]$ of
degree less than $m$ can be multiplied in $\MM(m)$ operations in
$\F_q$, under the assumptions of~\cite[Ch.~8.3]{vzGG}, together with
the slightly stronger one, that $\MM(mn)$ is in $O(m^{1+\varepsilon}
\MM(n))$ for all $\varepsilon > 0$. Using FFT multiplication, one can
take $\MM(m)\in O(m\log(m) \loglog(m))$~\cite{cantor+kaltofen91}.

We denote by $\omega$ the \emph{exponent of linear algebra}, i.e.\ a
constant such that $m\times m$ matrices with coefficients in any field
$k$ can be multiplied using $O(m^\omega)$ additions and
multiplications in $k$. One can take $\omega < 2.38$, the best result
to date being in~\cite{LeGall14}; on the other hand, we also suppose
that $\omega > 2$.

The algorithms presented in the next sections perform computations in
ring extensions of finite fields. Some of these extensions also happen
to be finite fields. As customary, if $k$ is a finite field and $\xi$
is some element of an algebraic extension of $k$, we will write
$k[\xi]$ for the ring generated by $\xi$. To avoid confusion, when the
extension generated by $\xi$ is a finite field, we will write instead
$k(\xi)$.

Some algorithms will operate in a polynomial ring $k[Z]$, where $k$ is
a field extension of $\F_q$; some other algorithms will operate in
$k[Z]/h(Z)$, where $h$ is a monic polynomial in $k[Z]$. We review the
basic operations in these rings. We assume that $k$ is represented as
a quotient ring $\F_q[X]/f(X)$, with $m=\deg f$, and we let $s=\deg h$
in the complexity estimates.

Multiplying and dividing polynomials of degree at most $s$ in $k[Z]$
is done in $O(\MM(sm))$ operations in $\F_q$, using Kronecker's
substitution~\cite{moenck76,kaltofen87,vzGG,vzgathen+shoup92,harvey09}.
Multiplication in $k[Z]/h(Z)$ is also done in $O(\MM(sm))$ operations using the
technique in~\cite{pascal+schost06}. By the same techniques, gcds of
degree $m$ polynomials in $k[Z]$ and inverses in $k[Z]/h(Z)$ are
computed in $O(\MM(sm)\log(sm))$ operations.

Given polynomials $e,g,h \in k[Z]$ of degree at most $s$, modular
composition is the problem of computing $e(g) \bmod h$. An upper bound
on the algebraic complexity of modular composition is obtained by the
Brent--Kung algorithm~\cite{brent+kung}; under our assumptions on the
respective costs of polynomial and matrix multiplication, its cost is
$O(s^{(\omega+1)/2}\MM(m))$ operations in $\F_q$
(so if $k=\F_q$, this is $O(s^{(\omega+1)/2})$). In the binary RAM
complexity model, the Kedlaya--Umans algorithm~\cite{KeUm11} and its
extension in~\cite{PoSc13a} yield an algorithm with essentially linear
complexity in $s$, $m$ and $\log(q)$. Unfortunately, making these
algorithms competitive in practice is challenging; we are not aware of
any implementation of them that would outperform Brent and Kung's
algorithm. 

\begin{note}\label{note:multimc}
If we have several modular compositions of the form $e_1(g) \bmod
h,\dots,e_t(g) \bmod h$ to compute, we can slightly improve the obvious
bound $O(ts^{(\omega+1)/2})$ (we discuss here $k=\F_q$, so $m=1$). If
$t=O(s)$, using~\cite[Lemma~4]{kaltofen+shoup98}, this can be done in
time $O(t^{(\omega-1)/2}s^{(\omega+1)/2})$. If $t=\Omega(s)$, this can
be done in $O(t s^{\omega-1})$ operations, by computing $1,g,\dots,g^{s-1}$
modulo $f$, and doing a matrix product in size $s \times s$ by
$s \times t$.
\end{note}

\paragraph{Frobenius evaluation} Consider an $\F_q$-algebra $Q$,
and an element $\alpha$ in $Q$. Given integers $c,d$, we will have to
compute expressions of the form
\[
\sigma_d= \alpha^{q^d}, \quad \tau_d = \sum_{i=0}^{d-1} \alpha^{q^{ci}}, \quad
\mu_d=\alpha^{\lfloor q^d/c\rfloor} \enspace .
\]
A direct binary powering approach
would yield a complexity of, e.g., $O(d\log(q))$ multiplications in $Q$
for the first expression.

To do better, we use a recursive approach that goes back
to~\cite{von1992computing}, with further ideas borrowed from~\cite{shoup94,kaltofen+shoup97}.
For $i \ge 1$, define integers $A_i, B_i$ as follows
\begin{equation*}
  q^i = A_ic + B_i, \quad 0\le B_i < c.
\end{equation*}
Then, we have the relations
$$\sigma_{i+j}=\sigma_j^{q^i}, \quad \tau_{i+j}=\tau_i +
\tau_j^{q^{ic}}, \quad \mu_{i+j}=\mu_j^{q^i}
\mu_i^{B_j}\alpha^{\lfloor B_iB_j / c \rfloor}.$$ Since we are
interested in $\sigma_d,\tau_d$ and $\mu_d$, using an addition chain
for $d$, we are left to perform $O(\log(d))$ steps as above.

To perform these operations, we will make a heavy use of a technique
originating in~\cite{von1992computing}. In its simplest form, it
amounts to the following: if $Q=\F_q[X]/f(X)$, for some polynomial $f$
in $\F_q[X]$, and $\beta$ is in $Q$, we can compute $\beta^{q}$ by
means of the modular composition $\beta(\xi)$, where $\xi = x^{q}$ and
$x$ is the image of $X$ modulo $f$.

In the following proposition, we discuss
versions of this idea for various kinds of algebras $Q$, and how they
allow us to compute the expressions $\sigma_d, \tau_d, \mu_d$ defined
above.

\begin{proposition}
\label{prop:trace-like}
Let $f \in \F_q[X]$ be a polynomial of degree $m$,
and define the $\F_q$-algebra $Q = \F_q[X]/f(X)$.
Let $h \in Q[Z]$ be a polynomial of degree $s$,
and define the $Q$-algebra $S =  Q[Z]/h(Z)$.
Finally, whenever $h \in \F_q[Z]$, define
the $\F_q$-algebra $Q' = \F_q[Z]/h(Z)$.

Denote by $T_Q, T_S, T_{Q'}$ the cost, in terms of $\F_q$-operations,
of one modular composition in $Q, S, Q'$ respectively. %
Also denote by $T_{Q,t}\le tT_Q$ (resp. $T_{S,t}, T_{Q',t}$) the cost of $t$
modular compositions sharing the same polynomial (see
Note~\ref{note:multimc}).

Then the expressions
\[
\sigma_d= \alpha^{q^d}, \quad \tau_d = \sum_{i=0}^{d-1} \alpha^{q^{ci}}, \quad
\mu_d=\alpha^{\lfloor q^d/c\rfloor}
\]
can be computed using the following number of operations:
\begin{enumerate}[label=\textbf{Case~\theenumi.},leftmargin=*, align=left]
\item $\alpha \in Q$:
\begin{itemize}
\item
$\sigma_d$: $O(\MM(m) \log(q) + T_Q \log(d))$,
\item
$\tau_d$: $O(\MM(m) \log(q) + T_Q \log(d c))$,
\item
$\mu_d$: $O(\MM(m) \log(q) + (T_Q + \MM(m)\log(c))\log(d))$;
\end{itemize}
\item $\alpha \in Q$ with $f | X^r - 1$:
\begin{itemize}
\item
$\sigma_d$: $O(\MM(m) \log(q) + \MM(r) \log(d))$,
\item
$\tau_d$: $O(\MM(m) \log(q) + \MM(r) \log(d c))$,
\item
$\mu_d$: $O(\MM(m) \log(q) + (\MM(r) + \MM(m)\log(c))\log(d))$;
\end{itemize}
\item $\alpha \in S$:
\begin{itemize}
\item
$\sigma_d$: $O(\MM(ms)\log(q) + (T_{Q,s}+T_S)\log(d))$,
\item
$\tau_d$: $O(\MM(ms)\log(q) + (T_{Q,s}+T_S)\log(d c))$,
\item
$\mu_d$: $ O(\MM(ms)\log(q) + (T_{Q,s}+T_S + \MM(ms)\log(c))\log(d))$;
\end{itemize}
\item $\alpha \in S$ with $h \in \F_q[Z]$:
\begin{itemize}
\item
$\sigma_d$: $O((\MM(m)+\MM(s))\log(q) + (T_{Q,s}+T_{Q',m})\log(d)$,
\item
$\tau_d$: $O((\MM(m)+\MM(s))\log(q) + (T_{Q,s}+T_{Q',m})\log(d c))$,
\item
$\mu_d$: $O((\MM(m)+\MM(s))\log(q) + (T_{Q,s}+T_{Q',m} + (m\MM(s)+s\MM(m))\log(c))\log(d))$;
\end{itemize}
\item $\alpha \in S$ with $h | X^r - a$ for $a \in Q$:
\begin{itemize}
\item
$\sigma_d$: $O(\MM(m)\log(q) + (T_{Q,s}+ \MM(mr))\log(d))$,
\item
$\tau_d$: $O(\MM(m)\log(q) + (T_{Q,s} + \MM(mr))\log(d c))$,
\item
$\mu_d$: $O(\MM(m)\log(q) + (T_{Q,s} + \MM(mr) + \MM(m)\log(c))\log(d))$.
\end{itemize}
\end{enumerate}
\end{proposition}

\begin{proof}
The complexity estimates mostly rely on the complexity of modular composition.
\begin{enumerate}[label=\textbf{Case~\theenumi.},leftmargin=*, align=left]
\item
  We let $x$ be the image of $X$ in $Q$, and we
  start by computing $x^q$, using $O(\MM(m)\log(q))$ operations in
  $\F_q$.

  For $i \ge 0$, given $\xi_i = x^{q^i}$ and $\beta$ in $Q$, we can
  compute $\beta^{q^i}$ as $\beta^{q^i}=\beta(\xi_i)$, using
  $T_Q=O(m^{(\omega+1)/2})$ operations; in particular, this allows us to
  compute $\xi_{i+j}$ from the knowledge of $\xi_i$ and $\xi_j$. Given
  an addition chain for $d$, we thus compute all corresponding
  $\xi_i$'s, and we deduce the $\sigma_i$'s similarly, since
  $\sigma_{i+j}=\sigma_j(\xi_j)$. Altogether, starting from
  $\xi_1=x^q$, this gives us $\sigma_d$ for $O(T_Q\log(d))$
  further operations in $\F_q$.

  The same holds for $\tau_d$, with a cost in $O(T_Q\log(cd))$,
  since we have to compute $\xi_c$ first;
  and for $\mu_d$, with a cost in
  $O((T_Q + \MM(m)\log(c))\log(d))$ operations,
  as the formula for $\mu_{i+j}$ shows that
  we can obtain it by means of a modular composition
  (to compute $\mu_j^{q^i}=\mu_j(\xi_i)$), together with
  two exponentiations of indices less than~$c$.

  The costs for computing $\sigma_d, \tau_d, \mu_d$ follow immediately.

\item
  When $f$ divides $X^r-1$, we obtain $\beta(\xi_i)$ by computing
  $\beta(X^{q^i \bmod r}) \bmod (X^r-1)$ first, and then reducing
  modulo $f(X)$.  %
  Thus, the cost of one modular composition is $T_Q=O(\MM(r))$, and
  the total cost is obtained by replacing this value in the estimates
  for the previous case.

\item
  We let $x$ and $z$ be the respective images of $X$ in $Q$ and $Z$ in $S$,
  and as a first step, we compute $z^q$ (and $x^q$, unless $f$ is as
  in Case~2 above), in $O(\MM(ms)\log(q))$ operations.
  
  In order to compute
  the quantities $\sigma_d,\tau_d,\mu_d$, we apply the same strategy
  as above; the key factor for complexity is thus the cost of
  computing $\beta^{q^i}$, for $\beta$ in $S$, given $\zeta_i =
  z^{q^i}$ and $\xi_i=x^{q^i}$ (as we did in Case~1, we apply this
  procedure to our input element $\alpha$, as well as to $\zeta_i$
  itself, and $\xi_i$, in order to be able to continue the
  calculation).

  To do so, we use an algorithm by Kaltofen and Shoup~\cite{kaltofen+shoup97},
  which boils down to writing $\beta=\sum_{j=0}^{s-1} c_j(x) z^j$, so
  that $\beta^{q^i} = \sum_{j=0}^{s-1} c_j(x)^{q^i} \zeta_i^j.$ The
  $s$ coefficients $c_j(x)^{q^i}$ are computed by applying the
  previous algorithms in $Q$ to $s$ inputs. This takes time at most
  $sT_Q$, but as pointed out in Note~\ref{note:multimc}, improvements 
  are possible if we base our algorithm on modular composition;
  we thus denote the cost $T_{Q,s}$.

  Then, we do a modular composition in $S$ to evaluate the result at
  $\zeta_i$; this latter step takes $T_S=O(s^{(\omega+1)/2} \MM(m))$
  operations in~$\F_q$.

\item
  The cost for computing $z^q$ is
  $O(\MM(s)\log(q))$ and that for computing $x^q$ is $O(\MM(m)\log(q))$.
  In the last step, the cost $T_S$ of modular
  composition in $S$ is now that of $m$ modular compositions in degree
  $s$ (with the same argument), as detailed in
  Note~\ref{note:multimc}, that we denote $T_{Q',m}$.
  Similarly, the cost of multiplication in $S$ can be reduced
  from $O(\MM(ms))$ to $O(s\MM(m) + m\MM(s))$ operations.

\item
  We start by computing $x^q$, using $O(\MM(m)\log(q))$ operations in
  $\F_q$.

  For $\beta$ as above, suppose that we have already computed all
  coefficients $d_j(x)=c_j(x)^{q^i}$ in $O(T_{Q,s})$ operations;
  we now have to compute $\beta^{q^i} = \sum_{j=0}^{s-1} d_j(x) \zeta_i^j$.

  We first do the calculation modulo $Z^r-a$ rather than modulo $h$;
  that is, we compute $\sum_{j=0}^{s-1} d_j(x) z_i^j$ where
  $z_i=z^{q^i}$.
  Because $z^r=a$, we
  have $z_i = a_i z^{q^i \bmod r}$, with $a_i =
  a^{\lfloor q^i/r\rfloor}$. If we assume that $a_i$ is
  known, we can compute $\sum_{j=0}^{s-1} d_j(x) z_i^j$ using
  Horner's method, in time $O(s \MM(m))$, and we reduce this result
  modulo $h$, for the cost $O(\MM(mr))$ of a Euclidean division in
  degree $r$ in $Q[Z]$. 

  In order to continue the calculation for all indices in our addition
  chain, we must thus compute the corresponding $a_i$'s as well,
  just like the $\mu_i$'s;
  this takes $O(T_{Q}+\MM(m)\log(r))$ operations.

  Since the first stage of the algorithm took $O(T_{Q,s})$ operations,
  we can take $T_S = O(\MM(mr))$
  for computing $\beta^{q^i}$.
  
  To initiate the procedure, the algorithm also needs to compute
  $a_1=a^{\lfloor q/r\rfloor}$, using $O(\log(q))$
  multiplications in $Q$ for a cost $O(\MM(m)\log(q))$. \qedhere
\end{enumerate}
\end{proof}

\paragraph{Computing subfields}
With $k = \F_q[X]/f(X)$ and $\deg f=m$ as above, we are given a
divisor $r$ of $m$, and we want to construct an intermediate extension
$\F_q \subset L \subset k$ of degree $r$ over $\F_q$. More precisely,
we want to compute a monic irreducible polynomial $g \in \F_q[X]$ of
degree $r$, and a polynomial $h \in \F_q[X]$ such that
$x \mapsto h(x)\bmod f$ defines an embedding
$L = \F_q[X] / g(X) \hookrightarrow k$. We proceed as follows.

Let $\alpha\in k$ be a random element. Then $\alpha$ has a minimal polynomial of degree $m$ 
over $\F_q$ with high probability. In other words, one needs $O(1)$ such random elements to find 
one with degree $m$ minimal polynomial. Now, the trace
\begin{equation}
	\label{equ:trace-simple}
	\trace_{k/L}(\alpha) = \alpha + \alpha^{q^r} + \cdots + \alpha^{q^{m - r}}
\end{equation}
has a minimal polynomial of degree $r$ over $\F_q$ with high
probability as well. This means we can compute, after $O(1)$ random
trials, the desired polynomials $\beta = \trace_{k/L}(\alpha)$,
its minimal polynomial
$g$, and $h$ the polynomial of degree less than $m$ representing $\beta$.

\begin{proposition}
	\label{prop:subfield}
	Let $\F_q \subset k$ be a finite extension of degree $m$, and
        let $r$ be a divisor of $m$.  Computing an intermediate field
        $\F_q \subset L \subset k$ with $[L:\F_q]=r$ takes an expected
        $O(m^{(\omega+1)/2}\log(m) + \MM(m)\log(q))$ operations in $\F_q$.  Once
        $L$ is computed, any element $\gamma\in L$ can be lifted to
        its image in $k$ using $O(m^{(\omega+1)/2})$ operations.
\end{proposition}
\begin{proof}
  Computing the minimal polynomial of an element in $k$ takes
  $O(m^{(\omega+1)/2})$ operations in $\F_q$, see~\cite{shoup93}. %
  The trace in Eq.~\eqref{equ:trace-simple} is computed as the
  expression $\tau_m$ of the previous paragraph (with $c=r$ and
  $d=m/r$), at a cost of $O(m^{(\omega+1)/2}\log(m)+\MM(m)\log(q))$
  operations in $\F_q$.
  
  Finally, given an element $\gamma\in L$, its image in $k$ is
  computed by evaluating $h(\gamma)$, where $h$ is the polynomial
  representation of $\trace_{k/L}(\alpha)$. This can be done by a
  modular composition at cost $O(m^{(\omega+1)/2})$.
\end{proof}

\paragraph{Root finding in cyclotomic extensions}
Given a field $k = \F_q[X]/f(X)$ of degree $m$ as above, we will need
to factor some special polynomials in $k[Z]$: we are interested in
finding one factor of a polynomial that splits into factors of the
same, known, degree. This problem is known as \emph{equal degree
  factorization} (EDF), and the best generic algorithm for it is the
Cantor--Zassenhaus method~\cite{cantor1981,von1992computing}, which
runs in $O(\MM(sm)(dm\log(q) + \log(sm)))$ operations in
$\F_q$~\cite[Th.~14.9]{vzGG}, where $s$ is the degree of the
polynomial to factor, and $d$ is the degree of the factors.

More efficient variants of the Cantor--Zassenhaus method are known for
special cases. When the degree $s$ of the polynomial is small compared
to the extension degree $m$, Kaltofen and Shoup~\cite{kaltofen+shoup97} 
give an efficient algorithm which is as follows.

\begin{algorithm}[Kaltofen--Shoup EDF for extension fields]
	\label{alg:ks}
	\begin{algorithmic}[1]
		\REQUIRE A polynomial $h$  with irreducible factors of degree $d$ over $k=\F_q[X]/f(X)$.
		\ENSURE An irreducible factor of $h$ over $k$.
		\STATE If $\deg h = d$ return $h$.
		\STATE Take a random polynomial $a_0\in k[Z]$ of degree less than $\deg h$,
		\STATE\label{alg:ks-pseudotrace} Compute $\displaystyle a_1 
		\leftarrow \sum_{i=0}^{md-1} a_0^{q^i} \mod h$,
		\IF{$q$ is an even power $q=2^e$}
		\STATE\label{alg:ks:even} Compute $\displaystyle a_2 \leftarrow 
		\sum_{i=0}^{e-1} a_1^{2^i}\mod h$
		\ELSE
		\STATE\label{alg:ks:odd} Compute $a_2 \leftarrow a_1^{(q-1)/2}\mod h$
		\ENDIF
		\STATE\label{alg:ks:gcd} Compute $h_0\leftarrow\gcd(a_2,h)$ and 
		$h_1\leftarrow\gcd(a_2-1,h)$ and $h_{-1}\leftarrow h/(h_0h_1)$,
		\STATE Apply recursively to the smallest non-constant polynomial among 
		$h_0,h_1,h_{-1}$.
	\end{algorithmic}
\end{algorithm}

We refer the reader to the original paper~\cite{kaltofen+shoup97} for
the correctness of the Kaltofen--Shoup algorithm. We are mainly
interested here in its application to root extraction in cyclotomic
extensions. Let $r$ be a prime power and let $f$ be an irreducible
factor of the $r$-th cyclotomic polynomial $\Phi_r$, with $s = \deg
f$. Denote $\F_q[X] / f(X)$ by $\F_q(\zeta)$, where $\zeta$ is the
image of $X$ in the quotient. Given an $r$-th power $\alpha \in
\F_q(\zeta)$ we want to compute an $r$-th root $\alpha^{1 / r}$, or
equivalently a linear factor of $Z^r - \alpha$ over $\F_q(\zeta)$.

We propose two different algorithms; one of them is quadratic in $r$,
whereas the other one has a runtime that depends on $r$ and $s$,
and will perform better for small values of $s$.

\begin{proposition}
  \label{prop:root-fpz}
  Let $r$ be a prime power and let $\zeta$ be a primitive $r$-th root of unity; let also $s = 
  [\F_q(\zeta): \F_q]$. One can take $r$-th roots in 
  $\F_q(\zeta)$ using either 
  $$O(\MM(s)\log(q) + rs^{\omega-1}\log(r)\log(s) + \MM(rs)\log(s)\log(r))$$
  or
  $$O(\MM(s)\log(q) + r\MM(r)\log(s) + \MM(rs)\log(s)\log(r))$$ operations in $\F_q$. 
\end{proposition}
\begin{proof}
We use Algorithm~\ref{alg:ks} with $k=\F_q(\zeta)$, to get a linear
factor of the polynomial $Z^r-\alpha$, so that $d=1$ (note that
$Z^r-\alpha$ splits into linear factors in $k[Z]$). We discuss
Step~\ref{alg:ks-pseudotrace}, which is the dominant step. Let $f \in
\F_q[X]$ be the defining polynomial of $\F_q(\zeta)$ and let $h$ be a
factor of $Z^r-\alpha$ of degree $n$. 

We are in Case~5 of our discussion on Frobenius evaluation, and we
want to compute a trace-like expression of the form $\tau_s$. As per
that discussion, two algorithms are available to do Frobenius
evaluation in $k$ (one of them uses modular composition, the other the
fact that $f$ divides $X^r-1$). Because $s \le r$, we deduce that
$a_1$ can be computed in either
$$O(\MM(s)\log(q) + rs^{\omega-1}\log(s) + \MM(rs)\log(s))$$
or
$$O(\MM(s)\log(q) + n\MM(r)\log(s) + \MM(rs)\log(s))$$ operations in
$\F_q$, where the first term accounts for computing $\alpha^{\lfloor
  q/r \rfloor}$ (so we need only compute it once). The depth of the
recursion in Algorithm~\ref{alg:ks} is $\log(r)$, and the degree $n$
is halved each time, so we obtain the
desired result.
\end{proof}

\paragraph{Root finding in some extensions of cyclotomic extensions}
Let $r = v^d$, where $v \ne p$ is a prime and $d$ is a positive
integer and let $s$ be the order of $q$ in $\Z / v\Z$.  We 
assume that $d \ge 2$, since this will be the case whenever we want to
apply the following.

Consider an extension $\F_q \subset k=\F_q[X]/f(X)$ of degree $r$, and
let $\F_q(\zeta)$ and $k(\zeta)$ be extensions of degree $s$ over
$\F_q$ and $k$ respectively, defined by an irreducible factor of the
$v$-th cyclotomic polynomial over $\F_q$. In this paragraph, we
discuss the cost of computing a $v$-th root in $k(\zeta)$, by adapting
the root extraction algorithm given in~\cite{doliskanischost2011}.

Following~\cite[Algorithm~3]{doliskanischost2011}, one reduces the
root extraction in $k(\zeta)$ to a root extraction in $\F_q(\zeta)$;
note that \cite[Algorithm~3]{doliskanischost2011} reduces the root
extraction to the smallest possible extension of $\F_p$, but
projecting to $\F_q(\zeta)$ is more convenient here. The critical
computation in this algorithm is a trace-like computation performing
the reduction.

\begin{algorithm}
[$v$-th root in $k(\zeta)$]
\label{algorithm:new}
\begin{algorithmic}[1]
\REQUIRE $a \in k(\zeta)^v$
\ENSURE a $v$-th root of $a$
\REPEAT
\STATE choose a random $c \in k(\zeta)$
\STATE $a'\leftarrow ac^v$
\STATE $\lambda \leftarrow {a'}^{(q^s-1)/v}$
\STATE $b \leftarrow 1 + \lambda + \lambda^{1+q^{s}} + \cdots + \lambda^{1+q^{s}+\cdots+q^{(r-2)s}}$
\UNTIL {$b \ne 0$}
\STATE $\beta \leftarrow (a'b^v)^{1/v}$ in $\F_q(\zeta)$
\RETURN $\beta b^{-1}c^{-1}$
\end{algorithmic}
\end{algorithm}

One multiplication in $k(\zeta)$ amounts to doing $r$ multiplications
modulo a degree $s$ factor of $\Phi_v$, and $s$ multiplications modulo
$f$; since $s \le r$, this takes $O(s \MM(r))$ operations in $\F_q$.
The computation of $\lambda = {a'}^{(q^s - 1)/v}= {a'}^{\lfloor
  q^s/v\rfloor}$ can then be done as explained in our discussion on
Frobenius evaluation (Case~4). The cost of each modular composition
 is $O(s^{(\omega-1)/2} r^{(\omega+1)/2})$, for a 
total of $O(s^{(\omega-1)/2} r^{(\omega+1)/2}\log(s) +
s\MM(r)\log(q))$ operations in $\F_q$.

The trace-like computation of $1 + \lambda + \lambda^{1+q^{s}} +
\cdots + \lambda^{1+q^{s}+\cdots+q^{(r-2)s}}$ can be done as follows.
Let $x$ be the image of $X$ in $k=\F_q[X]/f(X)$.  To compute $x^{q^s}$
we first compute $x^q$ using $O(\MM(r)\log(q))$ operations in $\F_q$,
and then do $\log(s)$ modular compositions in $k$.  To compute
$\lambda^{q^s}$, note that an element $\lambda \in k(\zeta)$ can be
written as $\lambda = \lambda_0(x) + \lambda_1(x) \zeta + \cdots +
\lambda_{s - 1}(x) \zeta^{s - 1}$ and that $\zeta^{q^s} = \zeta$.
Therefore for any $i$,
\[
\lambda^{q^{i s}} = \sum_{j = 0}^{s - 1} \lambda_j(x^{q^{i s}}) \left(\zeta^{q^{i s}}\right)^j = \sum_{j = 0}^{s - 1} \lambda_j(x^{q^{i s}}) \zeta^j.
\]
In particular, given $x^{q^{i s}}$, $\lambda^{q^{is}}$ can be computed using $O(s^{(\omega-1)/2}r^{(\omega+1)/2})$ operations in $\F_q$,
and~\cite[Algorithm~2]{doliskanischost2011} can be applied in a direct
way, with a cost of $O(s^{(\omega-1)/2}r^{(\omega+1)/2}\log(r) + \MM(r)\log(q))$ operations in $\F_q$.

The root extraction in $\F_q(\zeta)$ is done as in the previous
paragraph, and have a negligible cost, since we assumed that $s \le v
\le \sqrt{r}$. Therefore, we arrive at the following result.

\begin{proposition}\label{prop:root_high_degree_extension}
With $k$, $\zeta$ and $v$ as above, one can extract $v$-th roots in
$k(\zeta)$ using an expected
$O(s^{(\omega-1)/2}r^{(\omega+1)/2}\log(r) + s\MM(r)\log(q))$
operations in $\F_q$.
\end{proposition}

\subsection{The Embedding Description problem}

We are finally ready to address the problem of describing the embedding of
$k=\F_q[X]/f(X)$ in $K=\F_q[Y]/g(Y)$; throughout the paper we let $m=\deg f$ and
$n=\deg g$, so that $m|n$. The \emph{embedding description problem}
asks to find two elements $\alpha\in k$ and $\beta\in K$ such that
$\alpha\mapsto\beta$ for some field embedding $\phi:k\to K$. This is
equivalent to $\alpha$ and $\beta$ having the same minimal polynomial.

The most obvious way to solve this problem is to take the class of $X$
in $k=\F_q[X]/f(X)$ for $\alpha$, and a root of $f$ in $K$ for
$\beta$. Since $f$ splits completely in $K$, we can apply Algorithm
\ref{alg:ks} for the special case $d = 1$. Using our discussion on the
cost of Frobenius evaluation (precisely, Case~4), we obtain an upper
bound of $O\bigl((nm^{(\omega+1)/2} + \MM(m)n^{(\omega+1)/2} +
m\MM(n)\log(q))\log(m)\bigr)$ expected operations in $\F_q$ for the
problem. We remark that this complexity is strictly larger than
$\tildO(m^2)$.

For a more specialized approach, we note that it is enough to solve
the following problem: let $r$ be a prime power such that $r|m$ and
$\gcd(r,m/r)=1$, find $\alpha_r\in k$ and $\beta_r\in K$ such that
$\alpha_r$ and $\beta_r$ have the same minimal polynomial, of degree $r$.

Indeed, once such $\alpha_r$ and $\beta_r$ are known for every primary
factor $r$ of $m$, possible solutions to the embedding problem are
\begin{equation*}
  \alpha = \prod_{\substack{r|m,\\\gcd(r,m/r)=1}}\alpha_r,\qquad
  \beta = \prod_{\substack{r|m,\\\gcd(r,m/r)=1}}\beta_r,
\end{equation*}
or
\begin{equation*}
  \alpha = \sum_{\substack{r|m,\\\gcd(r,m/r)=1}}\alpha_r,\qquad
  \beta = \sum_{\substack{r|m,\\\gcd(r,m/r)=1}}\beta_r.
\end{equation*}

Moreover, to treat the general embedding description problem,
it is sufficient to treat the case where $[k:\F_q]=[K:\F_q]=r$.
Indeed, we can reduce to this situation by applying
Proposition~\ref{prop:subfield}, at an additional cost of
$O(n^{(\omega+1)/2}\log(n) + \MM(n)\log(q))$ for each primary factor~$r$.
Therefore, to simplify the exposition, we focus on algorithms
solving the following problem.
\begin{problem}
\label{prob:embedding}
Let $r$ be a prime power and $k, K$ a pair of extensions of $\F_q$
of degree $r$.
Describe an isomorphism between $k$ and $K$.
\end{problem}
Note that although some algorithms are restricted to this situation,
especially those presented in Section~\ref{sec:kummer},
some of them could still be readily applied to a more general situation,
especially those from Sections~\ref{sec:rains-algorithm}
and~\ref{sec:rains-elliptic}.

All algorithms presented next are going to rely on one common
principle: construct an element in $k$ (and in $K$) such that its
minimal polynomial (or, equivalently, its orbit under the absolute
Galois group of $\F_q$) is uniquely (or \emph{almost} uniquely)
defined.

\section{Kummer-type algorithms}
\label{sec:kummer}

In this section, we review what we call \emph{Kummer-type} approaches to the embedding problem for prime power degree extensions. 
We briefly review the works of Lenstra~\cite{LenstraJr91},
and Allombert~\cite{Allombert02,Allombert02-rev}, then 
we give variants of these algorithms with significantly lower complexities.
As stated above, we let $k, K$ be degree $r$ extensions of $\F_q$,
where $r$ is a prime power,
and we let $p$ be their characteristic.
We give our fast versions of the algorithms for two separate cases: the case $p \nmid r$
is treated in Section~\ref{sec:fast-kummer}, the case $r = p^d$, where $d$ 
is a positive integer, is treated in Section~\ref{sec:fast-artin-schreier}.
Finally, in Section~\ref{sec:fast-algor-large} we give a variant of the case
$p \nmid r$ better suited for the case where $r$ is a high-degree prime power.

In~\cite{LenstraJr91}, Lenstra proves that 
given two finite fields of the same size, there exists a deterministic polynomial time algorithm 
that finds an isomorphism between them.
The focus of the paper is on theoretical computational complexity;
in particular, it avoids using randomized subroutines, such as polynomial
factorization. 
In~\cite{Allombert02,Allombert02-rev}, Allombert gives a similar approach with more focus on practical efficiency.
In contrast to Lenstra's, his algorithm relies on polynomial factorization, thus it is
polynomial time Las Vegas.
Even though neither of the two algorithms is given a detailed complexity analysis, both rely
on solving linear systems, thus a rough analysis yields an estimate of $O(r^{\omega})$ operations
in $\F_q$ in both cases.

The idea of Lenstra's algorithm is as follows.  Assume that $r$ is
prime, and let $\F_q[\zeta]$ denote the ring extension $\F_q[Z] /
\Phi_r(Z)$ where $\Phi_r$ is the $r$-th cyclotomic polynomial.  Let
$\tau$ be a non $r$-adic residue of $\F_q[\zeta]$, and let
$\F_q[\zeta][\theta]$ denote the quotient $\F_q[\zeta][Y]/(Y^r -
\tau)$ such that $\theta=\tau^{1/r}$ is the residue class of $Y$.
Lenstra shows that $\F_q[\zeta][\theta]$ is isomorphic to $k[\zeta]$
as a ring (Lenstra actually goes the other way around and constructs
$\tau$ from $\theta$ as $\tau = \theta^r$ whereas $\theta$ itself
comes from a normal basis of $k$ computed using linear algebra.  In
Lenstra's terminology, $\theta$ and $\tau=\theta^r$ are generators of
the Teichm\"uller subgroups of $k[\zeta]$ and $\F_q[\zeta]$ and
solutions to Hilbert's theorem 90).

Furthermore, the algorithm constructs $\theta_1, \theta_2$, and
$\tau_1, \tau_2$ in such a way that an integer $j > 0$ can be found
such that
\[
\begin{array}{lrll}
\psi: & \F_q[\zeta][\theta_1] & \rightarrow & \F_q[\zeta][\theta_2] \\
& \theta_1 & \mapsto & \theta_2^j
\end{array}
\]
is an isomorphism of rings.
Finally, denoting by $\Delta$ the automorphism group of $k[\zeta]$
over $k$, an embedding $k \hookrightarrow K$ is obtained by
restricting the above isomorphism $\psi$ to the fixed field
$k[\zeta]^\Delta$.
To summarize, the algorithm is made of three steps:
\begin{itemize}
\item Construct elements $\theta_1\in k[\zeta]$ and $\theta_2\in K[\zeta]$;
\item Letting $\tau_i=\theta_i^r$, find the integer $j$ such that
  $\tau_1=\tau_2^j$ by a discrete logarithm computation in
  $\F_q[\zeta]$;
\item Compute $\alpha\in k$ and $\beta\in K$ as some functions of
  $\theta_1,\theta_2^j$ invariant under $\Delta$.
\end{itemize}
The algorithm is readily generalized to prime powers $r$ by iterating
this procedure.

Allombert's algorithms differ from Lenstra's in two key steps, both
resorting to polynomial factorization.
First, he computes an irreducible factor $h$ of the cyclotomic
polynomial $\Phi_r$ of degree $s$,
and so constructs a field extension $\F_q(\zeta)$ as $\F_q[Z]/h(Z)$.
Then he defines $k[\zeta]=k[Z]/h(Z)$ and $K[\zeta]=K[Z]/h(Z)$
(note that these are not fields if $r$ is not prime), and constructs
$\theta_1\in k[\zeta]$ and $\theta_2\in K[\zeta]$ in a way equivalent
to Lenstra's using linear algebra. 
At this point, rather than computing a discrete
logarithm, Allombert points out that there exists a $c\in\F_q(\zeta)$
such that $\theta_1\mapsto c\theta_2$ defines an isomorphism,
and that such value can be
computed as the $r$-th root of $\theta_1^r/\theta_2^r$.
Finally, by making the automorphism group of $k[\zeta]$ over $k$ act
on $\theta_1$ and $\theta_2$, he obtains an embedding $k \hookrightarrow K$.%

\subsection{Allombert's algorithm}
\label{sec:fast-kummer}

In this section, we analyze the complexity of Allombert's original
algorithm~\cite{Allombert02}, that of its revised version~\cite{Allombert02-rev},
and we present new variants with the best known asymptotic complexities.
The main difference with respect to the 
versions presented in~\cite{Allombert02,Allombert02-rev} is in the way
we compute $\theta_1, \theta_2$, which are solutions to Hilbert's theorem 90
as will become clear below.
Whereas Allombert resorts to linear algebra, we rely instead on evaluation
formulas that have a high probability of yielding a solution.
Recently, Narayanan~\cite[Sec.~3]{narayanan2016fast} independently described 
a variant which is similar to our Proposition~\ref{prop:xitheta}
in the special case $s=1$.

\subsubsection{General strategy}
Let $k=\F_q[X]/f(X)$ where $f$ has degree $r$, a prime power, and let $x$ be the image of $X$ in $k$.
Let $h(Z)$ be an 
irreducible factor of the $r$-th cyclotomic polynomial over $\F_q$. Then $h$ has degree $s$ where 
$s$ is the order of $q$ in the multiplicative group $(\Z/r\Z)^\times$. We form the field extension
$\F_q(\zeta) \cong \F_q[Z] / h(Z)$ and the ring extension $k[\zeta] = k[Z] / h(Z) \cong k \otimes
\F_q(\zeta)$ where $\zeta$ is the image of $Z$ in the quotients. The action of the Galois group
$\gal(k / \F_q)$ can be extended to $k[\zeta]$ by
\[
\left.
\begin{array}{llll}
\sigma: & k[\zeta] & \rightarrow & k[\zeta] \\
& x \otimes \zeta & \mapsto & x^q \otimes \zeta
\end{array}
\right.
.
\]
Allombert shows (see~\cite[Prop.~3.2]{Allombert02}) that $\sigma$ is
an automorphism of $\F_q(\zeta)$-algebras, and that its fixed set is
isomorphic to $\F_q(\zeta)$.
The same can be done for the ring $K[\zeta]$.
Let us restate the algorithm for clarity.

\begin{algorithm}[Allombert's algorithm]
	\begin{algorithmic}[1]
		\REQUIRE Field extensions $k, K$ of $\F_q$ of degree $r$.
		\ENSURE The description of a field embedding $k\to K$.
		\STATE Factor the $r$-th cyclotomic polynomial and make the extensions $\F_q(\zeta), 
		k[\zeta], K[\zeta]$;
		\STATE Find $\theta_1 \in k[\zeta]$ such that $\sigma(\theta_1) = \zeta\theta_1$;
		\STATE Find $\theta_2 \in K[\zeta]$ such that $\sigma(\theta_2) = \zeta\theta_2$;
		\STATE Compute an $r$-th root $c$ of $\theta_1^r / \theta_2^r$ in $\F_q(\zeta)$;
		\STATE Let $\alpha, \beta$ be the constant terms of $\theta_1, c\theta_2$ respectively;
		\RETURN The field embedding defined by $\alpha\mapsto\beta$.
	\end{algorithmic}
        \label{alog:allombert}
\end{algorithm}


The cyclotomic polynomial $\Phi_r$ is factored over $\F_q$
using~\cite[Theorem~9]{shoup94}, and $r$-th root extraction in
$\F_q(\zeta)$ is done using Proposition~\ref{prop:root-fpz}, so we are
left with the problem of finding $\theta_1$ (and $\theta_2$), that is,
instances of Hilbert's theorem 90.  

We now show how to do it in the extension $k[\zeta]/\F_q(\zeta)$, the
case of $K[\zeta]$ being analogous. We review approaches due to
Allombert, that rely on linear algebra, and propose new algorithms
that rely on evaluation formulas and ultimately polynomial
arithmetic. Note that all these variants can be directly applied to
any extension degree $r$ as long as $p \nmid r$, and do not require $r$
to be a prime power.  Nevertheless, in practice, it is more efficient
to perform computations for each primary factor independently and glue
the results together in the end.

If $A$ is a polynomial with coefficients in $\F_q(\zeta)$, we will
denote by $\hat{A}$ the morphism $A(\sigma)$ of the algebra
$k[\zeta]$; note that the usual property of \emph{$q$-polynomials}
holds: $\widehat{AB} = \hat{A}\circ\hat{B}$.

\subsubsection{Algorithms relying on linear algebra}
\label{sec:algor-rely-line}

As some algorithmic details were omitted in Allombert's publications,
and no precise complexity analysis was performed, we extracted the
details from PARI/GP source code~\cite{Pari} and perform the complexity analysis
here.  We also propose another variant, using an algorithm by Paterson
and Stockmeyer.

\paragraph{Allombert's original algorithm}
A direct solution to Hilbert's theorem 90 is to find a non-zero
$\theta\in k[\zeta]$ such that $\widehat{(S-\zeta)}(\theta)=0$.

The original version of Allombert's algorithm~\cite{Allombert02} does
precisely this, by computing the matrix of the Frobenius automorphism
$\sigma$ of $k/\F_q$ using $O(\MM(r) \log(q) + r \MM(r))$ operations
in $\F_q$ and then an eigenvalue of $\sigma$ for $\zeta$ over
$\F_q(\zeta)$ using linear algebra, at a cost of $O((rs)^\omega)$
operations in $\F_q$. This gives a total cost of $O(s \MM(r) \log(q) +
(rs)^{\omega})$ operations in~$\F_q$.

\paragraph{Allombert's revised algorithm}
Allombert's revision of his own algorithm~\cite{Allombert02-rev} uses
the factorization
\begin{equation}
  \label{eq:allomb-b}
  h(S)=(S-\zeta) b(S).
\end{equation}
If we set $h(S)=S^s+\sum_{i=0}^{s-1}h_iS^i$, we can
explicitly write $b$ as
\begin{equation}
  \label{eq:hprime}
  b(S)=\sum_{i=0}^{s-1}b_i(S)\zeta^i,\quad
  \text{where}\quad
  \left\{\begin{aligned}
    b_{s-1}(S) &= 1,\\
    b_{i-1}(S) &= b_i(S) S + h_i.
  \end{aligned}\right.
\end{equation}
Indeed, Horner's rule shows that $b_{-1}(S)=h(S)$, and by direct
calculation we find that $(S-\zeta)\cdot b(S) = b_{-1}(S)$.

We get a solution to Hilbert's theorem 90 by evaluating
$b(S)=h(S)/(S-\zeta)$ on an element in the kernel of $\hat{h}$ over
$k$, linear algebra now taking place over $\F_q$ rather than
$\F_q(\zeta)$. The details on the computation of $\hat{h}$ were
extracted from PARI/GP source code and yield the following complexity.

\begin{proposition}
  Using Allombert's revised algorithm, a solution $\theta$ to
  Hilbert's theorem 90 can be computed in $O(\MM(r) \log(q) + s r
  \MM(r) + r^\omega)$ operations in $\F_q$.
\end{proposition}

\begin{proof}
As in Allombert's original algorithm, one first
computes the matrix of $\sigma$ over $k$ at a cost of
$O(\MM(r) \log(q) + r \MM(r))$ operations in $\F_q$.

To get the matrix of $\hat{h}$ over $k$, one first computes the powers
$x^{q^i}$ for $0 \leq i \leq s$ using the matrix of $\sigma$, at a
cost of $O(s r^2)$ operations in $\F_q$.  From them, one can
iteratively compute the powers $x^{j q^i}$ for $2 \leq j \leq r$ for a
total cost of $O(s r \MM(r))$ operations in $\F_q$, and iteratively
compute the matrix of $\hat{h}$ for an additional total cost of $O(s
r^2)$ operations in $\F_q$, accounting for the scalar multiplications
by the coefficients of $h$.  The total cost is therefore dominated by
$O(s r \MM(r))$ operations in $\F_q$.

Given the matrix of $\hat{h}$ over $k$, computing an element in its
kernel costs $O(r^\omega)$ operations in $\F_q$.  The final evaluation
of $\hat{b}$ is done using Eq.~\eqref{eq:hprime} and the matrix of
$\sigma$ for Frobenius computations, for a cost of $O(s r^2)$ operations
in $\F_q$.
\end{proof}

\paragraph{Using the Paterson--Stockmeyer algorithm}
Given the matrix $M_\sigma$ of $\sigma$, there is a natural way of evaluating
$\hat{h}$ at a reduced cost: the Paterson--Stockmeyer
algorithm~\cite{paterson_stockmeyer} computes the matrix of $\hat h$
and $h(M_\sigma)$, using $O(\sqrt{s} r^\omega)$ operations in
$\F_q$. The evaluations of $\sigma$ that take a total of $O(s r^2)$
operations in $\F_q$ can be done directly using modular
exponentiations, for a total of $O(s \MM(r) \log(q))$.

\begin{proposition}
Using the Paterson--Stockmeyer algorithm and modular exponentiations,
a solution $\theta$ to Hilbert's theorem 90 can be computed in
$O(s \MM(r) \log(q) + \sqrt{s} r^\omega)$ operations in~$\F_q$.
\end{proposition}

Although this complexity is not as good as the ones we will obtain
next, this variant performs reasonably well in practice, as discussed
in Section~\ref{sec:experimental-results}.

\subsubsection{Algorithms relying on polynomial arithmetic}
\label{sec:algor-rely-polyn}

It is immediate to see that the minimal polynomial of $\sigma$ over $k[\zeta]$ is
$S^r-1$; by direct calculation, we verify that it
factors as
\begin{equation}
  \label{eq:theta90}
  S^r-1 = (S-\zeta)\cdot\Theta(S) =
  (S-\zeta)\sum_{i=0}^{r-1} \zeta^{-i-1}S^i. 
\end{equation}
Hence, we can set
\begin{equation}
  \label{eq:thetaa}
  \theta_a = \hat\Theta(a)
  =  a\otimes\zeta^{-1} + \sigma(a)\otimes\zeta^{-2} + \cdots + \sigma^{r-1}(a)\otimes\zeta^{-r}
\end{equation}
for some $a\in k$ chosen at random. %
Because of Eq.~\eqref{eq:theta90}, $\theta_a$ is a solution as long as
it is non-zero. %
This is reminiscent of Lenstra's algorithm~\cite[Th.~5.2]{LenstraJr91}.

To ensure the existence of $a$ such that $\theta_a\ne0$, we only need
to prove that $k$ is not entirely contained in $\ker\hat\Theta$. %
But the maps $\sigma^i$ restricted to $k$ are all distinct, thus
Artin's theorem on character independence (see~\cite[Ch~VI, Theorem~4.1]{lang})
shows that they are linearly 
independent, and therefore $\hat\Theta$ is not identically zero on $k$.
In practice, we take $a \in k$ at random until
$\theta_a\ne0$. Since the map $\hat\Theta$ is
$\F_q$-linear and non-zero, it has rank at least 1, thus a
random $\theta_a$ is zero with probability less than $1 / q$. Therefore, we only 
need $O(1)$ trials to find $\theta_1$ (and $\theta_2$).

Using the polynomial $b(S)$ introduced in Eq.~\eqref{eq:allomb-b},
and defining
$g(S)=(S^r-1)/h(S)$, we can rewrite Eq.~\eqref{eq:theta90} as
\begin{equation}
  \label{eq:theta90bis}
  \Theta(S) = b(S) \cdot g(S).
\end{equation}
Then, the morphism $\hat\Theta$ can be evaluated
as $\hat{b}\circ\hat{g}$, the advantage being that $g$ has
coefficients in $\F_q$, rather than in $\F_q(\zeta)$: we set $\tau_a = \hat{g}(a)$ for some $a\in k$ chosen at random
and compute $\theta_a = \hat{b}(\tau_a)$ using Eq.~\eqref{eq:hprime},
yielding a solution to Hilbert's theorem 90 as soon as $\tau_a \neq 0$.
As before, $O(1)$ trials are enough to get $\theta_a \neq 0$.

We now give three variations on the above algorithm to compute 
a candidate solution $\theta_a$ more efficiently.
Which algorithm has the best asymptotic complexity depends on the
value of $s$ with respect to $r$; we arrange them by increasing $s$.

\paragraph{First solution: divide-and-conquer recursion.}
We use a recursive algorithm similar to the computation of trace-like
functions in Proposition~\ref{prop:trace-like}, to directly evaluate $\theta_a$
using Eq.~\eqref{eq:thetaa}.  Let $\xi_1=x^q$ and
$\theta_{a,1}=a\zeta^{-1}$, and set the following recursive relations:
\begin{equation}
\label{eq:theta-recursive}
\xi_j = 
\begin{cases}
\sigma^{j/2}(\xi_{j / 2}) & j \text{ even,} \\
\sigma(\xi_{j - 1}) & j \text{ odd,}
\end{cases} \quad
\theta_{a, j} = 
\begin{cases}
\theta_{a, j / 2} + \zeta^{-j / 2}\sigma^{j / 2}(\theta_{a, j / 2})& j \text{ even,} \\
(a + \sigma(\theta_{a, j - 1}))\zeta^{-1} & j \text{ odd.}
\end{cases}
\end{equation}
Then $\theta_a=\theta_{a,r}$.

\begin{proposition}
  \label{prop:xitheta}
  Given $a\in k$, the value $\theta_a$ in Eq.~\eqref{eq:thetaa} can be
  computed using \[
O(s^{(\omega-1)/2}r^{(\omega+1)/2}\log(r)+\MM(r)\log(q))
\]
  operations in $\F_q$.
\end{proposition}
\begin{proof}
  The value $\xi_1$ is computed by binary powering using
  $O(\MM(r)\log(q))$ operations, while the value $\theta_{a,1}$ is
  deduced from the polynomial $h$ using $O(rs)$ operations.
  
  To compute the recursive formulas in Eq.~\eqref{eq:theta-recursive}
  we use the same technique as in Proposition~\ref{prop:trace-like}:
  given $b \in k[\zeta]$, the value $\sigma^j(b)$ is computed as the
  modular composition of the polynomial $b(x,z)$ with the polynomial
  $\xi_j(x)$ in the first argument. %
  Each modular composition in $k[\zeta]$ is done using $s$ 
  modular compositions in $k$, at a cost of $O(s^{(\omega-1)/2}r^{(\omega+1)/2})$ operations (see Note~\ref{note:multimc}). %
  Multiplications by $\zeta^{-j}$ are done by seeing the elements of
  $k[\zeta]$ as polynomials in $x$ over $\F_q(\zeta)$, thus performing
  $r$ multiplications modulo $h$, at a cost of $O(r\MM(s))$
  operations. %
  Given that the total depth of the recursion is $O(\log(r))$, we
  obtain the stated bound.
\end{proof}

\paragraph{Second solution: automorphism evaluation.}
We use Eq.~\eqref{eq:theta90bis} and Eq.~\eqref{eq:hprime}
to compute $\theta_a$ as
$\theta_a = \hat{b}\circ \hat{g} (a)$.
\begin{proposition}
  \label{prop:ks-theta}
  Given $a\in k$, the value $\theta_a$ in Eq.~\eqref{eq:thetaa} can
  be computed using
  \[
  O(r^{(\omega^2-4\omega-1)/(\omega-5)}+(s+r^{2/(5-\omega)})\MM(r)\log(q))
  \]
  operations in $\F_q$.
\end{proposition}
\begin{proof}
  We proceed in two steps. We first compute $\hat{g}(a)$ using the
  \emph{automorphism evaluation} algorithm of Kaltofen and
  Shoup~\cite[Algorithm~AE]{kaltofen+shoup98}, at a cost of
  $O(r^{(\omega+1)/2 + (3-\omega)\lvert\beta-1/2\rvert}+
  r^{(\omega+1)/2 + (1-\beta)(\omega-1)/2}+r^\beta\MM(r)\log(q))$, for
  any $0\le\beta\le1$. Choosing $\beta=2/(5-\omega)$ minimizes the
  overall runtime, giving the exponents reported above.
  
  We then use Eq.~\eqref{eq:hprime} to compute
  $\theta_a=\sum_{i=0}^{s-1}a_i\otimes\zeta^i$, where
  $a_{s-1}=\hat{g}(a)$, and $a_{i-1}=\sigma(a_i)+h_i\hat{g}(a)$. 
  The cost of this computation is dominated by the evaluations of
  $\sigma$, which take $O(\MM(r)\log(q))$ operations each, thus
  contributing $O(s\MM(r)\log(q))$ total operations.
\end{proof}

\paragraph{Third solution: multipoint evaluation.}
Finally, we can compute all the values
$\sigma(a),\dots,\sigma^{r-1}(a)$ directly, write $\theta_a$ as a
polynomial in $x$ and $\zeta$ of degree $r-1$ in both variables, and
reduce modulo $h$ for each power $x^i$.

\begin{proposition}
  \label{prop:iter-frob-theta}
  Given $a\in k$, the value $\theta_a$ in Eq.~\eqref{eq:thetaa} can
  be computed using 
  \[
O(\MM(r^2)\log(r) + \MM(r)\log(q))
\] operations in
  $\F_q$.
\end{proposition}
\begin{proof}
  The values $\sigma(a),\dots,\sigma^{r-1}(a)$ can be computed by
  binary powering using $O(r\MM(r)\log(q))$. %
  We can do slightly better using the iterated Frobenius technique of
  von~zur~Gathen and Shoup~\cite[Algorithm~3.1]{von1992computing} (see
  also~\cite[Ch.~14.7]{vzGG}), which costs of
  $O(\MM(r^2)\log(r) + \MM(r)\log(q))$ operations. %
  The final reduction modulo $h$ costs $O(r\MM(r)\log(r))$ operations,
  which is negligible in front of the previous step.
\end{proof}


The following proposition summarizes our analysis. To clarify the
order of magnitude of the exponents, let us assume $q=O(1)$ and
neglect polylogarithmic factors; then, if $\omega=2.38$ (best bound to
date), the runtimes are $O(s^{0.69}r^{1.69})$ for $s \in O(r^{0.23})$,
$O(r^{1.85}+s^{1.38}r)$ for $s \in \Omega(r^{0.23})$ and $s\in
O(r^{0.72})$, and $\tildO(r^2)$ otherwise.
For $\omega=3$, all costs are at best quadratic.

\begin{proposition}
  \label{proposition:XiDelta-updated}
  Given $k,K$ of degree $r$ over $\F_q$, assuming that $s$ is the
  order of $q$ in $(\Z/r\Z)^\times$,   Algorithm~\ref{alog:allombert} 
 computes its output using 
  \begin{itemize}
  \item $O(s^{(\omega-1)/2}r^{(\omega+1)/2}\log(r)+\MM(r)\log(q))$
    expected operations in $\F_q$ if $s \in O(r^{(\omega-3)/(\omega-5)})$, or
  \item
    $O(
 r^{(\omega^2-4\omega-1)/(\omega-5)}+(s+r^{2/(5-\omega)})\MM(r)\log(q)
+s^{\omega-1}r\log(r)\log(s))$
    expected operations in $\F_q$ if if $s \in \Omega(r^{(\omega-3)/(\omega-5)})$
    and $s \in O(r^{1/(w-1)})$, or
  \item $O(\MM(r^2)\log^2(r) + \MM(r)\log(r)\log(q))$ expected operations in
    $\F_q$ otherwise.
  \end{itemize}
\end{proposition}
\begin{proof}
  The cost of factoring the $r$-th cyclotomic polynomial is an
  expected $O(\MM(r)\log(rq))$ operations in $\F_q$,
  using~\cite[Theorem~9]{shoup94}. This is negligible compared with
  other steps. The solutions $\theta_1,\theta_2$ to Hilbert's theorem
  90 are computed as described above, according to the size of $s$.
  The powers $\theta_1^r,\theta_2^r$ are computed using Kronecker
  substitution in $O(\MM(sr)\log(r))$ operations, which is also
  negligible. Finally, the cost of computing an $r$-th root in
  $\F_q(\zeta)$ is given by Proposition~\ref{prop:root-fpz} and can
  not be neglected.

  Combining the costs coming from the solution to Hilbert's theorem 90
  and the $r$-th root extraction, we obtain the following complexities
  according to $s$.
  \begin{itemize}
  \item If we use the algorithm described in
    our first solution, combining Proposition~\ref{prop:xitheta} with the first
    case of Proposition~\ref{prop:root-fpz}, we obtain an estimate of
    $O(s^{(\omega-1)/2}r^{(\omega+1)/2}\log(r)+\MM(r)\log(q))$
    operations.
  \item If we use the algorithm described in
    our second solution, combining Proposition~\ref{prop:ks-theta} with the first
    case of Proposition~\ref{prop:root-fpz}, we obtain an estimate of
    $O( r^{(\omega^2-4\omega-1)/(\omega-5)}+(s+r^{2/(5-\omega)})\MM(r)\log(q)
+s^{\omega-1}r\log(r)\log(s)+\MM(rs)\log(r)\log(s))$
    operations.
  \item Otherwise, we use the algorithm described in our third solution. Combining
    Proposition~\ref{prop:iter-frob-theta} with the second case of
    Proposition~\ref{prop:root-fpz}, and replacing $s$ with $r$
    everywhere, we obtain an estimate of
    $O(\MM(r^2)\log^2(r) + \MM(r)\log(r)\log(q)) $ expected operations.
  \end{itemize}
  For $s\in O(r^{(\omega-3)/(\omega-5)})$, the first solution has the
  better runtime. Assuming $s\in \Omega(r^{(\omega-3)/(\omega-5)})$,
  the runtime in the second case can be written as $O(
  r^{(\omega^2-4\omega-1)/(\omega-5)}+(s+r^{2/(5-\omega)})\MM(r)\log(q)
  +s^{\omega-1}r\log(r)\log(s))$. If in addition $s$ is in $O(r^{1/(w-1)})$,
  this runtime is subquadratic, that is, better than that in our third solution.
\end{proof}

\subsection{The Artin--Schreier case}
\label{sec:fast-artin-schreier}

This section is devoted to the case $r = p^d$ for some positive integer $d$.
The technique we present here originates in Adleman and Lenstra's work~\cite[Lemma 5]{Adleman-Lenstra},
and appears again in Lenstra's~\cite{LenstraJr91}
and Allombert's~\cite{Allombert02}.
The chief difference with previous work once again consists
in replacing linear algebra
with a technique to solve the additive version of Hilbert's theorem 90
similar to the one in the previous section.
Recently, Narayanan~\cite[Sec.~4]{narayanan2016fast} independently described
a related variant with a similar complexity.

The idea is to build a 
tower inside the extension $k/\F_q$ using polynomials of the form $X^p - X - a$ where $a \in k$.
To  start, let $a_1 \in \F_q$ be such that $\trace_{\F_q/\F_p}(a_1) \ne 0$.
Let $\sigma \in  \gal(\F_q/\F_p)$ be a generator of the Galois group.
Then by the additive version of  Hilbert's theorem 90 there is no element 
$\alpha \in \F_q$ such that $\sigma(\alpha) - \alpha = a_1$.
Equivalently, the polynomial $f_1 =  X^p - X - a_1$ has no root in $\F_q$.
By the Artin--Schreier theorem in~\cite[Ch VI]{lang} $f_1$ 
is irreducible over $\F_q$. For a root $\alpha_1 \in k$ of $f_1$ the extension $\F_q(\alpha_1) / \F_q$ is of degree $p$.
Now let $a_2 = a_1\alpha_1^{p - 1}$. Then by~\cite[Lemma 5]{Adleman-Lenstra} the polynomial $f_2 = 
X^p - X - a_2$ is irreducible over $\F_q(\alpha_1)$. So, for a root $\alpha_2 \in k$ of $f_2$ the 
extension $\F_q(\alpha_2, \alpha_1) / \F_q(\alpha_1)$ is of degree $p$. Continuing the above 
process we build a tower
\begin{equation}
	\label{equ:art-sch-tower}
	\F_q \subset \F_q(\alpha_1)  \subset \cdots \subset \F_q(\alpha_1, \cdots, \alpha_d) = k.
\end{equation}
The idea of building such tower using the Artin--Schreier polynomials $f_i$ can also be found
in~\cite{LenstraJr91, Allombert02, shoup93}. By construction, $\alpha_i \notin \F_q(\alpha_1, \cdots, 
\alpha_{i - 1})$ for all $1 \le i \le d$. This means that the minimal polynomial of $\alpha_d$ over 
$\F_q$ is of degree $r = p^d$. Therefore, $k = \F_q(\alpha_d)$, and the element $\alpha_d$ is 
uniquely defined up to $\F_q$-isomorphism.

The above construction boils down to computing a root of the polynomial $f = X^p - X - a \in k[X]$.
We now show how to efficiently compute such a root.
By construction, $a$ is always in an intermediate subfield 
$\F_q \subseteq k' \subset k$. This means 
\[ \trace_{k / \F_p}(a) = \trace_{k' / \F_p}(\trace_{k / k'}(a)) = \trace_{k' / \F_p}(p^ia) = 0  \]
for some $i > 0$. By Hilbert's theorem 90 there exists $\alpha \in k$ such that $\alpha - 
\sigma(\alpha) = -a$ for a generator $\sigma \in \gal(k / \F_p)$. In other words, $\alpha^p - 
\alpha - a = 0$. Therefore, $\alpha$ is a root of $f$. On the other hand, for a random element 
$\theta \in k$ with nonzero trace, $\alpha$ can be explicitly set as
\begin{equation}
	\label{equ:art-sch-hil}
	\alpha = \frac{1}{\trace(\theta)}[a\sigma(\theta) + (a + \sigma(a))\sigma^2(\theta) + \cdots + 
	(a + \sigma(a) + \cdots + \sigma^{rt - 2}(a))\sigma^{rt - 1}(\theta)]
\end{equation}
where $t = [\F_q : \F_p]$.
To compute $\alpha$ using Eq.~\eqref{equ:art-sch-hil} efficiently, we define
\[ \xi_i = \sigma^i(x), \quad \beta_i(u) = u + \sigma(u) + \cdots + \sigma^{i - 1}(u), \quad 
\alpha_i(v) = \beta_1(a)\sigma(v) + \cdots + \beta_i(a)\sigma^i(v). \]
A simple calculation gives
\[ \alpha_{j + k}(v) = \alpha_j(v) + \sigma^j(\alpha_k(v)) + \beta_j(a)\sigma^{j + 1}(\beta_k(v)). 
\]
From these we can extract the following recursive relations:
\[
\begin{aligned}
	\xi_j & = 
	\begin{cases}
		\sigma^{j / 2}(\xi_{j / 2}) & j \text{ even} \\
		\sigma(\xi_{j - 1}) & j \text{ odd}
	\end{cases}, \\
	\beta_j(u) & = 
	\begin{cases}
		\beta_{j / 2}(u) + \sigma^{j / 2}(\beta_{j / 2}(u)) & j \text{ even} \\
		u + \sigma(\beta_{j - 1}(u)) & j \text{ odd}
	\end{cases}, \\
	\alpha_{j}(v) & = 
	\begin{cases}
		\alpha_{j / 2}(v) + \sigma^{j / 2}(\alpha_{j / 2}(v)) + \beta_{j / 2}(a)\sigma^{j / 2 + 
		1}(\beta_{j / 2}(v)) & j \text{ even} \\
		\alpha_1(v) + \sigma(\alpha_{j - 1}(v)) + a\sigma^2(\beta_{j - 1}(v)) & j \text{ odd} 
	\end{cases}
\end{aligned}
\]
Therefore, the values $\trace(\theta) = \beta_{rt}(\theta)$, and $\alpha = 
\beta_{rt}(\theta)^{-1}\alpha_{rt}(\theta)$ can be computed recursively, in $O(\log(rt))$ steps. At 
step $j$ of the recursive algorithm, $\xi_j, \beta_j(a), \beta_j(\theta), \alpha_j(\theta)$ are 
computed. As before, the action of $\sigma^j$ is the same as composing with $\xi_j$. So each step 
of the recursion is dominated by $O(1)$ modular compositions over $\F_q$ at the cost of $O(r^{(\omega+1)/2})$ 
operations in $\F_q$. The initial value of $\xi_1 = x^q$ is computed using $O(\MM(r)\log(q))$ 
operations in $\F_q$. Therefore, the cost of computing a root of $f$ is $O(r^{(\omega+1)/2}\log(rt) + 
\MM(r)\log(q))$ operations in $\F_q$.

Now, to compute $\alpha_d$ in Eq.~\eqref{equ:art-sch-tower} we 
need to take $d$ roots where $d \in O(\log(r) / \log(p))$ which leads to the following result. (Note that $\xi_1$ is computed only once and reused thereafter.)

\begin{proposition}
	Let $r = p^d$ for a positive integer $d$, and let $t = [\F_q : \F_p]$. An isomorphism of two 
	extensions $k / \F_q$, $K / \F_q$ of degree $r$ can be constructed using 
	$O(r^{(\omega+1)/2}\log(rt)\log(r) + \MM(r)\log(q))$ operations in $\F_q$.
\end{proposition}

\subsection{High-degree prime powers}
\label{sec:fast-algor-large}

We end this section with an algorithm that is particularly efficient when the extension 
degree $r$ is a high-degree prime power.
Allombert's algorithm works well in this case, however its
complexity depends linearly on the order $s$ of $q$ modulo $r$. If $r=v^d$ for some prime $v\ne p$,
it is natural to seek an algorithm which depends on the order of $q$ modulo $v$ instead.
The idea we present is a variation on Lenstra's algorithm, using successive $v$-th root
extractions.
We are not aware of this algorithm appearing anywhere in the literature. %
We also note that Narayanan~\cite[Sec.~5]{narayanan2016fast} recently
published a radically different generalization
of Allombert's algorithm with a very similar complexity in $r$ (his
algorithm has much worse complexity in $q$, though).

An overview of our construction is as follows. Let $r = v^d$ where $v \ne p$ is a 
prime and $d$ is a positive integer. Suppose the extension $k/\F_q$ is of degree $r$. Let $s$ be 
the order of $q$ in $\Z / v\Z$, and write $q^s - 1 = uv^t$ where $\gcd(v, u) = 1$. We first move to 
cyclotomic field extensions $\F_q(\zeta), k(\zeta), K(\zeta)$ of degree $s$ over $\F_q, k, K$ respectively, by 
obtaining an irreducible factor of the $v$-th cyclotomic polynomial over $\F_q$. Then we obtain a  
random non-$v$-adic residue $\eta \in \F_q(\zeta)$.

We have $[k(\zeta): \F_q(\zeta)] = r$, so we can 
compute an $r$-th root $\theta$ of $\eta$ in $k(\zeta)$ using $d$ successive $v$-th 
root extractions in $k(\zeta)$. Therefore, $\theta$ is a generator for the unique subgroup of 
$k(\zeta)^*$ of order $v^{d + t}$. Then the constant term $\alpha$ of $\theta$ is such that $k = 
\F_q(\alpha)$. %
Doing the same in $K$ yields an element $\beta\in K$ such that
the map $\alpha\mapsto\beta$ defines an isomorphism. %
The main difficulty in applying such an algorithm resides in computing
efficiently $v$-th roots in $k(\zeta)$, for which we use Proposition~\ref{prop:root_high_degree_extension};
this yields the main result of this section.
\begin{theorem}
	\label{theorem:isom-root}
	Let $r = v^d$ where $v \ne p$ is a prime and $d$ is a positive integer. Also let $s$ be the 
	order of $q$ in $\Z / v\Z$. Given extensions $k/\F_q$, $K/\F_q$ of degree $r$, an
        embedding $k\hookrightarrow K$ can be constructed at the cost of 
an expected $O(s^{(\omega-1)/2}r^{(\omega+1)/2}\log(r)^2 + s\MM(r)\log(r)\log(q)$
operations in $\F_q$.
\end{theorem}

\begin{proof}
We can construct the embedding of Theorem~\ref{theorem:isom-root} as
follows. We first build the extensions $k(\zeta) / \F_q(\zeta)$ and
$K(\zeta)/\F_q(\zeta)$. Let $\eta$ be a non-$v$-adic residue in
$\F_q(\zeta)$.  Then $\eta$ is an $r$-power in $k(\zeta)$ and
$K(\zeta)$. To obtain $r$-th roots $\theta_1\in k$, $\theta_2\in K$ of
$\eta$ we take $d$ successive $v$-th roots.

\begin{algorithm}
	[Kummer-type algorithm for extension towers]
	\label{algorithm:alpha-root}
	\begin{algorithmic}[1]
		\REQUIRE Extensions $k / \F_q$ $K/\F_q$ of degree prime-power $r = v^d$, with $v\ne p$.
		\ENSURE The description of a field embedding $k\hookrightarrow K$.
		\STATE\label{step:factor-cyclo} Factor the $v$-th cyclotomic polynomial over $\F_q$ to 
		build the extensions $k(\zeta) / \F_q(\zeta)$ and $K(\zeta)/\F_q(\zeta)$;
		\STATE\label{step:non-v-res} Find a random non-$v$-adic residue $\eta \in \F_q(\zeta)$;
		\STATE\label{step:rep-root1} Compute $r$-th roots $\theta_1,\theta_2$ of $\eta$ in $k(\zeta),K(\zeta)$;
		\STATE Let $\alpha, \beta$ be the constant terms of $\theta_1, \theta_2$ respectively;
		\RETURN The field embedding defined by $\alpha\mapsto\beta$.
	\end{algorithmic}
\end{algorithm}

Step~\ref{step:factor-cyclo} is done using~\cite[Theorem~9]{shoup94},
which takes $O(\MM(v)\log(vq))$ operations in $\F_q$. We do
Step~\ref{step:non-v-res} by taking random elements in $\F_q(\zeta)$
until a non-$v$-adic residue is found. Testing $v$-adic residuosity of
$\eta$ amounts to computing $\eta^{(q^s-1)/v}$ in $\F_q(\zeta)$, which
can be done in $O(s^{(\omega-1)/2}\log(s) + \MM(s)\log(v)\log(s) +
\MM(s)\log(q))$ operations in $\F_q$, in view of our discussion in
Section~\ref{sec:preliminaries}.

Step~\ref{step:rep-root1} is done using $d = O(\log(r) / \log(v))$
successive root extractions, each of which takes an expected
$O(s^{(\omega-1)/2}r^{(\omega+1)/2}\log(r) + s\MM(r)\log(q))$
operations in $\F_q$. Therefore Algorithm~\ref{algorithm:alpha-root}
runs in an expected $O(s^{(\omega-1)/2}r^{(\omega+1)/2}\log(r)^2 + s\MM(r)\log(r)\log(q)$ operations in
$\F_q$.
\end{proof}

\section{Rains' algorithm}
\label{sec:rains-algorithm}

We now move on to a different family of algorithms based on the theory of
algebraic groups. The simplest of these is Pinch's
cyclotomic algorithm~\cite{Pinch}. The idea is very simple: given $r$,
select an integer $\ell$ such that $[\F_q(\mu_\ell):\F_q]=r$, where
$\mu_\ell$ is the group of $\ell$-th roots of unity.  Then, any
embedding $k\to K$ takes $\mu_\ell\subset k^\ast$ to $\mu_\ell\subset
K^\ast$, and the minimal polynomial of any primitive $\ell$-th root of
unity has degree exactly $r$.

Pinch's algorithm is very effective when $r=\euler(\ell)$. Indeed in
this case the $\ell$-th cyclotomic polynomial $\Phi_\ell$ is
irreducible over $\F_q$, and its roots form a unique orbit under the
action of the absolute Galois group of $\F_q$. Thus we can take any
primitive $\ell$-th roots of unity $\alpha\in k$ and $\beta\in K$ to
describe the embedding.

In the general case, however, the roots of $\Phi_\ell$ are partitioned
in $\euler(\ell)/r$ orbits, thus for two randomly chosen $\ell$-th
roots of unity $\zeta_1\in k$ and $\zeta_2\in K$, we can only say that
there exists an exponent $e$ such that
\begin{equation*}
  \alpha = \zeta_1 \mapsto \zeta_2^e = \beta
\end{equation*}
defines a valid embedding. Pinch's algorithm tests all possible
exponents $e$, until a suitable one is found. To test for the validity
of a given $e$, it applies the embedding $\phi:\zeta_1\mapsto\zeta_2$
to the class of $X$ in $k$, and verifies that its image is a root of
$f$ in $K$ (see Section~\ref{sec:eval} for details on embedding
evaluation).

The trial-and-error nature of Pinch's algorithm makes it impractical,
except for rare favorable cases where a \emph{small} $\ell$ such that
$r=\euler(\ell)$ can be found. One possible workaround, suggested by
Pinch himself, is to replace the group of roots of unity with a group
of torsion points of a well chosen elliptic curve. We analyze this
idea in greater detail in Section~\ref{sec:rains-elliptic}.

This section is devoted to a different way of improving Pinch's
algorithm, imagined by Rains~\cite{rains2008}, and implemented in the
Magma computer algebra system~\cite{MAGMA}. Rains' technical
contribution is twofold: first he replaces roots of unity with
Gaussian periods to avoid trial-and-error, second he moves to slightly
larger extension fields to ensure the existence of a small $\ell$ as
above.

\subsection{Uniquely defined orbits from Gaussian periods}

For the rest of the section, we are going to assume that $q$ is
prime. The case where $q$ is a higher power of a prime is discussed in
Note~\ref{note:rains-non-prime}.

Suppose that we have an $\ell$, coprime with $q$, such that
$[\F_q(\mu_\ell):\F_q]=r$, then the cyclotomic polynomial $\Phi_\ell$
factors over $\F_q$ into $\euler(\ell)/r$ distinct factors of degree
$r$. Pinch's method, by choosing random roots of $\Phi_\ell$ in $k$
and $K$, randomly selects one of these factors as minimal polynomial.
By combining the roots of $\Phi_\ell$ into Gaussian periods, Rains'
method uniquely selects a minimal polynomial of degree $r$.

\begin{definition}
  Let $q$ be a prime, and let $\ell$ be a squarefree integer such that
  $(\Z/\ell\Z)^\times = \langle q\rangle \times S$ for some $S$.  For any
  generator $\zeta_\ell$ of $\mu_\ell$ in $\F_q(\mu_\ell)$, define the
  Gaussian period $\eta_q(\zeta_\ell)$ as
  \begin{equation}
    \eta_q(\zeta_\ell) = \sum_{\sigma\in S}{\zeta_\ell^{\sigma}}.
  \end{equation}
\end{definition}

It is evident from the definition that the Galois orbit of
$\eta_q(\zeta_\ell)$ is independent of the initial choice of
$\zeta_\ell$. Much less evident is the fact that this orbit has
maximal size and forms a normal basis of $\F_q(\mu_\ell)$, as stated
in the following lemma.

\begin{lemma}
  \label{th:gaussian}
  Let $q$ be a prime, and let $\ell$ be a squarefree integer such that
  $(\Z/\ell\Z)^\times = \langle q\rangle \times S$ for some $S$.  The
  periods $\eta_q(\zeta_\ell^\tau)$ for $\tau$ running through
  $\langle q\rangle$ form a normal basis of $\F_q(\mu_\ell)$ over
  $\F_q$, independent of the choice of $\zeta_\ell$.
\end{lemma}
\begin{proof}
  See~\cite[Main Theorem]{feisel1999normal}.
  The main idea of the proof is to show that cyclotomic units are
  normal in characteristic zero, then that integrality conditions
  carry normality through reduction modulo $q$.
\end{proof}

In what follows we are going to write $\eta(\zeta_\ell)$ when $q$ is
clear from the context.

\begin{example} 
  Consider the extension $\F_8/\F_2$ of degree $3$, which is generated
  by the $7$-th roots of unity. We have a decomposition
  $(\Z/7\Z)^\times=\langle 2\rangle\times\langle-1\rangle$, and the
  cyclotomic polynomial factors as
  \begin{equation}
    \Phi_7(X) = (X^3 + X + 1) (X^3 + X^2 + 1).
  \end{equation}
  For any root $\zeta_7$, we define the period
  \begin{equation}
    \eta_2(\zeta_7) = \zeta_7+\zeta_7^{-1}.
  \end{equation}
  The three periods $\eta_2(\zeta_7)$, $\eta_2(\zeta_7)^2$ and
  $\eta_2(\zeta_7)^4$ are all roots of the polynomial $x^3+x^2+1$ and
  form a normal basis of $\F_8/\F_2$.
\end{example}

\subsection{Rains' cyclotomic algorithm}

The bottom-line of Rains' algorithm follows immediately from the
previous section: given $k$, $K$ and $r$,
\begin{enumerate}
\item find a \emph{small} $\ell$ satisfying the conditions of
  Lemma~\ref{th:gaussian} with $[\F_q(\mu_\ell):\F_q]=r$;
\item take random $\ell$-th roots of unity $\zeta_\ell\in k$ and
  $\zeta_\ell'\in K$;
\item return the Gaussian periods $\alpha_r=\eta(\zeta_\ell)$ and
  $\beta_r=\eta(\zeta_\ell')$.
\end{enumerate}

The problem with this algorithm is the vaguely defined
\emph{smallness} requirement on $\ell$. Indeed the conditions of
Lemma~\ref{th:gaussian} imply that $\ell$ divides $\Phi_r(q)$, thus in
the worst case $\ell$ can be as large as $O(q^{\euler(r)})$, which
yields an algorithm of exponential complexity in the field size.

To circumvent this problem, Rains allows the algorithm to work in
small auxiliary extensions of $k$ and $K$, and then descend the
results to $k$ and $K$ via a field trace. In other words, Rains'
algorithm looks for $\ell$ such that $[\F_q(\mu_\ell):\F_q]=rs$ for
some small $s$. We summarize this method in
Algorithm~\ref{algorithm:rains-cyclo}; we only give the procedure for
the field $k$, the procedure for the field $K$ being identical.

\begin{algorithm}[Rains' cyclotomic algorithm]
  \label{algorithm:rains-cyclo}
  \begin{algorithmic}[1]
    \REQUIRE A field extension $k/\F_q$ of degree $r$; a squarefree
    integer $\ell$ such that
    \begin{itemize}
    \item $(\Z/\ell\Z)^\times = \langle q\rangle \times S$ for some $S$,
    \item $\#\langle q\rangle = rs$ for some integer $s$;
    \end{itemize}
    a polynomial $h$ of degree $s$ irreducible over $k$.
    \ENSURE A normal generator of $k$ over $\F_q$,
    with a uniquely defined Galois orbit.
    
    \STATE Construct the field extension $k'=k[Z]/h(Z)$;
    \REPEAT
    \STATE\label{algorithm:rains-cyclo:power} Compute $\zeta\leftarrow \theta^{(\#k'-1)/\ell}$ for a random $\theta\in k'$
    \UNTIL $\zeta$ is a primitive $\ell$-th root of unity;
    \STATE\label{algorithm:rains-cyclo:period} Compute $\eta(\zeta) \leftarrow \sum_{\sigma\in S}\zeta^\sigma$;
    \RETURN\label{algorithm:rains-cyclo:trace} $\alpha \leftarrow \trace_{k'/k}\eta(\zeta) = \sum_{i=0}^{s-1}\eta(\zeta)^{q^{ri}}$.
  \end{algorithmic}
\end{algorithm}

\begin{proposition}
  Algorithm~\ref{algorithm:rains-cyclo} is correct. On input
  $q,r,\ell,s$ it computes its output using
  $O(sr^{(\omega+1)/2}\log(sr) + \MM(sr)(\log(q) + (\ell/r)\log(\ell)))$ operations in $\F_q$ on average.
\end{proposition}
\begin{proof}
  By construction $k'$ is isomorphic to $\F_q(\mu_\ell)$. By
  Lemma~\ref{th:gaussian} $\eta(\zeta)$ is a normal generator of $k'$,
  and by~\cite[Prop.~5.2.3.1]{mullen2013handbook} $\alpha$ is a
  normal generator of $k$. This proves correctness.

  According to Proposition~\ref{prop:trace-like}, computing $\zeta$
  in Step~\ref{algorithm:rains-cyclo:power}
  costs
  $$O\bigl(\bigl(s^{(\omega+1)/2}\MM(r) + sr^{(\omega+1)/2}
  +\MM(sr)\log(\ell)\bigr)\log(sr)+\MM(sr)\log(q)\bigr),
  $$
  and the loop is executed $O(1)$ times on average. %
  By observing that $s^{(\omega-1)/2}\in O(\ell/r)$, this fits into
  the stated bound.
  
  Steps~\ref{algorithm:rains-cyclo:period}
  and~\ref{algorithm:rains-cyclo:trace} can be performed at once by
  observing that
  \[\alpha = \sum_{i=0}^{s-1}\eta(\zeta^{q^{ri}})= \sum_{i=0}^{s-1}\sum_{\sigma\in S}\zeta^{q^{ri}\sigma}.\]
  By reducing $q^{ri}\sigma$ modulo $\ell$, we can compute this sum at
  the cost of $\euler(\ell)/r$ exponentiations of degree at most
  $\ell$ in $k'$, for a total cost of
  $O((\MM(sr)(\ell/r)\log(\ell))$,
  using the techniques of Section~\ref{sec:preliminaries}.
  The final result is obtained as an element of $k$.
\end{proof}

The attentive reader will have noticed the irreducible polynomial $h$
of degree $s$ given as input to Rains' algorithm. Computing this
polynomial may be expensive. For a start, we may ask $s$ to be coprime
with $r$, so that $h$ can be taken with coefficients in $\F_q$. Then,
for small values of $s$ and $q$, one may use a table of irreducible
polynomials. For larger values, the
constructions~\cite{couveignes+lercier11,DeDoSc13,DeDoSc2014}
are reasonably efficient, and yield an irreducible polynomial in time
less than quadratic in $s$.
However negligible from an asymptotic point of view, the construction
of the polynomial $h$ and of the field $k'$ take a serious toll on the
practical performances of Rains' algorithm.%
\footnote{
A straightforward way to avoid these constructions consists in computing
a factor $h$ of the cyclotomic polynomial $\Phi_\ell$
over the extension $k$ following case~$5$ from Section~\ref{sec:fundamentalgo}.
Then, using Newton's identities, the period can be recovered from the
logarithmic derivative of the reciprocal of $h$. %
Nevertheless, the cost of factoring $\Phi_\ell$ renders this approach
unpractical.}

This concludes the presentation of Rains' algorithm. However, we are
still left with a problem: how to find $\ell$ satisfying the
conditions of the algorithm, and what bounds can be given on it. These
questions will be analyzed in Section~\ref{sec:selection}.

\begin{note}
  \label{note:rains-non-prime}
  Rains' algorithm is easily extended to a non-prime field $\F_q$, as
  long as $q=p^d$ with $\gcd(d,r)=1$. In this case, indeed, any
  generator of $\F_{p^r}$ over $\F_p$ is also a generator of
  $\F_{q^r}$ over $\F_q$. The algorithm is unchanged, except for the
  additional requirement that $\gcd(\euler(\ell),d)=1$, which ensures
  that the Gaussian periods indeed generate $\F_{p^r}$.

  However, when $\gcd(d,r)\ne 1$, it is impossible to have
  $(\Z/\ell\Z)^\times=\langle q\rangle\times S$, so Rains' algorithm
  simply cannot be applied to this case. In the next section we are
  going to present a variant that does not suffer from this problem.
\end{note}


\section{Elliptic Rains' algorithm}
\label{sec:rains-elliptic}

The Pinch/Rains' algorithm presented in the previous section relies on the use
of the multiplicative group of finite fields.
It is natural to try to extend it to other types of algebraic groups in
order to cover a wider range of parameters.
And indeed Pinch~\cite{Pinch} showed how to use torsion points of elliptic
curves in place of roots of unity.
Rains also considered this possibility, but did not investigate it thoroughly
as no theoretical gain was to be expected.
However, the situation in practice is quite different.
In particular, the need for auxiliary extensions in the cyclotomic method
is very costly, whereas the elliptic variant has naturally more chances
to work in the base fields, and to be therefore very competitive.

In the next sections, we first introduce \emph{elliptic periods}, a
straightforward generalization of Gaussian periods for torsion points
of elliptic curves,
then analyze the cost of their computation.
The main issue with this generalization is that, contrary to Gaussian periods,
elliptic periods do not yield normal bases of finite fields.
We still provide experimental data and heuristic arguments
to support the benefit of using them.
Whether they always yield an element generating the right field extension,
a weak counterpart to Lemma~\ref{th:gaussian}, is left as an open problem.

\subsection{Uniquely defined orbits from elliptic periods}
\label{sec:ellperiods}

An elliptic curve $E/L$ defined over a field $L$ is given by an
equation of the form
\begin{equation*}
  E\;:\; y^2 + a_1xy + a_3y = x^3 + a_2x^2 + a_4x + a_6
  \qquad\text{with $a_1,a_2,a_3,a_4,a_6\in L$.}
\end{equation*}
For any field extension $M/L$ the group of $M$-rational points of $E$
is the set
\begin{equation*}
  E(M) = \{(x,y)\in M^2 \mid E(x,y) = 0\} \cup \{\mathcal{O}\}
\end{equation*}
endowed with the usual group law, where $\mathcal{O}$ is the point at
infinity.

For an integer $\ell$, we denote by $E[\ell]$ the $\ell$-torsion
subgroup of $E(\bar{L})$, where $\bar{L}$ denotes the algebraic
closure of $L$. In this section we are going to consider integers
$\ell$ coprime with the characteristic of $L$, then $E[\ell]$ is a
group of rank $2$.

For an elliptic curve $E/\F_q$ defined over a finite field, we denote
by $\pi$ its \emph{Frobenius endomorphism}. It is well known that
$\pi$ satisfies a quadratic equation $\pi^2-t\pi+q=0$, where $t$ is
called the \emph{trace of $E$}, and that this equation determines the
cardinality of $E$ as $\#E(\F_q)=q+1-t$.

Like in the cyclotomic case, the Frobenius endomorphism partitions
$E[\ell]$ into orbits. Our goal is to take traces of points in
$E[\ell]$ so that a uniquely defined orbit arises. This task is made
more complex by the fact that $E[\ell]$ has rank 2, hence we are going
to restrict to a family of primes $\ell$ named \emph{Elkies primes}.

\begin{definition}[Elkies prime]
  Let $E/\F_q$ be an elliptic curve, let $\ell$ be a prime number not
  dividing $q$.  We say that $\ell$ is an Elkies prime for $E$ if the
  characteristic polynomial of the Frobenius endomorphism $\pi$ splits
  into two distinct factors over $\Z/\ell\Z$:
\begin{equation}
\pi^2-t\pi+q=(\pi-\lambda)(\pi-\mu)\bmod\ell
\qquad\text{with $\lambda\ne\mu$}.
\end{equation}
\end{definition}

Note that if $\ell$ is an Elkies prime for $E$, then $E[\ell]$ splits
into two eigenspaces for $\pi$ which are defined on extensions of
$\F_q$ of degrees $\order_\ell(\lambda)$ and $\order_\ell(\mu)$. We
are now ready to define the elliptic curve analogue of Gaussian
periods.

\begin{definition}
  \label{definition:ellperiod}
  Let $E/\F_q$ be an elliptic curve of $j$-invariant not $0$ or
  $1728$. %
  Let $\ell > 3$ be an Elkies prime for $E$, $\lambda$ an eigenvalue
  of $\pi$, and $P$ a point of order $\ell$ in the eigenspace
  corresponding to $\lambda$ (i.e., such that $\pi(P)=\lambda P$).
  Suppose that there is a subgroup $S$ of $(\Z/\ell\Z)^{\times}$ such
  that
  \begin{equation}
    (\Z/\ell\Z)^{\times} = \langle{\lambda}\rangle \times S.
  \end{equation}
  
  Then we define an elliptic period as
  \begin{equation}
    \eta_{\lambda,S}(P) =
    \begin{cases}
      \sum_{\sigma\in S/\{\pm1\}} {x \left([\sigma] P \right)} & \text{if $-1\in S$,}\\
      \sum_{\sigma\in S} {x \left([\sigma] P \right)} & \text{otherwise,}
    \end{cases}
  \end{equation}
  where $x(P)$ denotes the abscissa of $P$.
\end{definition}

For a generalization of this definition that also covers the cases
$j=0,1728$, see Appendix~\ref{app:elliptic-curves}.

\begin{lemma}
  \label{lemma:ellperiods-order}
  With the same notation as in Definition~\ref{definition:ellperiod},
  let
  \begin{equation*}
    \#\langle\lambda\rangle =
    \begin{cases}
      r & \text{if $-1\notin\langle\lambda\rangle$,}\\
      2r & \text{otherwise.}
    \end{cases}
  \end{equation*}
  Then, for any point $P$ in the eigenspace of $\lambda$, the period
  $\eta_{\lambda,S}(P)$ is in $\F_{q^r}$, and its minimal polynomial
  does not depend on the choice of $P$.
\end{lemma}
\begin{proof}
  By construction, the Frobenius endomorphism $\pi$ acts on
  $\langle P\rangle$ as multiplication by the scalar $\lambda$. It is
  well known that two points have the same abscissa if and only if
  they are opposite, hence the Galois orbit of $x(P)$ has size $r$,
  and we conclude that both $x(P)$ and $\eta_{\lambda,S}(P)$ are in
  $\F_{q^r}$.

  Now let $P'=[a]P$ be another point in the eigenspace of
  $\lambda$. By construction, $a=\pm \lambda^i\sigma$, for some
  $0\le i<r$ and some $\sigma\in S$. Hence
  $\eta_{\lambda,S}(P')=\eta_{\lambda,S}([\lambda^i]P)$, implying that
  $\eta_{\lambda,S}(P)$ and $\eta_{\lambda,S}(P')$ are conjugates in
  $\F_{q^r}$.
\end{proof}

We remark that the previous lemma only states that the elliptic
periods $\eta_{\lambda,S}([\lambda^i]P)$ uniquely define an orbit
inside $\F_{q^r}$, but gives no guarantee that they generate the whole
$\F_{q^r}$. %
At this point, one would like to have an equivalent of
Lemma~\ref{th:gaussian} for elliptic periods, i.e.\ that the elliptic
period $\eta_{\lambda,S}(P)$ is a normal generator of $\F_q(x(P))$.
However, it is easy to find non-normal elliptic periods, as the
following example shows.

\begin{example}
\label{ex:non-normal}
  Let $E/\F_7$ be defined by $y^2 = x^3 + 5 x + 4$, and consider the
  degree $3$ extension of $\F_7$ defined by
  $k=\F_7[X]/(X^3 + 6 X^2 + 4)$. Then
  \begin{itemize}
  \item $\ell = 31$ is an Elkies prime for $E$;
  \item the eigenvalues of the Frobenius modulo $\ell$ are
    $\lambda = 25$ of multiplicative order $3$ and $\mu = 4$ of
    multiplicative order $5$;
  \item $P = (5 a^2+2 a, 4)$ is a point of order $31$ of $E/k$;
  \item $\eta=\eta_{\lambda,S}(P) = 5 a^2 + 5 a + 4$ is not a normal
    element, indeed
    $\eta + 4 \eta^7 + 2 \eta^{49} = 0$.
  \end{itemize}
\end{example}

All well known proofs of Lemma~\ref{th:gaussian} rely on the fact that
the $\ell$-th cyclotomic polynomial is irreducible over $\Q$, and its
roots form a normal basis of $\Q(\zeta_\ell)$. %
This fails in the elliptic case: there is indeed no guarantee that the
eigenspace of $\lambda$ can be lifted to a normal basis over some
number field.

Note however that, even if the elliptic period is not normal, it is
enough for our purpose that it generates $\F_q(x(P))$ as a field, like
in the example above.  In Appendix~\ref{app:ellprdsdata} we gather
some experimental evidence suggesting that this might always be the
case. Thus, we state this as a conjecture.

\begin{conjecture}
\label{conj:ellperiods}
With the above notation, the elliptic period $\eta_{\lambda,S}(P)$
generates $\F_q(x(P))$ over $\F_q$.
\end{conjecture}

If the conjecture is false, the only arguments we can give are of a
heuristic nature. First and most simply, we can assume that the
elliptic period behaves like a random element of $\F_q(x(P))$. In this
case the chance of it not being a generator is approximately
$1/q^r$. %
Secondly, in Appendix~\ref{app:ellprdsdata} we give a sufficient
condition for the period to be a normal generator of
$\F_q(x(P))$. This is a weak counterpart to Lemma~\ref{th:gaussian},
based on the polynomially cyclic algebras setting
of~\cite{Mihailescu2010825}. Heuristically, it suggests that the
chance of the period being normal is approximately $1-1/q$. %
We are now ready to present the generalization of Rains' algorithm,
with the warning that the algorithm may fail, with low probability, if
Conjecture~\ref{conj:ellperiods} is false.

\subsection{Elliptic variant of Rains' algorithm}

Rain's cyclotomic algorithm needs auxiliary extensions to accommodate
for sufficiently small subgroups $\mu_\ell$ of the unit group. By
replacing unit groups with torsion groups of elliptic curves, we gain
more freedom on the choice of the size of the group, thus we are able
to work with smaller fields.  

The algorithm is very similar to
Algorithm~\ref{algorithm:rains-cyclo}, and follows immediately from
the previous section. For simplicity, we are going to state it only
for $r$ odd. Given $k$, $K$ and $r$,
\begin{enumerate}
\item find a prime $\ell$, an elliptic curve $E$, and an eigenvalue
  $\lambda$ of the Frobenius endomorphism, satisfying the conditions
  of Definition~\ref{definition:ellperiod}, and such that
  $\order_\ell(\lambda)=r$;
\item take random points $P\in E(k)[\ell]$ and $P'\in E(K)[\ell]$ in
  the eigenspace of $\lambda$;
\item return the elliptic periods $\alpha := \eta_{\lambda,S}(P)$ and
  $\beta:= \eta_{\lambda,S}(P')$.
\end{enumerate}

Here we are faced with a difficulty: given $E$ and $\lambda$ it is
easy to pick a random point in $E[\ell]$, but it is potentially much
more expensive to compute a point in the eigenspace of $\lambda$. We
will circumvent the problem by forcing $E(\F_{q^r})[\ell]$ to be of
rank $1$, and to coincide exactly with the eigenspace of $\lambda$.
If we write $\mu = q/\lambda$ for the other eigenvalue of $\pi$, this
is easily ensured by further asking that $\order_\ell(\mu) \nmid r$.

We defer the discussion on the search for the elliptic curve $E$ to
Section~\ref{sec:selection}. Here we suppose that we are already given
suitable parameters $\ell$, $E$ and $\lambda$, and analyze the last
two steps of the algorithm, summarized below.  We only give the
procedure for $k$, the procedure for the field $K$ being
identical.

\begin{algorithm}[Elliptic Rain's algorithm]
\label{algorithm:compell}
  \begin{algorithmic}[1]
    \REQUIRE A field extension $k/\F_q$ of odd degree $r$,
    an elliptic curve $E/\F_q$, its trace $t$, a prime $\ell$ not dividing $q$,
    an integer $\lambda$ such that:
    \begin{itemize}
    \item $X^2 - tX + q = (X-\lambda)(X-q/\lambda) \mod\ell$,
    \item $\order_\ell(\lambda)=r$, $\order_\ell(q/\lambda)\nmid r$,
    \item $(\Z/\ell\Z)^{\times} = \langle{\lambda}\rangle \times S$ for some $S$.
    \end{itemize}
    \ENSURE A generator of $k$ over $\F_q$, with a uniquely defined Galois orbit, or FAIL.
    \REPEAT
    \STATE Compute $P\leftarrow[\# E(k)/\ell]Q$ for a random $Q\in E(k)$;
    \UNTIL{$P\neq\mathcal{O}$;}
    \STATE Compute $\alpha\leftarrow\eta_{\lambda,S}(P)$;
    \RETURN $\alpha$ if $k=\F_q(\alpha)$, FAIL otherwise.
  \end{algorithmic}
\end{algorithm}

\begin{proposition}
  Algorithm~\ref{algorithm:compell} is correct. Assuming the
  heuristics developed in Appendix~\ref{app:ellprdsdata}, it fails
  with probability $\le 1/q^r$.  On input
  $r,q,E,t,\ell,\lambda$ it computes its output using
  $O(\MM(r)(r\log(q) + (\ell/r)\log(\ell))$ operations in $\F_q$ on average, or
  $\tildO(r^2\log(q))$ assuming $\ell\in o(r^2)$.
\end{proposition}
\begin{proof}
  Correctness follows immediately from
  Lemma~\ref{lemma:ellperiods-order}. Success probability comes from
  the assumption that $\eta_{\lambda,S}(P)$ behaves like a random element of
  $\F_q(x(P))$, as discussed in Appendix~\ref{app:ellprdsdata}.

  From the knowledge of the trace $t$, we immediately determine the
  zeta function of $E$, and hence the cardinality $\# E(k)$, at
  no algebraic cost.

  To select the random point $Q\in E(k)$ we take a random
  element $x\in k$, then we verify that it is the abscissa of a
  point using a squareness test, at a costs of $O(r\MM(r)\log(q))$
  operations. Then, using Montgomery's formulas for scalar
  multiplication~\cite{montgomery}, we can compute the points $P$ and
  $[\ell]P$ without the knowledge of the ordinate of $Q$, at a cost of
  $O(r\MM(r)\log(q))$ operations. A valid point is obtained after
  $O(1)$ tries on average.

  The computation of the elliptic period $\alpha$ requires $O(\ell/r)$ scalar
  multiplications by an integer less than $\ell$, for a total cost of
  $O((\MM(r)(\ell/r)\log(\ell))$.

  Finally, testing that $\alpha$ generates $k$ is done by computing
  its minimal polynomial, at a cost of $O(r^{(\omega+1)/2})$ operations in
  $\F_q$ using~\cite{shoup93}.
\end{proof}

\section{Algorithm selection}
\label{sec:selection}

The algorithms presented in the previous sections have very similar
complexities, and no one stands out as absolute winner. The complexity
of all algorithms, besides the naive one, depends in a non-trivial way
on the parameters $q$ and $r$, and, for Rains' algorithms, on the
search for a parameter $\ell$ and an associated elliptic curve.

This section studies the complexity of the embedding description
problem from a global perspective. We give procedures to find
parameters for Rains' algorithms, criteria to choose the best among
the embedding algorithms, and asymptotic bounds on the embedding
description problem.

\subsection{Finding parameters for Rains' algorithms}

Given parameters $q=p^d$ and $r$, Rains' cyclotomic algorithm asks for
a \emph{small} parameter $\ell$ such that:
\begin{enumerate}
\item $(\Z/\ell\Z)^\times = \langle q\rangle \times S$ for some $S$,
\item $\langle q \rangle = rs$ for some integer $s$,
\item $\gcd(\euler(\ell),d)=1$ (see Note~\ref{note:rains-non-prime}).
\end{enumerate}

Since $r$ is a prime power, the second condition lets us take a prime
power for $\ell$ too. Indeed if
$\Z/\ell\Z\simeq\Z/\ell_1\Z\times\Z/\ell_2\Z$, then either
$q\bmod\ell_1$ or $q\bmod\ell_2$ has order a multiple of $r$.
Furthermore, if $\gcd(\ell,r)=1$, then we can take $\ell$ prime, since
higher powers would not help satisfy the conditions. On the other hand
if $\gcd(\ell,r)\ne1$, then the algorithms of Section~\ref{sec:kummer} have much
better complexity. Hence we shall take $\ell$ prime.

Given the above constraints, we can rewrite the conditions as:
\begin{enumerate}
\item $\ell = rsv + 1$ for some $s,u$ such that $\gcd(rs,v)=1$,
\item $\order_\ell(q) = rs$,
\item $\gcd(rsv,d)=1$.
\end{enumerate}

\begin{remark}
  Rains remarked that, when $q=2$ and $r$ is a power of $2$ greater
  than $4$, no $\ell$ can satisfy these constraints because $2$ is a
  quadratic residue modulo any prime of the form $8u+1$. This case,
  however, is covered by the Artin--Schreier technique in
  Section~\ref{sec:fast-artin-schreier}, we thus ignore it.
\end{remark}

In the elliptic algorithm we look for an integer $\ell$ and a curve
$E/\F_q$ that satisfy the preconditions of
Algorithm~\ref{algorithm:compell}, i.e., such that 
\begin{enumerate}
\item the Frobenius endomorphism $\pi$ satisfies a characteristic
  equation $(\pi-\lambda)(\pi-\mu) = 0 \mod \ell$,
\item $(\Z/\ell\Z)^\times = \langle\lambda\rangle\times S$ for some $S$,
\item $\#\langle\lambda\rangle=r$, and
\item $\mu^r\ne1\mod\ell$.
\end{enumerate}

As before, we only need to look at prime $\ell$. Because
$\mu=q/\lambda$, the last condition is equivalent to
$q^r\ne1\bmod\ell$. Hence, we can restate the conditions on $\ell$ as
\begin{enumerate}
\item $\ell = ru+1$ for some $u$ such that $\gcd(r,u)=1$,
\item $q^r\ne1\mod\ell$.
\end{enumerate}
Once $\ell$ is found, we compile a list of all integers of order $r$
in $(\Z/\ell\Z)^\times$, and look for a curve of trace
$t=\lambda+q/\lambda\bmod\ell$ for any $\lambda$ in the list. Note,
however, that for there to be such a curve, $t$ must have a
representative in the interval $[-2\sqrt{q},2\sqrt{q}]$. In order to
have a good chance of finding such curves, we are going to set an even
more stringent bound $\ell\in o(\sqrt[4]{q})$.

We propose a procedure to simultaneously find parameters for the
cyclotomic and the elliptic case in
Algorithm~\ref{algorithm:selectell}. The procedure is given bounds on
the size of the parameters sought, and outputs all suitable parameters
within those bounds.

\begin{algorithm}
    [Parameter selection for Rains' algorithms]
    \label{algorithm:selectell}
    \begin{algorithmic}[1]
      \REQUIRE Integers $q=p^d$ and $r$, bounds $\bar{u}$, $\bar{s}$, and $\bar{e}$;
      \ENSURE $\mathcal{C}$ and $\mathcal{E}$, lists of parameters for the cyclotomic and elliptic algorithm resp.
      \STATE $\mathcal{C}\leftarrow\{\}$, $\mathcal{E}\leftarrow\{\}$;
      \FOR{$u=1$ to $\bar{u}$}
      \IF{\label{alg:selectell:prime}$\ell=ur+1$ is prime}
      \IF{\label{alg:selectell:order}$\order_\ell(q)=rs$ with $s\le\bar{s}$, and $\gcd(rs,u/s)=1$, and $\gcd(ur,d)=1$}
      \STATE Add $\ell$ to $\mathcal{C}$.
      \ENDIF
      \IF{$\ell\le\bar{e}$ and $\gcd(u,r)=1$ and $q^r\ne1\bmod\ell$}
      \STATE\label{alg:selectell:ellorder} Compute $\mathcal{T} \leftarrow \{\lambda + q/\lambda \bmod\ell \;|\; \order_\ell(\lambda)=r\}$;
      \REPEAT\label{alg:selectell:ellloop}
      \STATE\label{alg:selectell:ellcount} Take random $E/\F_q$ and compute its trace $t$;
      \UNTIL{$(t\bmod\ell)\in\mathcal{T}$}
      \STATE Add $(\ell,E,t)$ to $\mathcal{E}$.
      \ENDIF
      \ENDIF
      \ENDFOR
      \RETURN $\mathcal{C}$ and $\mathcal{E}$.
    \end{algorithmic}
\end{algorithm}

\begin{proposition}
  On input $q$, $r$, $\bar{u}$ and $\bar{e}$,
  Algorithm~\ref{algorithm:selectell} computes its output using
  $\tildO\left(\sqrt{r\bar{u}^3} + \bar{u}^2\log^6(q)\right)$ binary
  operations, assuming $\bar{e}\in o(\sqrt[4]{q})$.
\end{proposition}
\begin{proof}
  We only consider naive integer arithmetic, since it is unrealistic
  to apply embedding algorithms to very large sizes.

  In Step~\ref{alg:selectell:prime} we need to test for the primality
  of $\ell$, while Steps~\ref{alg:selectell:order}
  and~\ref{alg:selectell:ellorder} require the factorization of
  $\ell-1$. Both operations can be performed in
  $\tildO(\sqrt{r\bar{u}})$ operations using naive algorithms.

  In Step~\ref{alg:selectell:ellcount} we need to count the number
  of points of an elliptic curve over $\F_q$. This can be done in
  $O\left(\log^6(q)\right)$ binary operations using the
  Schoof--Elkies--Atkin algorithm with naive integer
  arithmetic~\cite{schoof95,lercier+sirvent08}.

  All other operations have negligible cost compared to these ones. We
  finally need to account for the two loops in the algorithm. The
  inner loop at Step~\ref{alg:selectell:ellloop} stops when a curve
  with $t\bmod\ell$ in $\mathcal{T}$ is found. The set $\mathcal{T}$
  has size $\euler(r)$, hence, assuming that traces are evenly
  distributed modulo $\ell$, we expect to find a suitable curve after
  $O(\ell/r)\subset O(\bar{u})$ tries. Although it is well known that
  traces are not evenly distributed modulo prime
  numbers~\cite{lenstra87}, it is shown
  in~\cite[Th.~1]{castryck+hubrechts13} that the probability that a
  random curve has trace congruent to a fixed $t\bmod\ell$ approaches
  $1/\ell$, as $\ell$ and $q$ go to infinity, subject to $\ell\in
  o(\sqrt[4]{q})$. Hence, we shall assume that $\bar{e}\in
  o(\sqrt[4]{q})$ for the complexity analysis to hold.

  The outer loop multiplies the whole complexity by $\bar{u}$, we
  conclude that the overall complexity is in
  $\tildO\left(\sqrt{r\bar{u}^3} + \bar{u}^2\log^6(q)\right)$.
\end{proof}

\subsection{Selecting the best algorithm}

A natural question arises: what bounds $\bar{u},\bar{s},\bar{e}$ must
be taken to ensure that the lists $\mathcal{C}$, $\mathcal{E}$ in
Algorithm~\ref{algorithm:selectell} are non-empty?

It is not easy to give a precise answer: already the condition that
$\ell=ur+1$, in Step~\ref{alg:selectell:prime}, poses some
difficulties. Heuristically, we expect that about
$\bar{u}/\log(\bar{u}r)$ of those numbers are prime. However the best
lower bound on primes of the form $\ell=ur+1$, even under GRH, is
$\ell\in O(r^{2.4+\epsilon})$~\cite{heath1992zero}. Empirical data
show that the reality is much closer to the heuristic bound: in
Figure~\ref{fig:primes-arith-prog} we plot for all prime powers
$r<10^8$ the smallest $u$ such that $ur+1$ is prime. It appears that
$u$ is effectively bounded by $O(\log(r))$ for any practical purpose.

\begin{figure}
  \centering
  \includegraphics[width=\textwidth]{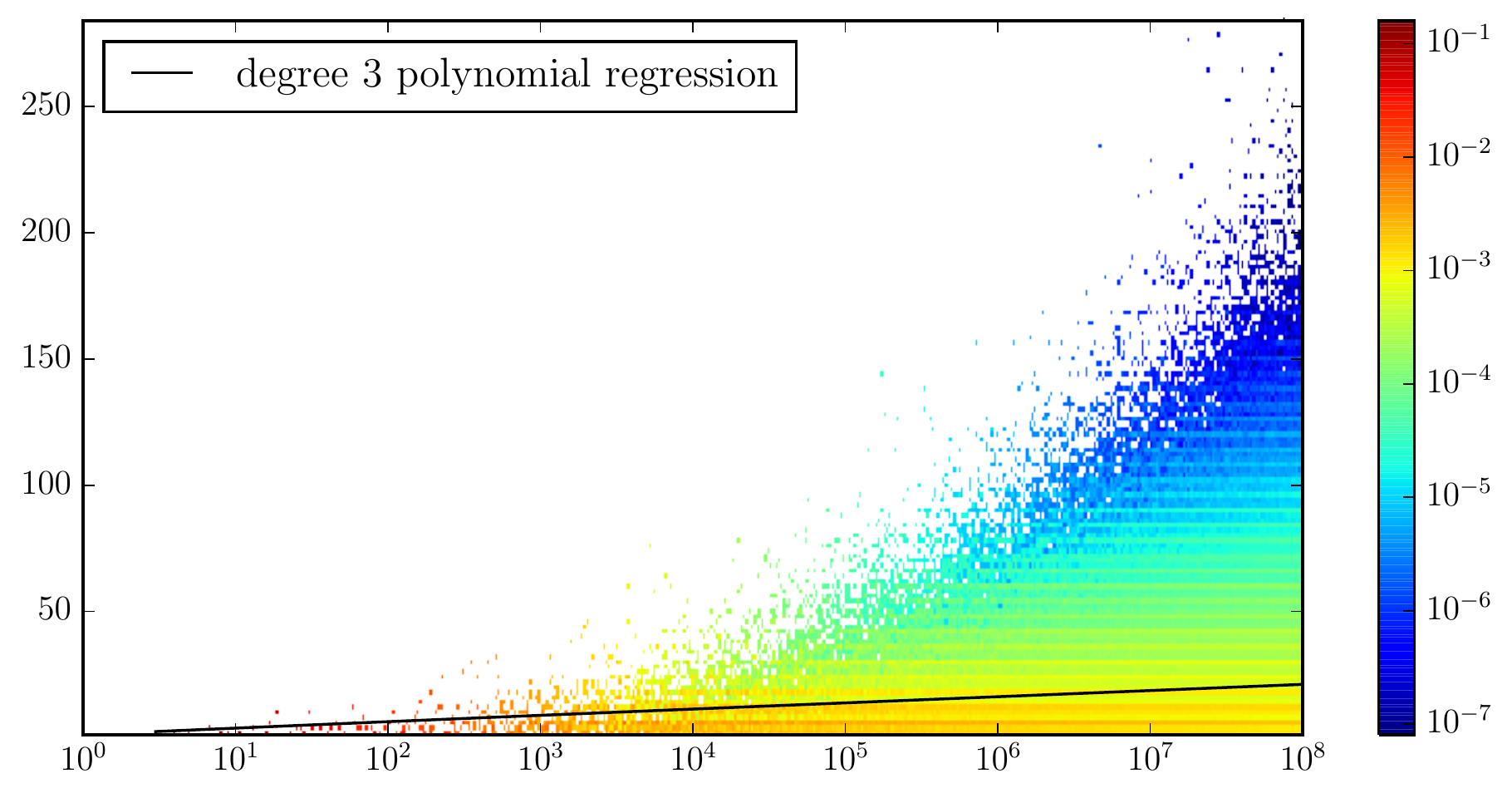}
  \caption{Prime powers $r$ (abscissa) versus smallest integer $u$
    (ordinate) such that $ur+1$ is prime. Abscissa in logarithmic
    scale, density normalized by $\log(x)/x$ and colored in
    logarithmic scale.}
  \label{fig:primes-arith-prog}
\end{figure}

For the cyclotomic algorithm we also require that $\order_\ell(q)$ is
a multiple of $r$. Assuming that $q$ is uniformly
distributed\footnote{This assumption is obviously false for any fixed
  $q$, but it is a good enough approximation in practice.} in
$(\Z/\ell\Z)^\times$, its order is exactly $\ell-1$ with probability
$(\ell-1)/\ell$, hence we can assume that asymptotically
$\order_\ell(q)\in O(\ell)=O(r\log(r))$. Similar considerations can be
made for the elliptic algorithm, assuming that $\ell\in
o(\sqrt[4]{q})$.
Finally, we must also take into account the possibility that the
elliptic algorithm fails. Under the heuristics of
Appendix~\ref{app:ellprdsdata}, this possibility only discards one in
$O(q^r)$ curves, and is thus negligible.

Summarizing, if we take $\bar{u},\bar{s}\in O(\log(r))$, we can expect
Algorithm~\ref{algorithm:selectell} to find suitable parameters for
the cyclotomic algorithm, leading to expected parameters $\ell\in
O(r\log(r))$, and to an expected running time of $\tildO(\sqrt{r})$
binary operations and $\tildO(r^{(\omega+1)/2}+\MM(r)\log(q))$ operations in
$\F_q$.  Similarly, if we also take $\bar{e}\in O(r\log(r))$,
assuming that $r\log(r)\in o(\sqrt[4]{q})$, we can expect
Algorithm~\ref{algorithm:selectell} to find suitable parameters for
the elliptic algorithm, leading to an expected running time of
$\tildO(\sqrt{r}+(\log(r))^2(\log(q))^6)$ binary operations and
$\tildO(r^2\log(q))$ operations in $\F_q$.

Although the complexity of the cyclotomic algorithm looks better, it
must not be neglected that the $\tildO$ notation
hides the cost of taking an auxiliary extension of degree
$O(\log(r))$; whereas the elliptic algorithm, when it applies, does
not incur such overhead. The impact of the hidden terms in the
complexity can be extremely important, as we will show in the next
section. 

The same considerations also apply when comparing Rains' algorithms to
Allombert's. Indeed, the latter performs extremely well when the
degree $s$ of the auxiliary extension is small, but becomes slower as
this degree increases.

In practice, it is hopeless to try and determine the appropriate
bounds for each algorithm from a purely theoretical point of view. The
best approach we can suggest, is to determine parameters at runtime,
and set bounds and thresholds experimentally.
To summarize, given parameters $q$ and $r$, we suggest
the following approach:
\begin{enumerate}
\item If $\gcd(q,r)\ne 1$, run the Artin--Schreier algorithm of
  Section~\ref{sec:fast-artin-schreier}.
\item If $r$ is a power of a small prime $v$, run the algorithm of
  Section~\ref{sec:fast-algor-large}.
\item Determine the order $s$ of $q$ in $(\Z/r\Z)^\times$. If it is
  small enough, run one of the variants of Allombert's algorithm
  presented in Section~\ref{sec:kummer}.
\item Run Algorithm~\ref{algorithm:selectell} with bounds $\bar{u}$,
  $\bar{s}$ and $\bar{e}$ determined according to $s$. Depending on
  the best parameters found by Algorithm~\ref{algorithm:selectell},
  run the best option among Rains' cyclotomic algorithm, Rains'
  elliptic algorithm, and Allombert's algorithm.
\end{enumerate}

In the next section we shall focus on the last two steps, by comparing
our implementations of the
algorithms involved, thus giving an estimate of the various thresholds
between them.  However we stress that these thresholds are bound to
vary depending on the implementation and the target platform, thus it
is the implementer responsibility to determine them at the moment of
configuring the system.

\begin{remark}
  Although our exposition focused on the case where $r$ is a prime
  power and $k\simeq K$, most of the algorithms presented here can be
  easily adapted to work more generally.

  In particular, it is a non-negligible practical improvement to work
  with composite $r$ by ``gluing'' together many prime powers. For
  example, let $r$ and $r'$ be two prime powers, and let $\ell$ and
  $\ell'$ be two primes selected for use in Rains' cyclotomic
  algorithm. If $\ell=\ell'$, then $\#\langle q\rangle=rr's$ for some
  $s$, and Rains' algorithm can be run only once for both $r$ and $r'$
  at the same time, with the added benefit of requiring a smaller
  auxiliary extension. Similarly, if $\ell\ne\ell'$, then we can run
  only once a straightforward generalization of Rains' algorithm using
  $(\ell\ell')$-th roots of unity; this is especially advantageous
  when $\order_\ell(q) = rs$ and $\order_{\ell'}(q) = r's$, so that
  the degree of the auxiliary extension is unchanged.

  In general, given the output of Algorithm~\ref{algorithm:selectell} for
  many prime powers, combining the results to obtain the best
  ``gluing'' requires solving an integer linear program (ILP). Given
  the availability of very fast ILP solvers, the practical speed-up
  can be significant.  Similar techniques also apply to Allombert's
  algorithm. On the other hand, it seems much more unlikely to apply
  them to the elliptic Rains' algorithm.
\end{remark}

\section{Experimental Results}
\label{sec:experimental-results}

To validate our results, we implemented the algorithms described in
the previous sections, and compared them to the implementation of
Allombert's algorithm available in PARI/GP~\cite{Pari}, and to that of
Rains' algorithm available in Magma~\cite{MAGMA}. %
The variants of Allombert's algorithm described in
Section~\ref{sec:fast-kummer} were implemented in C on top of the
Flint library~\cite{hart2010flint}. Rains' cyclotomic and elliptic
algorithms were implemented in Sage~\cite{Sage} (which itself uses
PARI and Flint to implement finite fields), with critical code
rewritten in C/Cython.  Our code only handles $q$ prime and $m,n$ odd.

We ran tests for a wide range of primes $q$ between $3$ and
$2^{60}+253$, and prime powers $r$ between $3$ and $2069$. All tests
were run on an
Intel(R) Xeon(R) CPU E5-4650 v2 clocked at 2.40GHz.
We report in Figure~\ref{fig:bench} statistics only on the runs for
$100<q<2^{20}$; other ranges show very similar trends. The source
code and the full datasets can be downloaded at
\url{https://github.com/defeo/ffisom}.

\begin{figure}
  \newlength{\mywidth}
  \setlength{\mywidth}{8cm}
  \centering

  \begin{subfigure}{.48\textwidth}
    \includegraphics[width=\mywidth]{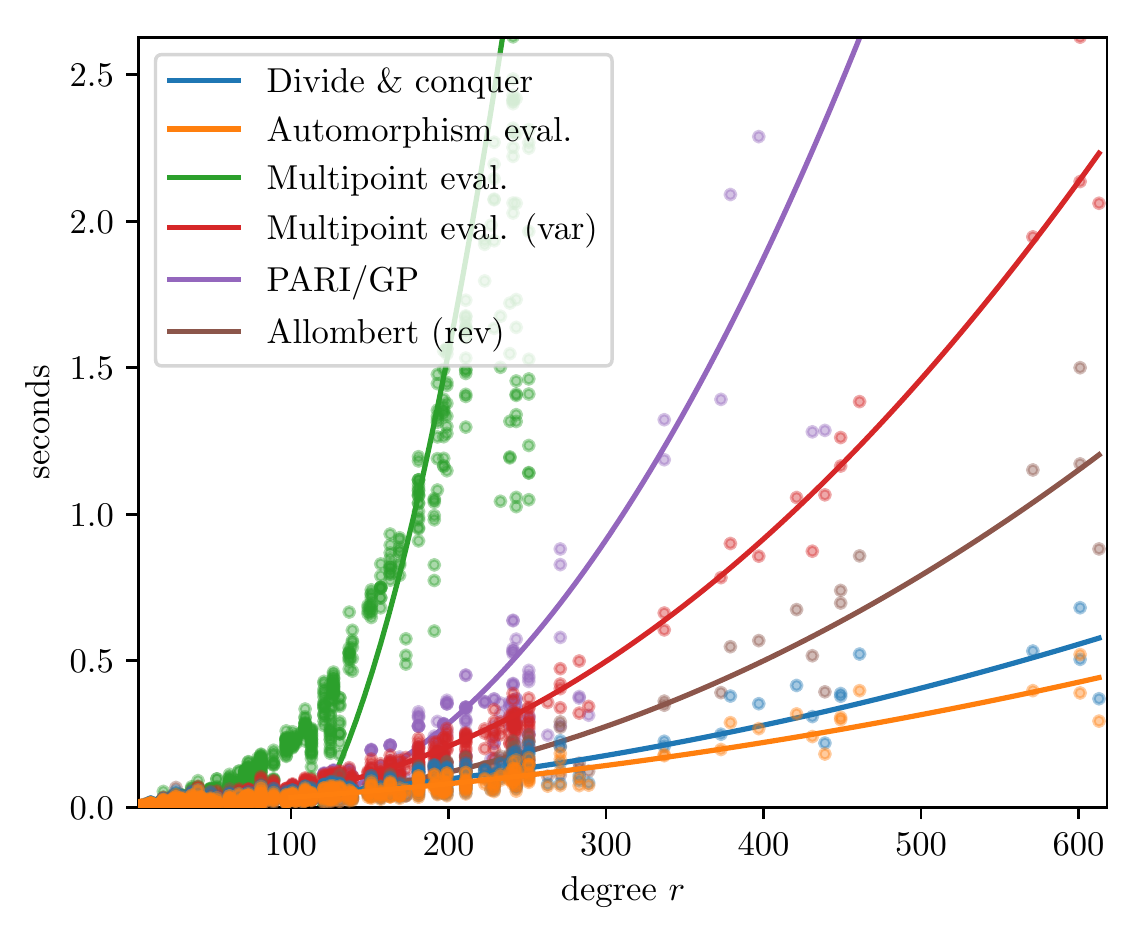}
    \caption{Comparison of various implementations of Allombert's
      algorithm, in the case where the auxiliary degree
      $s=\order_q(r)\le 10$.  Dots represent individual runs, lines
      represent degree 2 linear regressions.}
    \label{fig:bench:allombert-lowaux}
  \end{subfigure}
  \hfill
  \begin{subfigure}{.48\textwidth}
    \noindent
    \includegraphics[width=\mywidth]{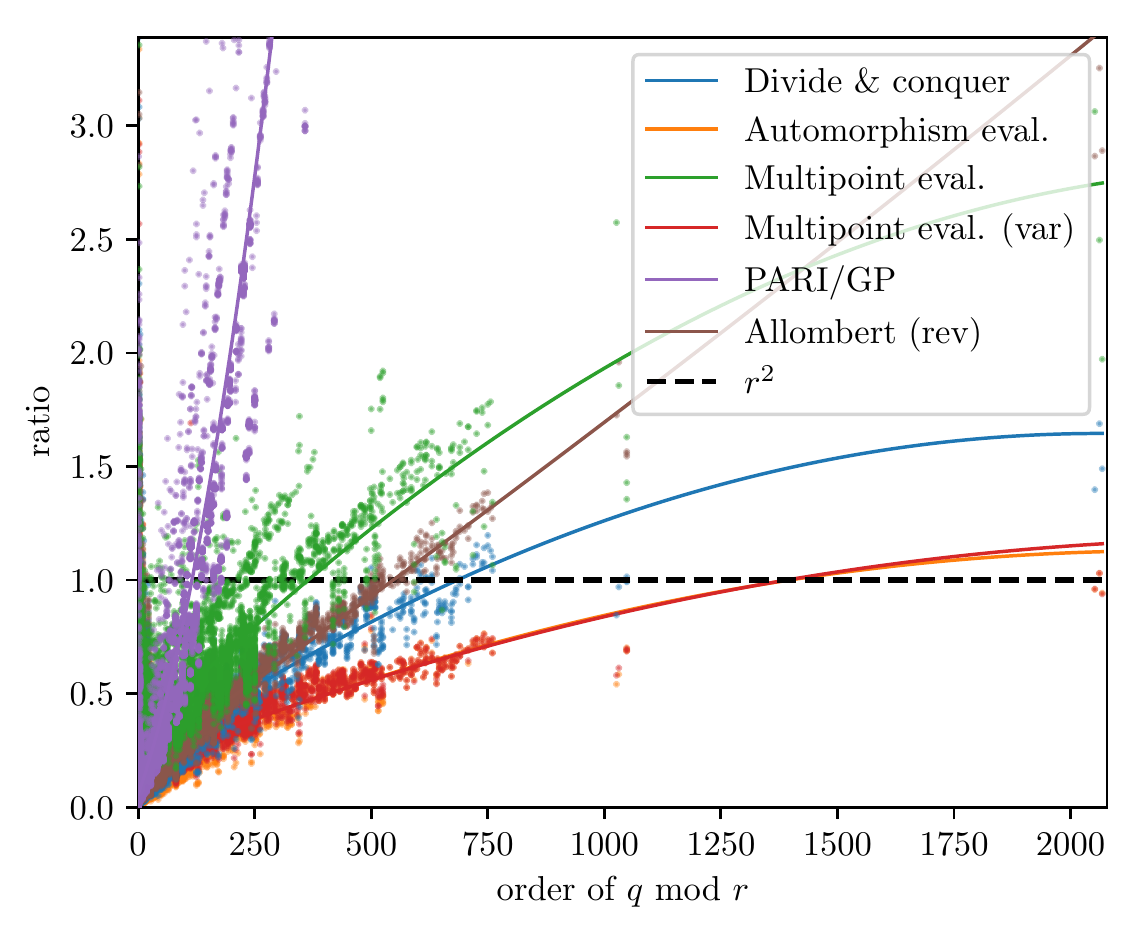}
    \caption{Comparison of various implementations of Allombert's
      algorithm, as a function of the auxiliary degree
      $s=\order_q(r)$.  Individual running times are scaled by down by
      $r^2$.  Dots represent individual runs, lines represent degree 2
      linear regressions.}
    \label{fig:bench:allombert-anyaux}
  \end{subfigure}

  \begin{subfigure}{.48\textwidth}
    \noindent
    \includegraphics[width=\mywidth]{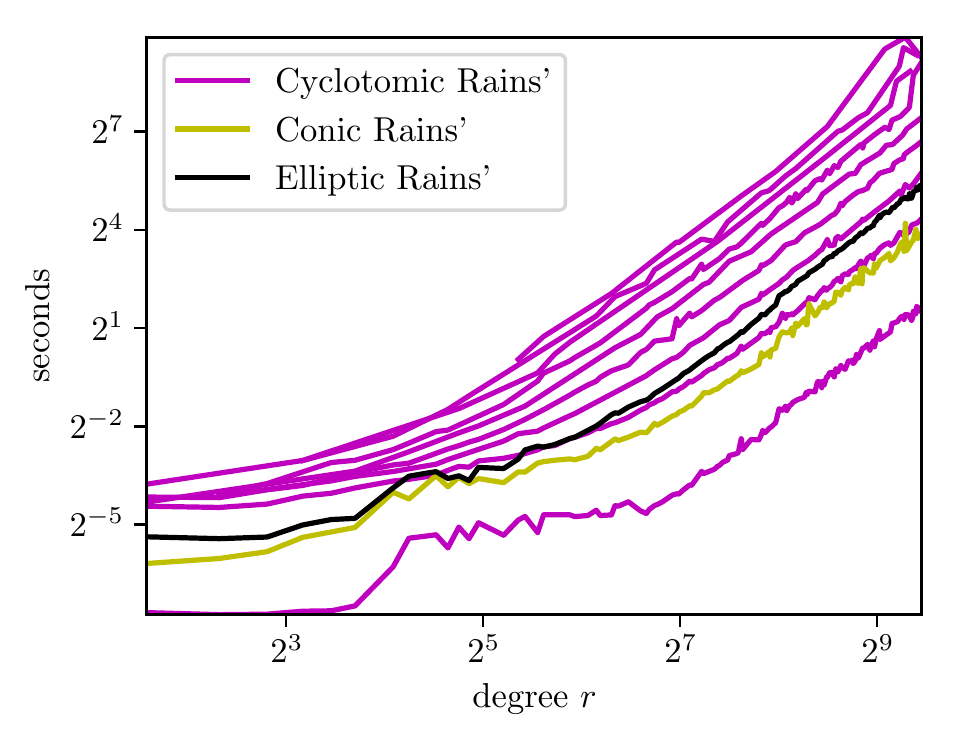}
    \caption{Cyclotomic, conic and elliptic variants of Rains'
      algorithm.  Auxiliary extension degrees $s$ for cyclotomic
      Rains' range between $1$ and $9$. Lines represent median times.}
    \label{fig:bench:rains}
  \end{subfigure}
  \hfill
  \begin{subfigure}{.48\textwidth}
    \noindent
    \includegraphics[width=\mywidth]{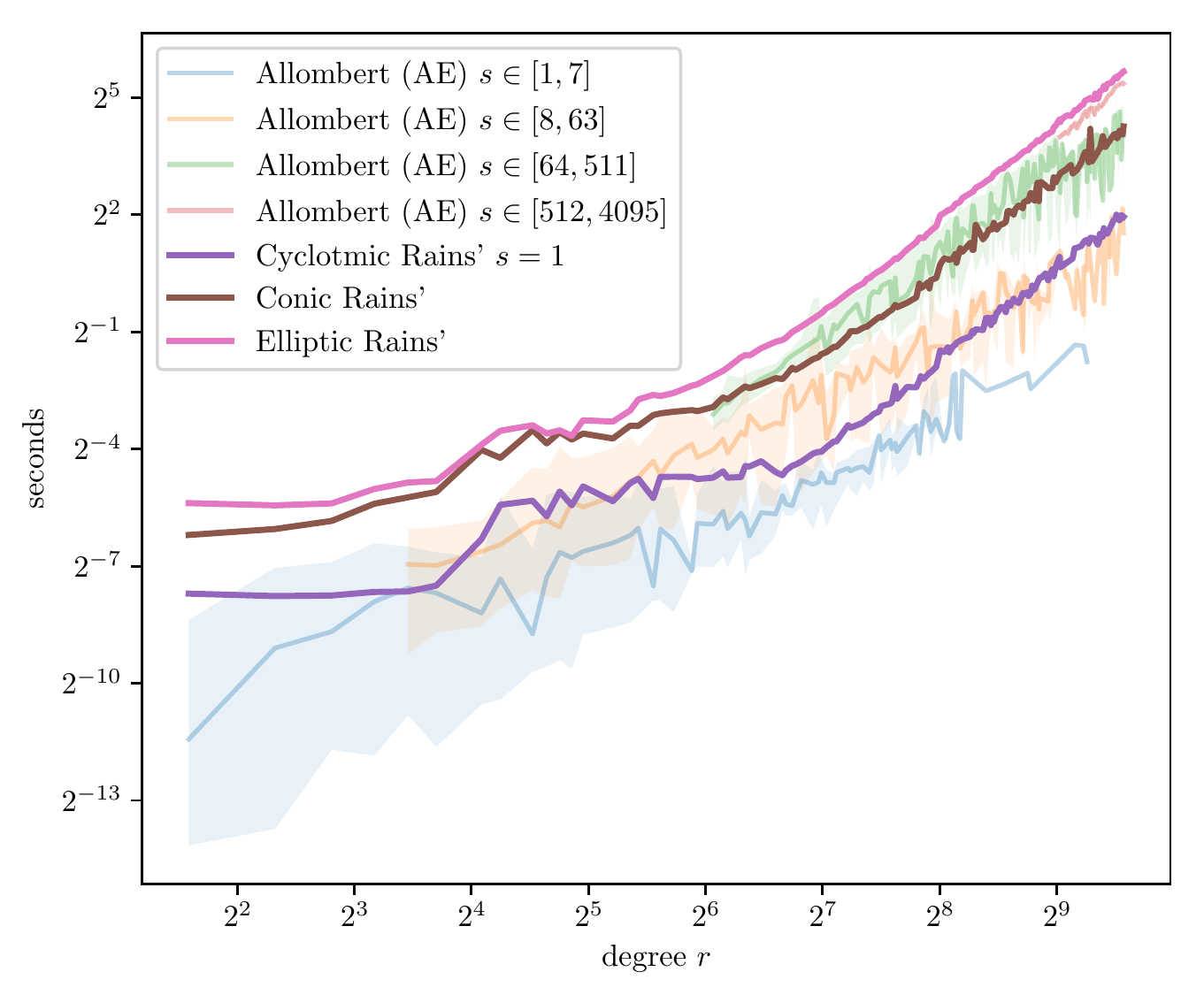}
    \caption{Comparison of Allombert's (Automorphism evaluation
      variant) and Rains' algorithms at some fixed auxiliary extension
      degrees $s$. Lines represent median times, shaded areas minimum
      and maximum times.}
    \label{fig:bench:all}
  \end{subfigure}

  \caption{Benchmarks for Rains' and Allombert's algorithms. $q$ is a
    prime between $100$ and $2^{20}$, $r$ is an odd prime power
    varying between $3$ and $2069$.  Plots~\subref{fig:bench:rains}
    and~\subref{fig:bench:all} are in doubly logarithmic scale. Full
    dataset available at \url{https://github.com/defeo/ffisom}.}
  \label{fig:bench}
\end{figure}

\begin{figure}
  \centering
  \includegraphics[width=8cm]{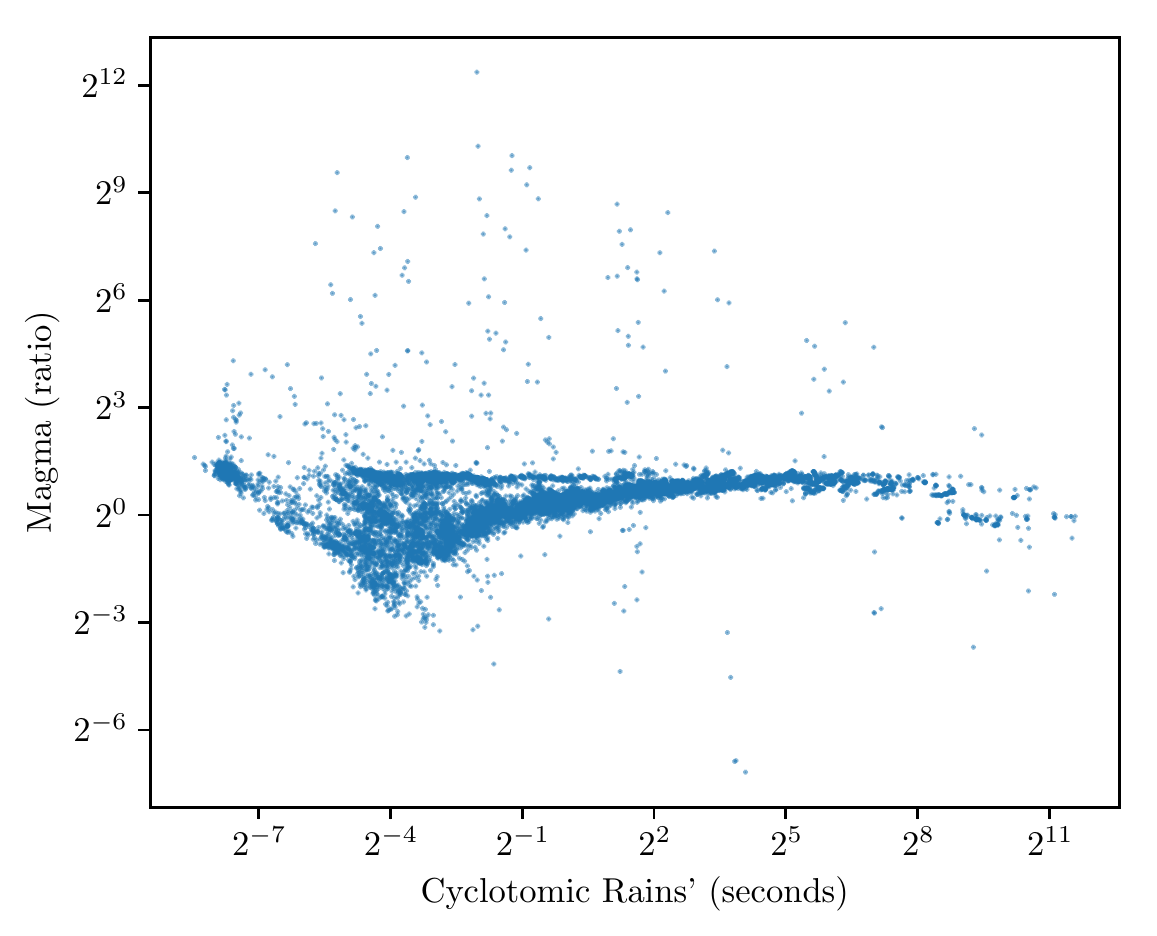}
  \caption{Comparison of our implementation of Rains' algorithm and
    Magma's. Running time of our implementation in seconds vs ratio of
    Magma running time over ours. Plot in doubly logarithmic scale.}
  \label{fig:bench:magma}
\end{figure}

We start by comparing our implementation of the three variants of
Allombert's algorithm presented in Section~\ref{sec:algor-rely-polyn}
with the original one in PARI. %
In Figure~\ref{fig:bench:allombert-lowaux} we plot running times
against the extension degree $r$, only for cases where the auxiliary
degree $s=\order_q(r)$ is at most $10$: dots represent individual
runs, continuous lines represent degree $2$ linear regressions. %
Analyzing the behavior for arbitrary auxiliary degree $s$ is more
challenging. %
Based on the observation that all variants have essentially quadratic
cost in $r$, in Figure~\ref{fig:bench:allombert-anyaux} we take
running times, we scale them down by $r^2$, and we plot them against
the auxiliary degree $s$. %

The first striking observation is the extremely poor performance of
PARI, especially as $s$ grows. %
To provide a fairer comparison, we re-implemented Allombert's revised
algorithm~\cite{Allombert02-rev}, as faithfully as possible, as
described in Section~\ref{sec:algor-rely-line}; this is the curve
labeled \emph{``Allombert (rev)''} in the graphs. %
For completeness we also implemented the Paterson-Stockmeyer variant
described previously; we do not plot it here, because it overlaps
almost perfectly with our \emph{``Divide \& conquer''} curve. %
Although our re-implementations are considerably faster than PARI, it
is apparent that Allombert's original algorithm does not behave as
well as our new variants.

Focusing now on our three new variants presented in
Section~\ref{sec:algor-rely-polyn}, one can't fail to notice that the
second one, named \emph{``Automorphism evaluation''}, beats the other
two by a great margin, both for small and large auxiliary degree. %
Although the \emph{``Multipoint evaluation''} approach is expected to
eventually beat the other variants as $s$ grows, the cross point seems
to be extremely far from the parameters we explored. %
However, we notice that the naive variant of \emph{``Multipoint
  evaluation''} not using the iterated Frobenius technique (labeled
\emph{``Multipoint evaluation (var)''} in the graphs), starts poorly,
then quickly catches \emph{``Automorphism evaluation''} as $s$ grows.

Now we shift to Rains' algorithm and its variants. %
In comparing our implementation with Magma's, discarding outliers, we
obtain a fairly consistent speed-up of about 30\% (see
Figure~\ref{fig:bench:magma}); hence we will compare these algorithms
only based on our timings. %
In Figure~\ref{fig:bench:rains} we group runs of the cyclotomic
algorithm by the degree $s$ of the auxiliary extension, and we plot
median times against the degree $r$; only the graphs for $s<10$ are
shown in the figure. %
We observe a very large gap between $s=1$ and larger $s$
($s=2$ is $8-16$ times slower). This is partly due to the fact that we
use generic Python code to construct auxiliary extensions, rather than
dedicated C; however, a large gap is unavoidable, due to the added cost
of computing in extension fields. %
We also plot median times for the elliptic variant and for the conic
variant (see Appendix~\ref{app:rains-vars}). %
It is apparent that the elliptic algorithm outperforms the cyclotomic
one as soon as $s\ge 3$, and that the conic algorithm conveniently
replaces the case $s=2$. %
Thus, at least for the parameter ranges we have tested, the cyclotomic
algorithm with auxiliary extensions seems of limited interest.

Finally, in Figure~\ref{fig:bench:all} we compare Rains' algorithms
against Allombert's. %
In light of the excellent performances of the \emph{``Automorphism
  evaluation''} variant of Allombert's algorithm, we only plot the
performances for this variant. %
We plot against the degree $r$ runs of Allombert's algorithm grouped
by ranges of the auxiliary degree $\order_r(q)$: we shade the area
between minimum and maximum running times, and trace the median
time. %
We also take from Figure~\ref{fig:bench:rains} the graphs for the
cyclotomic (only $s=1$), the conic and the elliptic variants of Rains'
algorithm. %
We notice that Allombert's algorithm, even with relatively large
auxiliary degrees, is extremely fast; the cyclotomic algorithm only
beats it when $\order_r(q)$ goes beyond 10 to 50, the conic algorithm
only beats extremely large $\order_r(q)$, and the elliptic algorithm
is never better. %
We also observe that Allombert's algorithm has a better asymptotic
behavior as the degree $r$ grows.

In light of these comparisons, it seems that the absolute winner is
our \emph{Automorphism evaluation} variant of Allombert's algorithm,
with Rains' cyclotomic algorithm being only occasionally more
interesting. %
Obviously, the comparisons are only relevant to our own code and test
conditions. Other implementations and benchmarks will likely find
slightly different cross-points for the algorithms.


\section{The Embedding Evaluation problem}
\label{sec:eval}

We end this work with a review of the known methods for
\emph{embedding evaluation}. %
We recall the problem statement: given two finite fields
$k=\F_q[X]/f(X)$ and $K=\F_q[Y]/g(Y)$, and given an embedding
$\phi:k\hookrightarrow K$ represented by two elements $\alpha\in k$
and $\beta\in K$ such that $k=\F_q(\alpha)$ and $\phi(\alpha)=\beta$,
answer the following questions:
\begin{itemize}
\item Given $\gamma\in k$, compute $\phi(\gamma)$;
\item Given $\delta\in K$, determine whether $\delta\in\phi(k)$;
\item Given $\delta\in\phi(k)$, compute
  $\phi^{-1}(\delta)$.\footnote{An additional natural question would
    be the following: given $\delta\in K$, compute an expression
    $\delta=\sum_{i=0}^{n/m-1} \phi(\gamma_i) Y^i$ with all
    $\gamma_i\in k$. The techniques presented here also apply to this
    more general problem, however we will skip them for conciseness.}
\end{itemize}
We are going to assume that elements of $k$ are represented on the
monomial basis $(1,X,\dots,X^{m-1})$, and elements of $K$ on the
monomial basis $(1,Y,\dots,Y^{n-1})$.

We review three solutions for this problem, each built on top
of the previous one. %
The first one uses basic linear algebra; it is simple and effective,
but has large space and time complexities. %
The second one improves the complexity by avoiding matrix inversion;
however it is still based on linear algebra, thus it has the same
storage requirements as the previous method. %
The last one replaces linear algebra with modular composition, thus
providing the best space and time complexities overall. %
Unlike the algorithms of the previous sections, all algorithms
presented in this section are deterministic, unless a randomized
algorithm is used for modular composition.

All the methods presented in this section are classic and well
understood, thus we will keep the presentation short, and will not
discuss implementation details.

\subsection{Linear algebra}
\label{sec:linear-algebra}

Since the map $\phi$ is $\F_q$-linear, one obvious solution is to
explicitly write its matrix on the monomial bases of $k$ and $K$. %
This is the solution employed by both the PARI/GP~\cite{Pari} and the
Magma~\cite{MAGMA,bosma+cannon+steel97} computer algebra systems. %
To stress the fact that $k$ is seen a vector space with a fixed basis
$(1,X,\dots,X^{m-1})$, we will write $V_X$ for it, and we will
similarly write $V_Y$ for $K$ with its monomial basis. %

The element $\alpha$ defines another basis
$(1,\alpha,\dots,\alpha^{m-1})$ of $k$, and we denote by $V_\alpha$
this vector space isomorphic to $V_X$. %
Similarly, $\beta$ defines a basis of a subspace $V_\beta\subset V_Y$,
also isomorphic to $V_\alpha$. %
Hence, we decompose the map $\phi:V_X\to V_Y$ as a composition of
three maps:
\[V_X \xrightarrow{\;\sim\;} V_\alpha \xrightarrow{\;\sim\;} V_\beta \lhook\joinrel\longrightarrow
  V_Y.\] %
The middle map from $V_\alpha$ to $V_\beta$ is trivially represented
by an identity matrix; the only maps that require actual computation
are the other two. %

For example, the map $V_\beta\hookrightarrow V_Y$ is represented by an
$n\times m$ matrix whose columns are the coefficients of the elements
$1,\beta,\dots,\beta^{m-1}$ written on the monomial basis of $V_Y$. %
The inverse map is then computed by solving a linear system; this
operation can be sped up by precomputing, e.g., an LU decomposition
for the matrix. %
The computation of the map $V_X\xrightarrow{\sim}V_\alpha$ is done
analogously.

Summarizing, both the full map $\phi:V_X\to V_Y$ and its inverse can
be computed by performing one matrix-vector product and solving one
linear system. %
The complexity is dominated by the cost of solving the linear system,
that is $O(m^{\omega-1}n)$ operations over $\F_q$. %
However this cost can be counted as a precomputation if we perform an
LU decomposition, and so can the computation of the powers $\alpha^i$
and $\beta^i$; in this case, the cost for evaluating $\phi$ or its
inverse drops down to $O(mn)$ operations. %
At any rate, the biggest drawback of this approach is the large memory
complexity: indeed storing the precomputed matrices requires $O(mn)$
elements of $\F_q$.

\subsection{Inverse maps and duality}

The first improvement to the linear algebraic method consists in
replacing the linear system solving with a much simpler matrix-vector
product, combined with an efficient change of basis. %
This technique is not
new~\cite{shoup94,shoup95,shoup99,bostan+salvy+schost03}, however it
is seldom found in the literature. %
We briefly recall it, following the presentation of~\cite{DeDoSc2014}.

Let $k=\F_q[X]/f(X)$ be a finite field with monomial basis
$(1,X,\dots,X^{m-1})$. %
The trace $\trace$ from $k$ to $\F_q$ defines a non-degenerate
bilinear form by
$\ang{\gamma,\delta}_k \equiv \trace(\gamma\delta)$,
which itself determines the \emph{dual basis}
$(X_0^*,X_1^*,\dots,X_{m-1}^*)$ to $(1,X,\dots,X^{m-1})$,
characterized by
\begin{equation*}
  \ang{X^j,X_i^*}_k = \begin{cases}
    1 &\text{if $i=j$,}\\
    0 &\text{otherwise.}
  \end{cases}
\end{equation*}
Given the polynomial $f$, conversions between the monomial and the
dual basis can be performed very efficiently at a cost of
$O(\MM(m)\log(m))$ operations in $\F_q$ (see~\cite[\S~3]{DeDoSc2014}).

Let now $K=\F_q[Y]/g(Y)$ be another finite field, let $\ang{,}_K$ be
the bilinear form defined by its trace to $\F_q$, and let
$\phi:k\hookrightarrow K$ be a field embedding. %
There exists a unique linear map $\phi^t:K\to k$, called the
\emph{dual map} of $\phi$, such that
\begin{equation*}
  \ang{\phi(\gamma),\delta}_K=\ang{\gamma,\phi^t(\delta)}_k
  \qquad\text{for any $\gamma\in k$ and $\delta\in K$.}
\end{equation*}
If $M=(m_{i,j})$ is the matrix of $\phi$ in the monomial bases of $k$
and $K$, then its transpose $M^t$ is the matrix of $\phi^t$ in their
\emph{dual bases}. %
Indeed,
\begin{equation}
  \label{eq:tellegen-matrix}
  m_{i,j} = \ang{\phi(X^j),Y_i^*} = \ang{X^j,\phi^t(Y_i^*)}.
\end{equation}

We are now going to show that $\phi^t$ is closely related to the
inverse map of $\phi$. %
If $\phi$ is an isomorphism of fields, then we immediately have
$\phi^t=\phi^{-1}$. %
Indeed, in this case $\phi$ preserves the bilinear forms:
\begin{equation*}
  \ang{\gamma,\delta}_k = \ang{\phi(\gamma),\phi(\delta)}_K
  \qquad\text{for any $\gamma,\delta\in k$.}
\end{equation*}
Hence, duality implies that
\begin{equation*}
  \ang{\gamma,\delta}_k = \ang{\gamma,\phi^t\circ\phi(\delta)}_k,
\end{equation*}
but then the non-degeneracy of the trace implies that
$\phi^t\circ\phi$ is the identity map.\footnote{More generally, in the
  category of finite-dimensional vector spaces with nondegenerate
  bilinear forms and morphisms that preserve bilinear forms, we have a
  natural isomorphism between the identity functor and the dual
  functor.}

The case where $\phi$ is a proper embedding calls for a more careful
handling. %
In this case, $\phi$ does not preserve traces, indeed, if $n=[K:\F_q]$,
\begin{equation*}
  \frac{n}{m}\ang{\gamma,\delta}_k = \ang{\phi(\gamma),\phi(\delta)}_K.
\end{equation*}
If $n/m$ is not a multiple of the characteristic, proceeding like
before we can show that $(m/n)\phi^t$ is the inverse of $\phi$. %
To handle the general case, take an element $\eta\in K$ such that
$\trace_{K/k}(\eta)=1$, and let $H$ be the map defined by
$\gamma\mapsto\eta\gamma$. %
Then, by composition of traces, we prove that
\begin{equation*}
  \ang{\gamma,\delta}_k = \trace_{k/\F_q}(\gamma\delta) =
  \trace_{K/\F_q}(\gamma\eta\delta) = \ang{\phi(\gamma),H\circ\phi(\delta)}_K =
  \ang{H\circ\phi(\gamma),\phi(\delta)}_K 
\end{equation*}
We have thus shown that both $\phi^t\circ H\circ\phi$ and
$\phi^t\circ H^t\circ\phi$ are the identity map.

Let us apply these findings to the linear algebraic approach of the
previous subsection. %
We had an embedding of fields, decomposed as three maps
\[V_X \xrightarrow{\;\sim\;} V_\alpha \xrightarrow{\;\sim\;} V_\beta
  \lhook\joinrel\longrightarrow V_Y.\] %
The maps $V_\alpha\xrightarrow{\sim} V_X$ and
$V_\beta\hookrightarrow V_Y$ are both field embeddings; the first one
is represented by its matrix on the monomial bases generated by
$\alpha$ and $X$; the second one on those generated by $\beta$ and
$Y$. %
For both maps, we are interested in computing their inverse; instead
of solving a linear system, we switch to dual bases and apply the
discussion above. %
Thus, the inverse of $V_\beta\hookrightarrow V_Y$ is evaluated as a
multiplication by a fixed element $\eta$, followed by a conversion to
the dual basis of $V_Y$, then a matrix-vector product with the
transposed matrix, and finally a conversion back to the monomial basis
of $V_\beta$. %
The inverse of $V_\alpha\xrightarrow{\sim}V_X$ is computed similarly. %

Note that the inverse map to $V_\beta\hookrightarrow V_Y$ is not
everywhere defined. %
Interestingly, while linear system solving could immediately recognize
elements of $V_Y$ that are not in the image of $V_\beta$, the new
solution will just project them onto an arbitrary element of
$V_\beta$. %
Indeed, any $\delta\in V_Y$ can be rewritten as
\begin{equation*}
  \delta = (\delta - \trace_{K/k}(\eta\delta)) + \trace_{K/k}(\eta\delta)
  = \delta' + \trace_{K/k}(\eta\delta),
\end{equation*}
for the same $\eta$ chosen above. %
Direct calculation shows that
\begin{equation}
  \label{eq:projection}
  \ang{\phi(\gamma),\eta\delta}_K =
  \ang{\phi(\gamma),\eta\delta'}_K + \ang{\phi(\gamma),\eta\trace_{K/k}(\eta\delta)}_K =
  \ang{\gamma, \trace_{K/k}(\eta\delta)}_k,
\end{equation}
for any $\gamma\in k$. %
Hence, applying the above algorithm to an arbitrary element
$\delta\in V_Y$ yields the element $\trace_{K/k}(\eta\delta)$ of
$V_\beta$, which coincides with $\delta$ whenever
$\delta\in\phi(k)$. %
Hence, the best way to test that an element $\delta\in K$ is in the
image of $\phi$ would be to project it to
$\gamma=\trace_{K/k}(\eta\delta)$ using this algorithm, and then test
that $\phi(\gamma)=\delta$.

It is easily seen that the complexity is dominated by the transposed
matrix-vector product, which costs $O(mn)$ operations (plus a cost of
$O(m\MM(n))$ operations for precomputing the matrices). %
Hence, by paying a little overhead in the changes of basis, we have
completely removed the cost of solving a linear system. %
We have not reduced yet the large storage cost, however.

\subsection{Modular composition}

Our final improvement consists in replacing the matrix computations
with modular composition. %
The technique originates in Shoup's
work~\cite{shoup94,shoup95,shoup99}.

Considering again the map $V_\beta\hookrightarrow V_Y$, we observe
that its evaluation is precisely a modular composition problem: given
polynomials $\gamma=\sum \gamma_i\beta^i$ and $\beta=\sum \beta_iY^i$
with coefficients in $\F_q$, of degree bounded by $n$, compute
$\gamma(\beta)\bmod g$. %
As seen previously this computation can be done more efficiently by a
dedicated algorithm, than by a naive Horner rule.

However we also need to compute the inverse map to
$V_\beta\hookrightarrow V_Y$, and this problem is clearly not a
modular composition one. %
In the previous subsection we have reduced the computation of this
inverse map to a change of bases, combined with a transposed
matrix-vector product. %
A very powerful generalization of Eq.~\eqref{eq:tellegen-matrix},
called \emph{transposition principle}, allows us to \emph{transpose}
any modular composition algorithm, much like one would transpose a
matrix. %
This technique was also introduced by Shoup, and then refined by many
authors~\cite{bostan+lecerf+schost:tellegen,df+schost10,df+thesis}. %
The \emph{dual problem} to modular composition was named \emph{power
  projection} by Shoup; its inputs are the polynomials $\beta,g$, and
an element $\gamma^*$ in the dual space of $\F_q[X]$ (i.e., the linear
forms on $\F_q[X]$); its output is the list of elements
$\gamma^*(\beta^0),\dots,\gamma^*(\beta^{n-1})$. %
Thanks to the transposition principle, the power projection problem
can be solved within the same complexity bound as modular
composition~\cite{shoup95,kedlaya+umans08}.

Summarizing, the inverse map to $V_\beta\hookrightarrow V_Y$, is
computed by Algorithm~\ref{alg:inv-mod-comp}.

\begin{algorithm}
    [Inverse embedding]
    \label{alg:inv-mod-comp}
    \begin{algorithmic}[1]
      \REQUIRE An element $\delta\in K$, and precomputed values:\\
      $\bullet\;\beta\in K$ generating a subfield isomorphic to $k$,\\
      $\bullet\;\eta\in K$ such that $\trace_{K/k}\eta=1$.
      \ENSURE $\trace_{K/k}(\eta\delta)$ written in the basis $(1,\beta,\dots,\beta^{m-1})$.
      \STATE Compute the minimal polynomial of $\beta$ over $\F_q$;
      \STATE Compute $\delta'=\eta\delta$;
      \STATE Convert $\delta'$ to the dual basis $(Y_0^*,\dots,Y_{n-1}^*)$;
      \STATE Compute $\gamma=\trace_{K/k}\delta'$ using \emph{power projection};
      \STATE Convert $\gamma$ to the monomial basis $(1,\beta,\dots,\beta^{m-1})$;
      \RETURN $\gamma$.
    \end{algorithmic}
\end{algorithm}

\begin{theorem}
  Algorithm~\ref{alg:inv-mod-comp} is correct. %
  When the input $\delta$ is in the image of $k$, it returns $\delta$
  itself written on the basis $(1,\beta,\dots,\beta^{m-1})$. %
  It computes its output using $O(n^{(\omega+1)/2})$ operations in $\F_q$
  in the worst case.
\end{theorem}
\begin{proof}
  Correctness follows from the discussion above, and the obvious fact
  that $\delta=\trace_{K/k}(\eta\delta)$ whenever $\delta\in k$.

  The minimal polynomial of $\beta$, required to compute conversions
  between the monomial and the dual bases generated by $\beta$, can be
  computed in $O(\MM(n)\log(n))$ operations using the Berlekamp--Massey
  algorithm.\footnote{The minimal polynomial could be precomputed
    along with $\beta$, however we include this computation in the
    algorithm, as it does not change the total complexity.} %
  Conversions between monomial and dual bases are also done in
  $O(\MM(n)\log(n))$ as shown in~\cite[\S~3]{DeDoSc2014}. %
  Finally, the power projection costs $O(n^{(\omega+1)/2})$ operations, using
  any of the algorithms in~\cite{shoup95,kedlaya+umans08}.
\end{proof}

We have not specified how the element $\eta$ is computed. %
If $n/m$ is not divisible by the characteristic, then one can simply
take $\eta=m/n$. %
In the general case, it suffices to know an element such that
$\trace_{K/k}\ne0$, and to divide it by its trace. %
If the $(n-1)$-th coefficient of $g$ is not $0$, then $Y$ is one such
element; otherwise we take elements at random until a suitable one is
found: only $O(1)$ trials are expected on average. %
In any case, computing one trace can be done using
$O(n^{(\omega+1)/2}\log(n)+\MM(n)\log(q))$ operations, thanks to
Section~\ref{sec:fundamentalgo}.

\begin{corollary}
  After a precomputation costing $O(n^{(\omega+1)/2}\log(n)+\MM(n)\log(q))$
  operations in $\F_q$ on average, all sub-questions of the
  \emph{embedding evaluation problem} can be answered using
  $O(n^{(\omega+1)/2})$ operations in $\F_q$.
\end{corollary}

\appendix
\part*{Appendices}
\addcontentsline{toc}{part}{Appendices}


\section{Rain's conic algorithm}
\label{app:rains-vars}

We have seen that Rains' cyclotomic algorithm suffers in practice from
the need to build a field extension $k'$ of $k$. %
The conic variant we are going to present reduces the degree of the
field extension from $s=[k':k]$ to $s/2$ whenever $s$ is even. %
This is especially useful when $s=2$, as highlighted in
Section~\ref{sec:experimental-results}. %
The algorithm is similar in spirit to Williams' $p+1$ factoring
method~\cite{williams1982}, where the arithmetic of the norm $1$
subgroup of ${k'}^*$ is performed using Lucas sequences on a subfield
of index $2$ of $k'$.

Let $\F$ be a finite field of odd characteristic, let $\Delta\in\F$ be
a quadratic non-residue, let $\delta$ be an element of the algebraic
closure of $\F$ such that $\delta^2=\Delta$, and define the norm $1$
subgroup of $\F[\delta]^*$ as
\[T_2(\F) = \{(x+\delta y)/2 \;\mid\; x,y\in\F \text{ and } x^2-\Delta
  y^2 = 4\};\] %
it is easy to verify that $T_2(\F)$ forms a group under
multiplication. %
If we see the elements $(x+\delta y)/2$ as points $(x,y)$ on a conic
$x^2-\Delta y^2=4$, the group law of $T_2(\F)$ induces a group law on
the conic. %
By projecting onto the $x$-coordinate, a straightforward calculation
shows that, for any point $(\theta,*)$ on the conic, its $n$-th power
has coordinates $(\theta_n,*)$, where $\theta_n$ is defined by the
Lucas sequence
\[\theta_0 = 2, \quad \theta_1 = \theta, \quad \theta_{i+1}=\theta\theta_i-\theta_{i-1}.\] %
We shall denote by $[n]$ the map $\theta\mapsto\theta_n$; notice how
it does not depend on the choice of $\Delta$.

The generalization of Rains' algorithm is now obvious: by projecting
on the $x$-coordinate, we work in a field extension twice as small
compared to the original algorithm. %
This is summarized in Algorithm~\ref{algorithm:rains-conic}.

\begin{algorithm}[Rains' conic algorithm]
  \label{algorithm:rains-conic}
  \begin{algorithmic}[1]
    \REQUIRE A field extension $k/\F_q$ of degree $r$; a prime $\ell$
    such that
    \begin{itemize}
    \item $(\Z/\ell\Z)^\times = \langle q\rangle \times S$ for some $S$,
    \item $\#\langle q\rangle = 2rs$ for some integer $s$;
    \end{itemize}
    a polynomial $h$ of degree $s$ irreducible over $k$.
    \ENSURE A normal generator of $k$ over $\F_q$,
    with a uniquely defined Galois orbit.
    
    \STATE Construct the field extension $k'=k[Z]/h(Z)$;
    \REPEAT
    \REPEAT
    \STATE Take a random element $\theta\in k'$,
    \UNTIL\label{algorithm:rains-conic:sqtest} $\theta^2-4$ is a quadratic non-residue;
    \STATE\label{algorithm:rains-conic:power} Compute $\zeta=[(\#k'+1)/\ell]\theta$,
    \UNTIL $\zeta\ne2$;
    \STATE\label{algorithm:rains-conic:period} Compute $\eta(\zeta) \leftarrow \sum_{\sigma\in S}[\sigma]\zeta$;
    \RETURN\label{algorithm:rains-conic:trace} $\alpha \leftarrow \trace_{k'/k}\eta(\zeta) = \sum_{i=0}^{s-1}[q^{ri}]\eta(\zeta)$.
  \end{algorithmic}
\end{algorithm}

\begin{proposition}
  Algorithm~\ref{algorithm:rains-conic} is correct: on input
  $q,r,\ell,s$ it returns an element in the same Galois orbit as
  Algorithm~\ref{algorithm:rains-cyclo} on input $q,r,\ell,2s$. %
  It computes its output using $O(\MM(sr)(sr\log(q)+(\ell/r)\log(\ell)))$
  operations in $\F_q$ on average, or $\tildO((sr)^2\log(q))$ assuming
  $\ell\in o(sr^2)$.
\end{proposition}
\begin{proof}
  By construction, all the $\ell$-th roots of unity are in
  $T_2(k')$. %
  Observe that if $(x+\delta y)/2$ is in $T_2(k')$, then its trace
  over $k'$ is equal to $x$. %
  Hence, the value $\zeta$ computed in
  Step~\ref{algorithm:rains-conic:power} is the trace over $k'$ of a
  primitive $\ell$-th root of unity. %
  We conclude by comparing this algorithm with
  Algorithm~\ref{algorithm:rains-cyclo}.

  The non-residuosity test in Step~\ref{algorithm:rains-conic:sqtest}
  is done by verifying that the $(\#k'-1)/2$-th power of $\theta$ is
  equal to $-1$. %
  We do this in $O(sr\log(q))$ operations in $k'$, or
  $O(sr\MM(sr)\log(q))$ operations in $\F_q$.

  To implement the other steps, we need to evaluate the map $[n]$
  efficiently. %
  We have the following classical relationships for the Lucas sequence
  of $\theta$:
  \begin{equation*}
    \theta_{2i} = \theta_{i}^2-2,\quad
    \theta_{2i+1} = \theta_i\theta_{i+1} - \theta,\quad
    \theta_{2i+2} = \theta_{i+1}^2-2.
  \end{equation*}
  Starting with $\theta_0=2$ and $\theta_1=\theta$, we use a binary
  scheme to deduce $\theta_i,\theta_{i+1}$ from
  $\theta_{\lfloor i/2\rfloor},\theta_{\lfloor i/2\rfloor+1}$. %
  We reach $\theta_n$ after $O(\log(n))$ steps, each requiring a
  constant number of operations in $k'$.

  Hence, Step~\ref{algorithm:rains-conic:power} costs
  $O(sr\MM(sr)\log(q))$ operations in $\F_q$, while
  Steps~\ref{algorithm:rains-conic:period}
  and~\ref{algorithm:rains-conic:trace} together cost
  $O((\MM(sr)(\ell/r)\log(\ell))$.
\end{proof}

Although this variant does not exploit the asymptotic improvement
offered by Proposition~\ref{prop:trace-like}, the fact that its
auxiliary degree $s$ is half the one of the original algorithm usually
gives an interesting practical improvement. %
Step~\ref{algorithm:rains-conic:power} can be modified so as to avoid
the premature projection on the $x$-axis, so that the algorithms of
Proposition~\ref{prop:trace-like} apply. %
We leave the details of this variant to the reader.


\section{Using $j=0,1728$ in the elliptic Rains' algorithm}
\label{app:elliptic-curves}

When we defined elliptic periods in Section~\ref{sec:ellperiods}, we
explicitly ruled out the case where the elliptic curve has
$j$-invariant $0$ and $1728$. %
Indeed the definition of elliptic periods for these curves is
complicated by the fact that they have additional automorphisms: had
we applied Definition~\ref{definition:ellperiod} to them, we would
have obtained periods that are always equal to $0$.

In this section we sketch how to extend the elliptic variant of Rains'
algorithm to these curves. %
Although this does not change the overall complexity of Rains'
algorithm, it makes for a small practical improvement, and a nice
\emph{recreational mathematics} read.

Recall (see~\cite[III.10]{Sil}) that the order of the
\emph{automorphisms group} $\Aut(E)$ of a curve $E$ defined over a
field $\F_q$ of characteristic $p$ is one of the following:
\begin{itemize}
\item $2$ if $j(E)\ne0,1728$,
\item $4$ if $j(E) = 1728$ and $p\ne2,3$,
\item $6$ if $j(E) = 0$ and $p\ne2,3$,
\item $12$ if $j(E) = 0 = 1728$ and $p=3$,
\item $24$ if $j(E) = 0 = 1728$ and $p=2$.
\end{itemize}
We can now give a meaningful definition of elliptic periods that
covers all elliptic curves.

\begin{definition}
  \label{definition:ellperiod-general}
  Let $E/\F_q$ be an elliptic curve with automorphism group $\Aut(E)$
  of order $2n$, and Frobenius endomorphism $\pi$. %
  Let $\ell > 3$ be an Elkies prime for $E$, $\lambda$ an eigenvalue
  of $\pi$, and $P$ a point of order $\ell$ in the eigenspace
  corresponding to $\lambda$ (i.e., such that $\pi(P)=\lambda P$). %
  Suppose that there is a subgroup $S$ of $(\Z/\ell\Z)^{\times}$
  containing $\Aut(E)$, and such that
  \begin{equation*}
    (\Z/\ell\Z)^{\times} = \langle{\lambda}\rangle \times S.
  \end{equation*}
  
  Then we define an elliptic period as
  \begin{equation*}
    \eta_{\lambda,S}(P) =
    \sum_{\sigma\in S/\Aut(E)} {x\left([\sigma] P \right)^n} 
  \end{equation*}
  where $x(P)$ denotes the abscissa of $P$.
\end{definition}

This definition is equivalent to Definition~\ref{definition:ellperiod}
when $\Aut(E)=\{\pm1\}$. %
When the automorphism group is larger, it avoids unnecessary
cancellations by quotienting out $S$ by $\Aut(E)$, and at the same
time it ensures unicity thanks the the $n$-th power in the sum. %
We leave as an exercise the proof of the statement analogous to
Lemma~\ref{lemma:ellperiods-order}. %
Note that it would be possible to generalize this definition to the
case where $\Aut(E)$ is (partially) contained in
$\langle{\lambda}\rangle$, however we leave out this detail, as it is
not necessary for Rains' algorithm.

\subsection{The ordinary case}

The curves of $j$-invariant $0$ and $1728$ are ordinary if and only if
all automorphisms are defined over $\F_p$. %
For $j=0$, this is equivalent to $p\equiv 1 \bmod 3$; for $j=1728$,
this is equivalent to $p\equiv 1 \bmod 4$. %

The splitting of the Frobenius polynomial of $E$ is well known in this
case, with early results dating back to Gauss. %
The two statements below follow easily
from~\cite[Th.~2.5,2.6]{silverberg2010group}.

\begin{proposition}
  Let $p$ be a prime congruent to $1$ modulo $3$, and let $q=p^d$. %
  Let $p=\pi\bar\pi$ be the unique decomposition of $p$ in the
  Eisenstein integers $\Z[\omega]$ with
  $\pi\equiv\bar\pi\equiv1\bmod3$. %
  Let $E$ be the curve defined by $y^2=x^3+b$, and let
  $\left(\tfrac{4b}{\pi}\right)_6$ be the unique sixth root of unity
  of $\Z[\omega]$ congruent to $(4b)^{(q-1)/6}\bmod \pi$. %
  Then the minimal polynomial of the Frobenius endomorphism of $E$
  splits in $\Z[\omega]$ as
  \begin{equation*}
    X^2 - tX + q =
    \Bigl(X - \left(\tfrac{4b}{\pi}\right)_6^{-1}\pi^d\Bigr)
    \Bigl(X - \left(\tfrac{4b}{\pi}\right)_6\bar\pi^d\Bigr).
  \end{equation*}
\end{proposition}

\begin{proposition}
  Let $p$ be a prime congruent to $1$ modulo $4$, and let $q=p^d$. %
  Let $p=\pi\bar\pi$ be the unique decomposition of $p$ in the
  Gaussian integers $\Z[i]$ with
  $\pi\equiv\bar\pi\equiv 1 \bmod (2+2i)$. %
  Let $E$ be the curve defined by $y^2=x^3-ax$, and let
  $\left(\tfrac{a}{\pi}\right)_4$ be the unique fourth root of unity
  of $\Z[i]$ congruent to $a^{(q-1)/4}\bmod \pi$. %
  Then the minimal polynomial of the Frobenius endomorphism of $E$
  splits in $\Z[i]$ as
  \begin{equation*}
    X^2 - tX + q =
    \Bigl(X - \left(\tfrac{a}{\pi}\right)_4^{-1}\pi^d\Bigr)
    \Bigl(X -  \left(\tfrac{a}{\pi}\right)_4\bar\pi^d\Bigr).
  \end{equation*}
\end{proposition}

From the above decomposition, it is immediately apparent that $\ell$
is an Elkies prime for $y^2=x^3+b$ if and only if
$\ell\equiv 1 \bmod 3$; and, similarly, it is an Elkies prime for
$y^2=x^3-ax$ if and only if $\ell\equiv 1 \bmod 4$. %

Algorithmically, we start by computing the splitting of $p$ in
$\Z[\omega]$ or $\Z[i]$ using Cornacchia's
algorithm~\cite{cornacchia1908di}; then, we deduce the image of the
eigenvalues in $(\Z/\ell\Z)^\times$ directly from the splitting. %
We give below more detailed step-by-step instructions for the case
$j=1728$; the case $j=0$ is analogous and can be treated in a similar
way.

Our goal is to find a uniquely defined generator for $\F_{q^r}$, with
$q=p^d$, given a prime $\ell$ such that $r|(\ell-1)$, assuming that
$q^r\ne 1\bmod\ell$ and that $p\equiv\ell\equiv 1 \bmod 4$:
\begin{enumerate}
\item Use Cornacchia's algorithm to compute a splitting $p=x^2+4y^2$;
\item Choose signs so that $p=\pi\bar\pi=(x-2iy)(x+2iy)$ and
  $x - 2y \equiv 1 \bmod 4$;
\item Find a primitive fourth root of unity $\hat\imath\in\F_p$ such that
  $x=2\hat\imath y$;
\item Find a primitive fourth root of unity $I\in\Z/\ell\Z$;
\item Look for an element of $(\Z/\ell\Z)^\times$ of order $r$ among
  the $I^{\pm c}(x\pm2Iy)^d$ for $0\le c \le 3$;
\item If such an element is found return the curve $y^2=x^3-ax$, where
  $a\in\F_q$ is such that $a^{(q-1)/4}=\hat\imath^c$.
\end{enumerate}

This procedure can be plugged inside
Algorithm~\ref{algorithm:selectell} to help find candidates for Rains'
algorithm. %
The most computationally intensive step is Cornacchia's algorithm,
which is negligible if compared to the point counting operations
needed for general elliptic curves. %
If an elliptic curve is found by this procedure, then it can be
readily used inside the elliptic variant of Rains' algorithm, by
replacing the earlier definition of elliptic period with
Definition~\ref{definition:ellperiod-general} above.

\subsection{The supersingular case}

The supersingular case is of much lesser interest, however we treat it
briefly for completeness. %
To keep things simple, we will not consider the cases $p=2,3$. %
The possibilities for the trace of $E$ are much more constrained in
this case. %
Indeed, the only possibilities are $t=0,\pm\sqrt{q},\pm2\sqrt{q}$
(see~\cite{waterhouse69}). %
Of these, only the case $\pm\sqrt{q}$ is interesting in the context of
Rains' algorithm; indeed when $t=\pm2\sqrt{q}$ the two eigenvalues are
equal, and Rains' algorithm does not apply; when $t=0$ the two
eigenvalues are opposite, thus if one has odd order $r$ modulo $\ell$,
$q$ must have order $2r$, but in this case the conic variant presented
in Appendix~\ref{app:rains-vars} is simpler and more efficient. %
The case $t=\pm\sqrt{q}$ is never realized by $j=1728$, however for
$j=0$ it constitutes an interesting optimization whenever $q$ is a
square and $\order_\ell(q)=3r$. %

Assume that $q$ is an even power of $p$, and $p\equiv 2\mod 3$. %
Let $j\in\F_q$ be a primitive cube root of unity, and let $E$ be the
curve $y^2=x^3+j$. %
Then $E$ has trace $\sqrt{q}$, and the minimal polynomial of the
Frobenius endomorphism splits in $\Z[\omega]$ as
\begin{equation*}
  X^2 - X\sqrt{q} + q = (X+\omega\sqrt{q})(X+\omega^2\sqrt{q}).
\end{equation*}

Let $r$ be a prime power not divisible by $2$ or $3$, and let $\ell$
be a prime such that $\ell=1\bmod 3$ and $\order_\ell(q)=3r$. %
Let $\hat\omega\in\Z/\ell\Z$ such that $q^r=\hat\omega\bmod\ell$, then
by direct calculation we see that
$\order_\ell(-\hat\omega^r\sqrt{q})=r$ and
$\order_\ell(-\hat\omega^{-r}\sqrt{q})=3r$. %
Hence, there is only one cyclic rational subgroup of $\ell$-torsion in
$\F_{q^r}$, and we are in the conditions to apply Rains' algorithm. %
This computation also shows that the case $t=-\sqrt{q}$ is not useful
in Rains' algorithm.

\begin{example}
  We illustrate the algorithm sketched above with a numerical
  example. %
  Let $p=5$, $q=25$, $r=5$ and $\ell=61$. %
  In what follows, elements of $\F_q$ are represented as polynomials
  in $j$ in the quotient ring $\F_5[j]/(j^2+j+1)$. %

  The elliptic curve $y^2=x^3+j$ has trace $5$, and its two Frobenius
  eigenvalues are $-\omega p$ and $-\omega^2 p$. %
  We set $\hat\omega=q^r=13\bmod 61$, and we verify that
  $\order(-\hat\omega^r p)=\order(-13^2\cdot 5)=r$, whereas
  $\order(-\hat\omega^{-r} p)=\order(-13\cdot 5)=3r$. %
  Indeed $\#E(\F_{q^r})= 3\cdot7\cdot61\cdot7621$, hence we can
  construct an elliptic period from any point $P\in E[61]$ defined
  over $\F_{q^r}$. %
  To apply Definition~\ref{definition:ellperiod-general}, we need a
  generator of the order $12$ subgroup of $\Z/61\Z$, e.g., $21$. %
  Quotienting out the order $6$ subgroup, we obtain the period
  \begin{equation*}
    \eta = x(P)^3 + x([21]P)^3,
  \end{equation*}
  and we verify that its minimal polynomial is
  \begin{equation*}
    x^5 + (3j-3)x^4 + (j+3)x^3 -jx^2 + (j+1)x + (2j-2)    
  \end{equation*}
  independently of the initial choice of $P$.
\end{example}

\section{Elliptic periods}
\label{app:ellprdsdata}

In this section, we cover additional topics related to the correctness
of the elliptic variant of Rains' algorithm. %
Specifically, we first prove a sufficient condition for an elliptic
period to be a normal element. %
This, together with heuristic assumptions, gives a lower bound for the
success probability of Algorithm~\ref{algorithm:compell}. %

Secondly, we review our attempts at finding a counterexample to
Conjecture~\ref{conj:ellperiods}, i.e.\ an elliptic point whose period
is defined over a smaller field than its abscissa. %
Such examples are not expected to be common, but our extensive
searches have produced absolutely none, giving a small argument in
favor of the conjecture.

\subsection{Normality of elliptic periods}

We prove here an analogue of Lemma~\ref{th:gaussian} for elliptic
periods. %
All well known proofs of that lemma follow a common scheme: for a
squarefree $\ell$ and $q$ prime,
\begin{enumerate}
\item the $\ell$-th cyclotomic polynomial is irreducible over $\Q$,
  and its roots form a normal basis of the splitting field;
\item Gaussian periods, seen as field traces of roots of unity,
then yield normal bases of subextensions;
\item for any inert prime $q \neq \ell$, integrality properties
  guarantee that the basis stays normal after reduction modulo $q$.
\end{enumerate}

In the elliptic setting, nothing guarantees that torsion subgroups
yield normal bases, thus the first step above crucially fails; and
indeed we have seen examples where elliptic periods are not normal. %
However, at least when the torsion subgroup satisfies some normality
conditions, we may hope to carry the rest of the proof through. %
The framework to express the appropriate normality conditions is given
to us by the polynomially cyclic algebras of~\cite{Mihailescu2010825}.

\begin{definition}
  \label{def:poly-cyclic}
  Let $E/\F_q$ be an elliptic curve of $j$-invariant not $0$ or
  $1728$. Let $\ell>3$ be an Elkies prime for $E$, let $\lambda$ be an
  eigenvalue of $\pi$, and let $\langle P\rangle$ be the corresponding
  eigenspace.

  We define the \emph{kernel polynomial} of $\lambda$ as
  \begin{equation*}
    f_\lambda(X) = \prod_{i\in (\Z/\ell\Z)^\ast/\{\pm 1\}}(X-x([i]P)),
  \end{equation*}
  and the \emph{algebra associated to $\lambda$} as
  $A_\lambda = \F_q[X]/f_\lambda(X)$.
\end{definition}

\begin{proposition}
  The algebra $A_\lambda$ is \emph{polynomially cyclic} in the sense
  of~\cite{Mihailescu2010825}: the automorphism group of $A_\lambda$
  over $\F_q$ is cyclic; if $\nu$ is a generator of the automorphism
  group, there exists a polynomial $C\in\F_q[X]$ such that for any
  $a(X)\in A_\lambda$
  \begin{equation*}
    \nu(a) = a(C(X)) \mod f_\lambda(X).
  \end{equation*}
\end{proposition}
\begin{proof}
  Let $c$ be a generator of $(\Z/\ell\Z)^{\ast}/\{\pm 1\}$.  The
  multiplication-by-$c$ map on $E$ cyclically permutes the generators
  of $\langle P\rangle/\{\pm 1\}$, thus it induces a cyclic
  automorphism $\nu$ of $A_\lambda$. %
  Multiplication by $c$ is defined by a rational function with no
  poles in $\langle P\rangle$, hence its action on
  $\langle P\rangle/\{\pm 1\}$ can be expressed by a polynomial $C$
  modulo $f_\lambda$.
\end{proof}

We are now ready to state our main result on the normality of elliptic
periods.

\begin{lemma}
\label{prop:xnormal}
Let $E$, $\ell$, $\lambda$, $f_\lambda$ and $A_\lambda$ be as in
Definition~\ref{def:poly-cyclic}, and let $\nu$ be a generator of the
automorphism group of $A_\lambda$.
We say that an element $a\in A_\lambda$ is \emph{normal} if the
$\nu^i(a)$ for $0\le i<(\ell-1)/2$ are linearly independent.
We say that $f_\lambda$ is normal if the image of
$X$ in $A_\lambda$ is normal.

Suppose that $(\Z/\ell\Z)^\ast=\langle\lambda\rangle\times S$. %
If $f_\lambda$ is normal, then the elliptic period
$\eta_{\lambda,S}(P)$ is a normal element of $\F_q(x(P))$, for any
point $P$ in the eigenspace of $\lambda$.
\end{lemma}
\begin{proof}
If $\lambda$ is of order $r$ modulo $q$, then $f_\lambda$ splits as
\[
f_\lambda = h_1 \cdots h_d
\]
where the $h_i$ are pairwise coprime irreducible polynomials of degree $r$,
and $d r = (\ell-1)/2$.
Therefore the algebra $A_\lambda$ splits as $d$ different copies of $\F_{q^r}$:
\[
A_\lambda \simeq \F_q[X]/h_1(X) \times \cdots \times \F_q[X]/h_d(X) \; .
\]
The $h_i$ correspond to the different Galois orbits of
$\langle P\rangle/\{\pm 1\}$.  Up to permutation of indices, we can
suppose that $\nu$ sends roots of $h_i$ onto roots of $h_{i+1}$.
We denote by $\sigma$ the automorphism of algebras $\nu^{r}$
and by $S$ the corresponding polynomial
(there should be no confusion with the subgroup $S$).
Note that $\nu^d$ is the Frobenius automorphism of $A_\lambda$.

Gathering the above remarks we obtain the following diagram:
\begin{center}
\begin{tikzpicture}
\node (A) {$A_\lambda=\F_q[X]/(f_\lambda(X))$};
\node[below=.3em of A] (Ax) {$x$};
\node[right=8em of A] (CRT) {$\F_q[X]/(h_1(X)) \times \F_q[X]/(h_2(X)) \times \cdots \times \F_q[X]/(h_d(X))$};
\node[below=.3em of CRT] (CRTx) {$(x, x, \ldots, x)$};
\node[below=6em of A] (T) {$A_\lambda^S \simeq \F_{q^r}$};
\node[below=.3em of T] (Tx) {$\mathrm{Tr}(x) = \sum_{i = 0}^{d-1} S^{(i)}(x)$};
\node[below=6em of CRT] (C) {$\F_q[X]/(h_1(X)) \times \F_q[X]/(h_1(X)) \times \cdots \times \F_q[X]/(h_1(X))$};
\node[below=.3em of C] (Cx) {$(x, S(x), \ldots, S^{(d-1)}(x))$};
\draw[->, shorten >= .5em, shorten <= .5em] (A) -- (CRT);
\draw[->, shorten >= .5em, shorten <= .5em] (Ax) -- (T);
\draw[->, shorten >= .5em, shorten <= .5em] (CRTx) -- (C);
\end{tikzpicture}
\end{center}

Let $x$ be the image of $X$ in $A_\lambda$.
Suppose that $x$ is normal in $A_\lambda$, then so is its trace in $A_\lambda^S$.
But taking the trace of $x$ in $A_\lambda$ and going through the two isomorphisms
of the diagram is the same thing as computing the period in the
$d$ copies of $\F_{q^r}$:
\[
x \mapsto \sum_{i=0}^{d-1} S^{(i)}(x) \mapsto
\left( \sum_{i=0}^{d-1} S^{(i)}(x) \pmod{h_1(x)}, \ldots,
\sum_{i=0}^{d-1} S^{(i)}(x) \pmod{h_d(x)} \right) \; .
\]
If a linear relation existed between $\eta_{\lambda,S}(P) \in \F_q(x(P))$
and its Galois conjugates for an eigenpoint $P$ of $\lambda$,
then it would be verified in the $d$ copies of $\F_{q^r}$
and it would lift up to $A_\lambda$,
thus forcing the same linear relation on
$\mathrm{Tr}(x)$ and its Galois conjugates.
Therefore, if $x$ is normal in $A_\lambda$,
then $\eta_{\lambda,S}(P)$ is normal in $\F_q(x(P))$.
\end{proof}

The converse is easily shown not to be true by experimentation.
There is also no relationship between the normality of $f_\lambda$
and that of $x(P) \in \F_q(x(P))$ for any eigenpoint $P$ of $\lambda$,
nor between the normality of $x(P)$ and that of $\eta_{\lambda,S}(P)$.

We conclude this section with a proof that normal elements are
abundant in the algebra $A_\lambda$. %
This can be used as a further heuristic argument that
Algorithm~\ref{algorithm:compell} fails with very low probability.

\begin{proposition}
\label{prop:euleralgebra}
Let $A$ be a polynomially cyclic algebra of degree $dr$
over the finite field $\F_q$.
The number of normal elements in $A$ is the same as the
number of normal elements in $\F_{q^{dr}}$:
$\Phi_q(x^{dr}-1) \in O(q^{dr})$ where $\Phi_q$ is the Euler function
for polynomials~\cite[Lemma~3.69]{lidl+niederreiter:2}.
\end{proposition}
\begin{proof}
There is at least one normal element in $A$~\cite[Theorem 4]{Mihailescu2010825}
(and it can be constructed from one in $\F_q$).
One such element $\alpha$ can be used to count the number of normal elements
as in the finite field case~\cite[Theorem~3.73]{lidl+niederreiter:2}:
\begin{enumerate}
\item every element $a$ can be written using the normal basis
associated to $\alpha$: $a = \sum_{i=0}^{dr-1} a_i \nu^{(i)}(\alpha)$;
\item as the Galois group of $A$ is cyclic, the decomposition
of the Galois conjugates of $a$ is obtained by shifting the coefficients $a_i$;
\item the circulant matrix $M_a$ associated to the $a_i$'s is invertible
if and only if $a$ is a normal element;
\item the polynomial $\sum_{i=0}^{dr-1} a_i X^i$ in $\F_q[X] / (X^{dr}-1)$
is invertible if and only if $M_a$ is.
\end{enumerate}
Therefore normal elements in $A$ are in one-to-one correspondence
with units of $\F_q[X] / (X^{dr}-1)$.
\end{proof}

\subsection{Experimental evidence for the conjecture}

We now present experimental evidence supporting the validity of
Conjecture~\ref{conj:ellperiods}. %
Assuming elliptic periods behave as random elements, the chance that
one of them does not generate $\F_{q^r}$ decreases as $q^{-r}$
asymptotically. %
Thus we may focus on looking for counterexamples with small values of
$q$ and $r$, as that increases the probability of finding one. %
We now present the various strategies by which we tried to obtain
counterexamples; surprisingly, none of them produced any.

\paragraph{Search strategy 1}
For small values of $q$,
one can generate random elliptic curves over $\F_q$,
look for small values of $\ell$ which are Elkies primes,
compute a point of $\ell$-torsion defined over $\F_{q^r}$,
the associated elliptic period and its minimal polynomial over $\F_q$.

If one restricts to prime values of $r$ and supposes elements
of $\F_{q^r}$ are represented as polynomials with coefficients in $\F_q$,
the computation of the minimal polynomial can be avoided:
one simply checks that the elliptic period is not in $\F_q$,
i.e. a constant polynomial.

\paragraph{Search strategy 2}
With further restrictions on $r$ and $\ell$ it is possible to avoid
the computation of an actual $\ell$-torsion point as well.
For example, let $r$ be a prime such that $\ell = 2 d r + 1$
with $d = 2$ is prime
and suppose that $\ell$ is an Elkies prime for $E$:
the $\ell$-division polynomial $f_\ell$ of $E$ has two factors
$h_1$ and $h_2$ of degree $r$ over $\F_{q}$
(whose product is $f_\lambda$ the kernel polynomial of an $\ell$-isogeny).
The elliptic period is defined as $x(P) + x([i] P)$ for some
$i \in \left(Z / \ell \Z\right)^*$ such that $i^2 \equiv -1 \pmod{\ell}$
and any $P \in E(\F_{q^r})[\ell]$.
Let $m_i$ be multiplication by $i$ map on the curve.
Then the mapping $x + m_i(x)$ sends the roots of $h_1$
and the roots of $h_2$ onto the elliptic periods.
If we suppose that the period is not generating $\F_{q}$, which implies that
$\Res_x(h_1(x),y-x-m_i(x)) = \Res_x(h_2,y-x-m_i(x))$
splits over $\F_{q}$,
then by Galois invariance $x + m_i(x) \equiv t \pmod{h_1}$
for a constant $t \in \F_{q}$,
and since the periods do not depend on the initial choice of $h_1$
or $h_2$,
we conclude that $m_i(x) \equiv (t-x) \pmod{f_\lambda}$.
Hence, we can simply compute $x + m_i(x) \pmod{f_\lambda}$ (or $\pmod{h_1}$
or $\pmod{h_2}$) and see if it is a constant.
The dominating step in this approach is the factorization of $f_\ell$.
Such an approach can easily be generalized to larger values of $d$.

\paragraph{Search strategy 3}
Finally, one can apply the above approach in a global way:
pick a rational elliptic curve with a rational $\ell$-isogeny
to ensure that $\ell$ is an Elkies prime over finite fields,
and such that the $\ell$-torsion (or its abscissas) is defined
over an extension of degree $r$ (or $2r$) with $\ell = 2 d r + 1$,
then compute the kernel polynomial of the isogeny
$f_\lambda = h_1 \ldots h_d$ with $h_1, \ldots, h_d$ of degree $r$
and check that there is no prime, outside those with discriminant $D=0,3,4$,
dividing all the coefficients, except possibly the constant one, of
\[
x + m_i(x) + \cdots + (\underbrace{m_i \circ \cdots \circ m_i}_\text{$d-1$ times}) (x) \pmod{f_\lambda}
\]
where $i$ a $d$-th root of $-1$ in $\left(Z / \ell \Z\right)^*$.
If such a prime $q$ is found, the reduction modulo $q$ is a
candidate for a counterexample.
It should still be checked that over $\F_q$ the kernel polynomial
does not split into more factors than over $\Q$,
and the corresponding counterexample should be explicitly computed.

For $r = 3$, $d = 2$ and $\ell = 2 d r + 1 = 13$, there exists an
infinite family of elliptic curves suitable for this approach.
Indeed, Daniels et al.~\cite{daniels_torsion_2015} showed that
all rational elliptic curves with $13$-torsion defined over a cubic extension
of $\Q$ belong to the parametrized family with $j$-invariant:
\[
j(t)=\frac{\left(t^4-t^3+5 t^2 + t + 1\right) \left(t^8 - 5 t^7 + 7 t^6 - 5 t^5 + 5 t^3 + 7 t^2 + 5 t + 1\right)^3}{t^{13} \left(t^2 -3 t - 1\right)}.
\]
In practice, the limiting step when using this family is computing
the $j$-invariant of the curves because its height quickly grows.

The next candidate for prime $r$ and $\ell$, and $d = 2$, can not be used.
Indeed, for $r = 7$ and $\ell = 29$, there is no rational elliptic curve
with a rational $29$-isogeny~\cite{Mazur1974} (or elliptic curves defined
over a septic number field with $29$-torsion~\cite{Derickx201452}).
More generally, rational elliptic curves with rational $\ell$-isogenies
exist only for a finite set of values of $\ell$ corresponding
to non-cuspidal rational points on $X_0(\ell)$, leaving few chances
to find other interesting curves.

Among the suitable prime values for $\ell \geq 13$~\cite{Mazur1974},
only $\ell = 13$ provides an infinite family~\cite{Lozano-Robledo2013}
parametrized by
\[
j(t) = \frac{\left(t^2 + 5t + 13\right)\left(t^4 + 7t^3 + 20t^2 + 19t + 1\right)^3}{t} \enspace .
\]
Within this family are the aforementioned curves with $13$-torsion
defined over a cubic extension of $\Q$.
The other curves might only have the abscissas of the $13$-torsion points
defined over a cubic extension of $\Q$,
which is completely equivalent for this approach,
or over a sextic extension of $\Q$,
in which case the kernel polynomials are irreducible of degree $6$.
These curves can still be used with $r = 3$ and $d = 2$,
but if a candidate counterexample is found, it should be checked
that the kernel polynomial splits as two cubic factors over $\F_q$.
Indeed, elliptic periods are not normal, and the trace trick used
in the cyclotomic Rains algorithm with an auxiliary extension degree
$s = 2$ does not apply.

\subparagraph{Prime values of $\ell$}
We now treat the other prime values for $\ell \geq 13$
for which only a finite number of $\overline{\Q}$-isomorphism
classes of rational elliptic curves come with a rational $\ell$-isogeny.
Curves labels come from Cremona's tables~\cite{cremona_tables}.

For $\ell = 17$, with $(\ell-1)/2 = 2^3$,
there are two rational elliptic curves
with a rational $17$-isogeny.
The factorization of $(\ell-1)/2$ implies
that $2$ would divide both $r$ and $d$
except in trivial cases where $r = 2^3$ and $d = 1$.
Therefore, none of these curves can be used.

For $\ell = 19$,  with $(\ell-1)/2 = 3^2$,
there are two rational elliptic curves
with a rational $19$-isogeny.
The factorization of $(\ell-1)/2$ implies
that $3$ would divide both $r$ and $d$
except in trivial cases where $r = 3^2$ and $d = 1$.
Therefore, none of these curves can be used.

For $\ell = 37$, with $(\ell-1)/2 = 2\cdot3^2$,
there are two rational elliptic curves
with a rational $37$-isogeny.
The first curve is curve~\texttt{1225H1}
and its kernel polynomial splits into three sextic factors.
It implies that $3$ must divide $d$ and can not divide $r$.
Therefore, the curve can only be used with $r = 2$ and $d = 3^2$.
The second curve is curve~\texttt{1225H2}
and its kernel polynomial is irreducible.
Therefore, it can be used with $r = 2$ and $d = 3^2$,
and $r = 3^2$ and $d = 2$.

For $\ell = 43$, with $(\ell-1)/2 = 3\cdot7$
there is one rational elliptic curve
with a rational $43$-isogeny 
and its kernel polynomial is irreducible.
Therefore, it can be used with $r = 3$ and $d = 7$,
and $r = 7$ and $d = 3$.

For $\ell = 67$, with $(\ell-1)/2 = 3\cdot11$
there is one rational elliptic curve
with a rational $67$-isogeny 
and its kernel polynomial is irreducible.
Therefore, it can be used with $r = 3$ and $d = 11$,
and $r = 11$ and $d = 3$.

For $\ell = 163$, with $(\ell-1)/2 = 3^4$,
there is one rational elliptic curve
with a rational $163$-isogeny.
The factorization of $(\ell-1)/2$ implies
that $3$ would divide both $r$ and $d$
except in trivial cases where $r = 3^4$ and $d = 1$.
Therefore, this curve can not be used.

\subparagraph{Composite values of $\ell$}
There are also a few odd composite values $\ell \geq 13$
for which rational elliptic curves with a rational $\ell$-isogeny exist.
In this case, only points of exact order $\ell$ should be used.
If moreover $\ell$ is not squarefree, periods should be computed
using the more general formula from~\cite{feisel1999normal}.


For $\ell = 15$, with $\phi(15)/2 = 2^2$,
there are four rational elliptic curves with a rational $15$-isogeny.
The factorization of $(\ell-1)/2$ implies
that $2$ would divide both $r$ and $d$
except in trivial cases where $r = 2^2$ and $d = 1$.
Therefore, none of these curves can be used.


For $\ell = 21$, with $\phi(21)/2 = 2\cdot3$,
there are four rational elliptic curves with a rational $21$-isogeny
All of them have a rational point of $3$-torsion,
and among them, curve~\texttt{162B1} is the unique rational elliptic curve
with $21$-torsion over a cubic extension of $\Q$~\cite{najman_cubic}.
For curve~\texttt{162B1},
the part of the kernel polynomial corresponding to primitive $21$-torsion
splits into two cubic factors, whereas,
for the three other curves~\texttt{162B2}, \texttt{162B3} and\texttt{162B4},
this part is a sextic irreducible.
Therefore, the former curve can only be used with $r = 3$ and $d = 2$,
whereas the latter ones can be used both with $r = 2$ and $d = 3$,
and $r = 3$ and $d = 2$.

For $\ell = 25$, with $\phi(25)/2 = 2\cdot5$,
there is an infinity of rational elliptic curves
with a rational $25$-isogeny parametrized~\cite{Lozano-Robledo2013} by
\[
j(t) =
\frac{\left(t^{10}+10t^8+35t^6-12t^5+50t^4-60t^3+25t^2-60t+16\right)^3}
{t^5+5t^3+5t-11}\enspace .
\]
There is also an infinity of elliptic curves defined over a quintic
extension $K$ of $\Q$ with $25$-torsion defined over
$K$~\cite{Derickx201452,2016arXiv160807549D}.
If one of these curves is actually defined over $\Q$,
it has rational $5$-torsion~\cite{gonzalez-jimenez_complete_2016}.
Proving that there is an infinity of such curves or
parametrizing them is outside the scope of this work,
but one can still sample curves within the larger family,
compute the kernel polynomial over the rationals,
and compute periods for $r = 2$ and $d = 5$, or $r = 5$ and $d = 2$,
if it is irreducible,
and $r = 5$ and $d = 2$ if it splits into two quintic factors,
and finally check that the kernel polynomials of the candidate counterexamples
factors as expected over $\F_q$.

For $\ell = 27$, with $\phi(27)/2 = 3^2$,
there is one rational elliptic curve with a rational $27$-isogeny.
The factorization of $(\ell-1)/2$ implies
that $3$ would divide both $r$ and $d$
except in trivial cases where $r = 3^2$ and $d = 1$.
Therefore, this curve can not be used.

\paragraph{Experimental results}
All above strategies were implemented and ran on different ranges of parameters.

For strategy 1, for prime values of $\ell$ in different ranges,
we tested all elliptic curves defined over $\F_q$
with $q$ prime up to some bounds
and such that $\ell$ is an Elkies prime with an eigenvalue
of odd prime order $r$ (without any bound on $r$).
The greatest values of $q$ and the total number of curves tested
are given in Table~\ref{table:strat1}.
Among all ranges we tested more than 43 millions of curves.
\begin{table}[!ht]
\[
\begin{array}{|*{6}{c|}}
\hline
\ell & [0, 100] & [100, 200] & [200, 300] & [300, 400] & [400, 500] \\
\hline
q_{\max} & 10729 & 8147 & 5807 & 4919 & 6427 \\
\hline
\text{nr. of curves} & 18321275 & 9371913 & 4567947 & 2936870 & 3587951 \\
\hline
\hline
\ell & [500, 600] & [600, 700] & [700, 800] & [800, 900] & [900, 1001] \\
\hline
q_{\max} & 3307 & 2939 & 1993 & 3643 & 4051 \\
\hline
\text{nr. of curves} & 1062818 & 1061413 & 450676 & 1194598 & 1284669 \\
\hline
\end{array}
\]
\caption{Largest value of $q$ tested for a given range of $\ell$ with strategy 1.}
\label{table:strat1}
\end{table}

For strategy 2, we picked up small values of $r$ for which
$\ell = 2 \cdot d \cdot r + 1$ with $d = 2$ is prime,
and checked all elliptic curves defined over $\F_q$ with $q$ prime up
to some bound.
The greatest values of $q$ tested are given in Table~\ref{table:strat2}.
\begin{table}[!ht]
\[
\begin{array}{|*{6}{c|}}
\hline
r & 3 & 7 & 13 & 37 & 43 \\
\hline
q_{\max} & 78593 & 29531 & 10993 & 1987 & 1291 \\
\hline
\hline
r & 67 & 73 & 79 & 97 & 127 \\
\hline
q_{\max} & 691 & 743 & 577 & 419 & 271 \\
\hline
\end{array}
\]
\caption{Largest value of $q$ tested for a given $r$ with strategy 2.}
\label{table:strat2}
\end{table}

For strategy 3, we tested all rational elliptic curves
with a rational $13$-isogeny and $13$-torsion defined over a cubic extension
in the infinite family parametrized by
\[
j(t) = \frac{\left(t^2 + 5t + 13\right)\left(t^4 + 7t^3 + 20t^2 + 19t + 1\right)^3}{t}
\]
for rational values of $t$ with numerator and denominator of absolute values
less than $1538$.
We also tested the sporadic curves for prime values of $\ell = 37, 43, 67$.

Among all these computations, no counterexample was found.

\paragraph{Final remark}
As to the significance of our experiments, it is worth noting that a
natural generalization of periods consists in replacing the sum by a
product (or any other symmetric function, indeed). %
Globally, this corresponds to replace a trace by a norm in some number
field, or polynomially cyclic algebra. %
This idea is bound to fail for Gaussian periods, because roots of
unity are conjugate to their inverses; however it looks in principle
as good as our definition in the elliptic case.

Looking more closely, however, we notice that using traces to define
periods grants some very special properties to them. %
For one, there is no analogue of Proposition~\ref{prop:xnormal} for
symmetric functions other than the trace; and, most strikingly, we
were able to find examples of \emph{generalized periods} not
generating their field of definition for most elementary symmetric
functions, except for the trace, using only small scale experiments.

In conclusion, even though we have little theoretical evidence to
support Conjecture~\ref{conj:ellperiods}, disproving it may turn out
to be a very challenging task.


\bibliographystyle{plain}
\bibliography{refs}

\end{document}